%% file: main.tex
\documentclass[11pt]{article} 

\pdfoutput=1

\usepackage{etoolbox}
\newtoggle{arxiv}
\toggletrue{arxiv}

\newlength\tindent
\usepackage{hyperref,etoolbox}

\include{def}
\usepackage{enumitem}
\usepackage{srcltx}
\usepackage{graphicx}
\usepackage{microtype}
\usepackage{graphicx}
\usepackage{booktabs} %

\usepackage{amsmath,amssymb,amsfonts,amsthm, thmtools}
\usepackage{mathtools}
\usepackage{dsfont}
\usepackage{bm}
\usepackage{nicefrac}
\usepackage{algorithm,algorithmic}
\usepackage[linesnumbered, ruled,vlined,algo2e]{algorithm2e}
\usepackage{xcolor}
\usepackage[round,sort&compress]{natbib}

\usepackage{bbm}
\DeclareMathOperator{\supp}{\text{Supp}}

\newcommand{\vep}{\varepsilon}
\newcommand{\la}{\langle}
\newcommand{\ra}{\rangle}

\theoremstyle{definition}
\newtheorem{lemma}{Lemma}
\newtheorem{theorem}{Theorem}

\newtheorem{corollary}{Corollary}
\newtheorem{proposition}{Proposition}

\newtheorem{definition}{Definition}

\let\cite\citep
\setlist[itemize]{itemsep=1pt,topsep=0pt}

\newcommand{\simplex}{
  \mathchoice
    {\includegraphics[height=1.4ex, trim= 0pt 30pt 0pt 0pt]{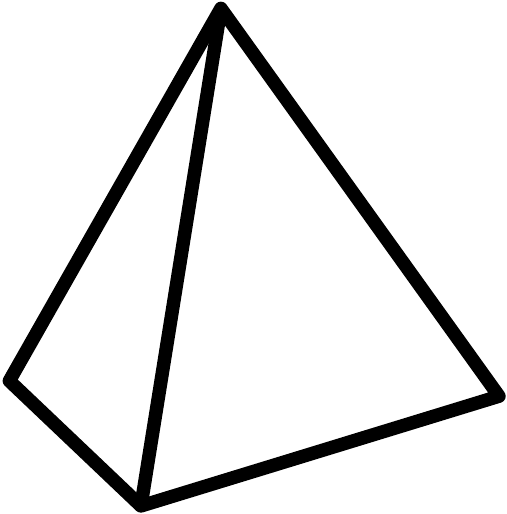}} 
    {\includegraphics[height=1.4ex, trim= 0pt 25pt 0pt 0pt]{simplex.pdf}} 
    {\includegraphics[height=1.1ex, trim=0pt 30pt 0pt 0pt]{simplex.pdf}} 
    {\includegraphics[height=.8ex, trim=0pt 30pt 0pt 0pt]{simplex.pdf}} 
}

\usepackage[margin=1in]{geometry}
\title{Instance-dependent Sample Complexity Bounds for Zero-sum Matrix Games}
\author{Arnab Maiti\thanks{University of Washington, WA. arnabm2@uw.edu } \and
Kevin Jamieson\thanks{University of Washington, WA. jamieson@cs.washington.edu }
\and
Lillian J. Ratliff\thanks{University of Washington, WA. ratliffl@uw.edu} 
}
\date{}
\begin{document}

\maketitle

\begin{abstract}
We study the sample complexity of identifying an approximate equilibrium for two-player zero-sum $n\times 2$ matrix games. 
That is, in a sequence of repeated game plays, how many rounds must the two players play before reaching an approximate equilibrium (e.g., Nash)?
We derive instance-dependent bounds that define an ordering over game matrices that captures the intuition that the dynamics of some games converge faster than others. 
Specifically, we consider a stochastic observation model such that when the two players choose actions $i$ and $j$, respectively, they both observe each other's played actions and a stochastic observation $X_{ij}$ such that $\E{ X_{ij}} = A_{ij}$. 
To our knowledge, our work is the first case of instance-dependent lower bounds on the number of rounds the players must play before reaching an approximate equilibrium in the sense that the number of rounds depends on the specific properties of the game matrix $A$ as well as the desired accuracy.
We also prove a converse statement: there exist player strategies that achieve this lower bound. 
\end{abstract}
\input{introduction}
\input{results}
\input{discussion}

\input{ack.tex}
\bibliographystyle{plainnat}
\bibliography{refs.bib}

\newpage
\onecolumn
\appendix
\input{appendix}


\end{document}

%% file: def.tex
\newcounter{assumption}%
\renewcommand{\theassumption}{\arabic{assumption}}

\newcommand*\R[0]{\mathbb{R}}

\newcommand*\E[1]{\mathbb{E}\left[#1\right]}

\newcommand*\bmat[1]{\begin{bmatrix} #1 \end{bmatrix}}

\renewcommand*{\P}{\mathbb{P}}

\newcommand{\g}{g}

\usepackage{multirow}
\usepackage{caption}
\usepackage{subcaption}

   \usepackage{xcolor}

%% file: introduction.tex
\section{Introduction}
\label{sec:introduction}
In single player stochastic games like multi-armed bandits and reinforcement learning, instance dependent or ``gap dependent'' sample complexity bounds that characterize the number of interactions with the environment to identify a good policy are well-understood.
In contrast to  minimax or worst-case sample complexity guarantees, these bounds and the algorithms that obtain them adapt to the true difficulty of the problem and are provably better when the problem is easy.
However, very little progress has been made on multiplayer settings. 
Even the simplest of such settings---i.e. two--player normal form matrix games---have only been studied in a minimax, worst-case sense to our knowledge. 
Nonetheless many practical applications are such that  the outcome for a decision--maker depends not just on their own action, but on the actions of other decision--makers in the environment. Indeed, finite normal form games represent a reasonable abstraction for a multitude of different important problems from economic decisions to voting systems to auctions to military abstractions (see, e.g., \citet{nisan2007algorithmic,von2007theory,bacsar1998dynamic} and references therein).

To concretize ideas, consider a setting in which two firms produce bids for a sequence of arriving customers. Each firm has one of two ways of preparing the bid (e.g., use a higher quality product versus lower quality product but include a warranty). Customers are drawn iid from a population, and select a firm meaning that the selected firm ``wins" the bid, while the other firm ``loses" the bid. This setting can be abstracted as a repeated  two--player zero--sum game defined by a $2\times 2$ stochastic matrix with independent entries in $\{-1,0,1\}$ and 
with expectation $A \in [-1,1]^{2\times2}$. For instance, the entries of $A$ may be $A_{11}=A_{22}=0$, $A_{12}=5/6$, and $A_{21}=-2/3$. 
Such abstractions arise in many applications including online platforms and other digital marketplaces where firms are competing for the same consumer demand.

With this motivation in mind, in this paper we consider  two--player, zero--sum normal form matrix games possessing a unique Nash equilibrium which are defined by a stochastic matrix of dimension $n\times 2$ such that $2\leq n<\infty$.
For this class of games, we characterize the instance dependent sample complexity of identifying  a joint mixed strategy that approximately achieves the value of the game, and  a joint mixed strategy from which players have no incentive to deviate in an approximate sense. 
That is, in a sequence of repeated game plays, we address the following questions:  how many rounds must the two players play before reaching a $(i)$   $\vep$--good solution, or $(ii)$  $\varepsilon$--Nash equilibrium, respectively? 

The repeated play proceeds as follows: for a fixed matrix $A \in \mathbb{R}^{n\times 2}$ with entries $A_{ij}$, at the start of each round $t$, the first and second player choose an $i \in \{1,\dots,n\}$ and $j \in \{1,2\}$ simultaneously, respectively, observe each others chosen actions, and then both simultaneously observe outcome $X_{ij}$ where $\mathbb{E}[X_{ij}]=A_{ij}$ and is $1$-sub-Gaussian (e.g., $X_{ij} \sim \text{Bernoulli}(A_{ij}) \in \{-1,1\}$ representing firm 1 winning the bid or not). 
Hence, the first and second player receive expected rewards of $A_{ij}$ and $-A_{ij}$, respectively. Throughout we refer to this zero-sum stochastic matrix game by simply referencing the matrix $A$ that induces the game.

Letting $\simplex_m$ denote the $m$--dimensional simplex,
we analyze the following two objectives: find a joint mixed strategy $(x,y)\in \simplex_n\times\simplex_2$ such that
\begin{enumerate}[label={\it (\roman*)},itemsep=2pt,topsep=2pt]
    \item $|V_A^\star-\la x,Ay\ra|\leq \vep$ where  \[V_A^\star:=\max_{x \in \simplex_n} \min_{y\in\simplex_2} x^\top A y\] is the value of the game, and
    \item both \[\la x,Ay\ra\geq \la x',Ay\ra-\vep\quad\text{and}\quad\la x,Ay'\ra \geq \la x,Ay\ra -\vep\] hold for all $(x',y')\in \simplex_n\times \simplex_2$.
\end{enumerate}
The former is precisely an \textbf{$\boldsymbol{\varepsilon}$--good solution}, and the latter an \textbf{$\boldsymbol{\varepsilon}$--Nash equilibrium}.

We characterize the instance-dependent sample complexity of identifying $\varepsilon$-approximate solutions for the above problems in the sense that they scale with not just $\varepsilon$ and the number of actions, but the particular properties of the matrix $A$. 
Thus, our characterization defines an ordering over games capturing the intuition that the dynamics of some games converge must faster than others.
Specifically, we prove lower bounds on the number of rounds necessary for any two players to converge to an approximate Nash equilibrium.
Moreover, we propose strategies for the two players that achieve this sample complexity.
Our instance-dependent sample complexities introduce a number of quantities that characterize  notions of the sub-optimality ``gap,'' and we discuss why it is non-trivial to extend these definitions and our analysis to the general $n\times m$ dimensional matrix games. 

Before we state our main contributions, we state some easily proven facts to contextualize our results (see Appendix \ref{appendix:properties:matrix} for proof).
\begin{proposition}
For any zero-sum matrix game $A$, an $\varepsilon$--Nash equilibrium  is also an $\varepsilon$--good solution.
\end{proposition}
This means any lower bound on identifying an $\varepsilon$--good solution is also a lower bound on identifying an $\varepsilon$--Nash equilibrium. 
Conversely, any algorithm that can identify an $\varepsilon$--Nash equilibrium can also identify an $\varepsilon$--good solution with the same sample complexity.
\begin{lemma}\label{lem:trivial:NE}
Fix any $\varepsilon > 0$ and $\delta \in (0,1)$, and  matrix $A \in \R^{n\times m}$. Suppose 
that $\bar{A}\in \R^{n\times m}$ has entries $\bar{A}_{ij}$ that are the empirical mean of $\frac{8 \log (2mn/\delta) }{\varepsilon^2}$ 1-sub-Gaussian observations resulting from players playing $(i,j)$, and such that $\bar{A}_{ij}$ has  expectation  $A_{ij}$. 
Let $(x,y)\in\simplex_n\times\simplex_m$ be the Nash equilibrium of the game defined by $\bar{A}$.
With probability at least $1-\delta$, the mixed strategy $(x,y)$ is an $\varepsilon$--Nash equilibrium of $A$.
\end{lemma}
The above strategy is minimax optimal: there exists a worst-case game matrix $A$ such that identifying an $\varepsilon$--Nash equilibrium or $\varepsilon$--good solution with constant probability requires at least $1/\varepsilon^2$ samples.
This worst-case result suggests that the sample compelxity of identifying an $\varepsilon$--Nash equilibrium and $\varepsilon$--good solution are about the same, and that this sample complexity scales with $\varepsilon$.
Remarkably, we will show that both these conclusions are false: there is a provable separation between these two problems, and that the sample complexities for natural problems can be as small as $1/\varepsilon$.



\subsection{Contributions}
%
    Consider a game defined by a fixed $2\times 2$ matrix $A$ which has a unique Nash equilibrium which is not a pure-strategy Nash equilibrium. 
In Theorems \ref{thm:lower1} and \ref{thm:lower2}, we show under some mild assumptions that to find an $\varepsilon$-good solution for the matrix game $A$ with probability at least $1-\delta$, we require at least \[\Omega \left(\min\left\{\frac{1}{\varepsilon^2},\max\left\{\frac{1}{\Delta_{\min}^2},\frac{1}{\varepsilon |D|}\right\}\right\}\log(1/\delta)\right)\] samples from the matrix $A$, where problem-dependent parameters $D$ and $\Delta_{\min}$ are functions of $A$ alone and defined in Section~\ref{sec:epsilon_good}.  
Complementing this result, we present an algorithm (Algorithm \ref{alg-ucb-1}) that, with probability $1-\delta$, identifies an $\varepsilon$-good solution using a number of samples  matching this lower bound up to logarithmic factors.

In the same setting,  we show (Theorem \ref{thm:lower3}) that identifying an $\varepsilon$--Nash equilibrium for the game defined by $A$ with probability at least $1-\delta$ requires at least \[\Omega \left(\frac{\Delta_{m_2}^2}{\varepsilon^2D^2}\log(1/\delta)\right)\] samples where $\Delta_{m_2}$, defined in Section~\ref{sec:epsilon_Nash}, is function of $A$ alone. 
Since a lower bound on $\vep$--good solution identification immediately implies a lower bound on identifying an $\varepsilon$--Nash equilibrium, as noted above, we conclude that the sample complexity of identifying an $\varepsilon$--Nash equilibrium with probability at least $1-\delta$ requires \[\Omega\left(\min\left\{\frac{1}{\varepsilon^2},\max\left\{\frac{1}{\Delta_{\min}^2},\frac{\Delta_{m_2}^2}{\varepsilon^2 D^2}\right\}\right\}\log(1/\delta)\right)\] samples.
In general, it is the case that \[\max\left\{\frac{1}{\Delta_{\min}^2},\frac{\Delta_{m_2}^2}{\varepsilon^2 D^2}\right\} \geq \max\left\{\frac{1}{\Delta_{\min}^2},\frac{1}{\varepsilon |D|}\right\}\] which demonstrates a separation in sample complexity between identifying an $\varepsilon$--good solution and $\varepsilon$--Nash equilibrium.
Again, we complement this lower bound result by designing an algorithm (Algorithm \ref{alg-ucb-2}) that, with probability $1-\delta$, identifies an $\varepsilon$--Nash equilibrium with a sample complexity matching this lower bound up to logarithmic factors. 

On the other hand, if the game does have a pure-strategy Nash equilibrium then we prove nearly optimal instance dependent upper bounds for identifying an $\varepsilon$-good solution or $\varepsilon$-Nash that are similar to multi-armed bandits. 
In summary, our results completely characterize the instance-dependent sample complexity of identifying an $\varepsilon$-good solution and $\varepsilon$-Nash in the $2 \times 2$ case. 

Now consider a game defined by a fixed $3\times 2$ matrix $A$ that has a unique Nash equilibrium which is not a pure-strategy Nash equilibrium. In Theorem \ref{thm:lower4}, we show under some mild assumptions that to find an $\varepsilon$-good solution for the matrix game defined by $A$ with probability at least $1-\delta$, we require at least $\Omega \left(\frac{1}{\Delta_g^2}\log(1/\delta)\right)$ samples. 
In fact, this number of samples, characterized by the problem-dependent constant $\Delta_g$, is required to just identify the support of the mixed strategy for player 1 in $\simplex_3$ which we show is necessary for $\varepsilon$-good identification.  
Now consider a game defined by an $n\times 2$ matrix $B$ for any $n \geq 3$ which has a unique Nash equilibrium $(x^*,y^*)$ that is not a pure-strategy equilibrium. 
We complement our lower bound result by designing an algorithm (Algorithm \ref{alg-ucb-3}) that, with probability at least $1-\delta$, samples each element of $B$ for \[O\left(\min\left\{\frac{1}{\varepsilon^2},\max\left\{\frac{1}{ \Delta_{\min}^2},\frac{1}{ \Delta_{g}^2}\right\}\right\}\log(1/\delta)\right)\] times (ignoring some logarithmic factors) and  either returns $\supp(x^*)$ and $\supp(y^*)$ or concludes that $\Delta_g$ is not sufficiently large compared to $\varepsilon$. If the support is successfully identified, the algorithms for the $2 \times 2$ cases can be applied. 
Otherwise, an $\varepsilon$--Nash equilibrium can be output after $O(\frac{n}{\varepsilon^2}\cdot \log(n/\delta))$ using the procedure of Lemma~\ref{lem:trivial:NE}. While these sample complexity results hint at necessary and sufficient conditions on the instance-dependent sample complexities of general $m \times n$ games, a full characterization of this setting is left for future work.

In their respective sections, we define these problem-dependent parameters and provide intuition for what they represent and how they arise, which itself provides some insight into the difficulty of the general $m\times n$ setting. 

Above we have highlighted our results in the special case of a unique Nash Equilibrium which is not a pure strategy. 
However, we address all other cases as well, they are simply more straightforward. 
For instance, if the Nash equilibrium is \emph{not} unique, then $\Delta_{\min}=0$ (and $\Delta_g=0$ for $n\times 2$ games) and therefore our upper and lower bounds match and correspond to a $1/\varepsilon^2$ rate. Moreover, our algorithms do not assume that the equilibrium is a pure or mixed strategy, or if it is unique or not. All cases are covered (see Theorems \ref{thm:alg1}, \ref{thm:alg2}, \ref{thm:alg3} and Lemma \ref{lem:trivial:NE}). Note that our lower bound results hold only for the mixed strategy case, and we do omit a lower bound for the pure strategy case. This was done because the mixed strategy case is the novel and challenging case while the pure strategy case is very similar to multi-armed bandit lower bounds (c.f., \citet{kaufmann2016complexity}).

\subsection{Related Work}
\label{sec:relatedwork}

\textbf{Complexity of Matrix Games.} Characterizing equilibrium behavior in normal form matrix games has been studied extensively in economics \citep{von1947theory,bohnenblust1950solutions}, as has learning as an abstraction for how players reach an equilibrium \cite{fudenberg1998theory}. The computational complexity of (exact) Nash equilibrium, especially in finite normal form games, is known to be PPAD-complete~\citep{daskalakis2009complexity,daskalakis2009ACM}. Given such hardness results, it is natural to reason about the computational complexity of \emph{approximate} equilibrium. For instance, it has been shown that $\vep$--approximate Nash can be computed in polynomial time~\cite{daskalakis2007progress,daskalakis2009note} where $\varepsilon$ is an absolute constant. These results primarily focus on settings of full information, and are concerned with computational complexity.  

Iteration complexity has been explored fairly extensively in partial information settings including settings with time-varying rewards and continuous action spaces; see, e.g., \cite{rakhlin2013optimization,cesa2006prediction,blum2007learning,syrgkanis2015fast,cardoso2019competing,bravo2018bandit,drusvyatskiy2021improved, daskalakis2011near} and references therein. Only recently has the focus shifted to characterizing statistical learnability---i.e., sample complexity---of equilibrium concepts, or other desiderata such as $\vep$--good solutions, in the presence of bandit feedback. For example, in the bandit feedback setting where players also observe the actions of their opponents,
\citet{donoghue2021matrix} show that  players adopting an optimism in the face of uncertainty principle when selecting actions experience sublinear regret---i.e., the short-fall in cumulative rewards relative to the value of the game---and further show that alternative strategies such as Thompson sampling cannot do not have a guarantee of sublinear regret.

\textbf{Instance Dependent Bounds for Games.}
To our knowledge, instance dependent sample complexity bounds remain under explored in games.  That being said, there are very recent results on special classes of games. For instance,  
\citet{dou2022gap} provide the first minimax bounds for the class of congestion games, which have the nice property of being equivalent to an optimization problem due to their potential game structure. Additionally, \citet{dou2022gap} provide sample complexity results for the centralized and decentralized problem settings under both semi-bandit and bandit feedback.
Similarly, \citet{cui2022learning} study the regret of the Nash Q-learning algorithm for two-player turn based Markov games, and introduce the first gap dependent logarithmic upper bounds, which match theoretical lower bounds up to log factors, in the episodic tabular setting.

\textbf{Instance Dependent Bounds in Stochastic Bandits.}
The sample complexity of stochastic bandits is well-understood: given $n$ actions each yielding a stochastic reward, to identify an action with a mean within $\varepsilon$ of the maximum with probability $1-\delta$, it is necessary and sufficient to take $\sum_{i=1}^n \min\{ \frac{1}{\varepsilon^2},\frac{1}{\Delta_i^2} \} \log(1/\delta)$ total samples, where $\Delta_i$ is the difference between the $i$th mean and the highest mean (up to $\log\log(1/\Delta_i)$ factors) \citep{mannor2004sample,kaufmann2016complexity,karnin2013almost}.
Stochastic bandits can be directly compared to our setting where $A$ is an $n\times 1$ matrix with the means of the arms on the rows.

%% file: results.tex
\section{{Results for $\varepsilon$--Good Solutions of $2\times 2$ Matrix Games}}\label{sec:epsilon_good}
 
This section is devoted to instance-dependent sample complexity bounds for identifying an $\varepsilon$-good solution $(x,y)$ for a zero-sum game matrix $A$. Recall that $|V_A^*-\langle x, Ay\rangle|\leq \varepsilon$. 
In what follows, we will frequently assume that the mixed strategies of the unique Nash equilibrium have \emph{full support}: the mixed strategy $x=\in \simplex_m$ is said to have a full support if $\supp(x):= \{ i \in [m]: x_i > 0\}$ is equal to $[m]:=\{1,\dots,m\}$. Here $x_i$ is the $i$-th component of $x$.
If $m=2$ then the unique equilibrium is either a full support mixed strategy or is a pure strategy, but not both.
\begin{definition}[Pure Strategy Nash Equilibrium]
An element $(i^*,j^*)$ is a Pure Strategy Nash Equilibrium (PSNE) of the game induced by the matrix $A\in \mathbb{R}^{m\times n}$ if $A_{i^*j^*}=\max_{i\in[m]}A_{ij^*}$ and $A_{i^*j^*}=\min_{j\in[n]}A_{i^*j}$. Moreover, a Nash equilibrium $(x,y)\in \simplex_m\times \simplex_n$ where $\supp(x)=\{i\}$ and $\supp(y)=\{j\}$ corresponds to a PSNE $(i,j)$.
\end{definition}

For a matrix $A = [a, b; c,d ]$ (elements of a row are separated by a comma and rows are separated by a semicolon) that has a unique Nash equilibrium which is not a PSNE, our bounds will be given in terms of instance-dependent quantities:
\begin{equation*}
    \begin{aligned}
    D&=a-b-c+d,\quad \Delta_{\min}&=\min\{|a-b|,|a-c|,|d-b|,|d-c|\}.
    \end{aligned}
\end{equation*}
The matrix $A = [a, b; c,d ]$ has a unique Nash equilibrium which is not a PSNE if and only if either of the following hold:
\begin{equation*}
    a<b,\ a<c,\ d<b,\ d<c,\quad \text{or} \quad a>b, a>c,\ d>b,\ d>c.
\end{equation*} Hence $|D|\geq 2\Delta_{\min}>0$. The material of this section show that the sample complexity of identifying an $\varepsilon$-good solution behaves as $\min\big\{\frac{1}{\varepsilon^2},\max\big\{\frac{1}{\Delta_{\min}^2},\frac{1}{\varepsilon |D|}\big\}\big\} \log(1/\delta)$ up to log factors. 
To motivate this bound, for matrix $A=[1, 0; 0,1]$ we have that $\min\big\{\frac{1}{\varepsilon^2},\max\big\{\frac{1}{\Delta_{\min}^2},\frac{1}{\varepsilon |D|}\big\}\big\}\approx \frac{1}{\varepsilon}$ which is significantly better than the trivial bound of $\frac{1}{\varepsilon^2}$.

To provide some intuition about where these quantities come from, suppose we measured each entry of $A$ exactly $T$ times and compiled the empirical means into a matrix $\widehat{A}$. 
If we let $(x,y)$ and $(\widehat{x},\widehat{y})$ be the Nash equilibria for $A$ and $\widehat{A}$, respectively, then $x^\top A y = \frac{ad-bc}{D}$ and we show in Appendices~\ref{appendix:minimax} and \ref{appendix:thm3} that we roughly have $| x^\top A y - \widehat{x}^\top A \widehat{y} | \leq \min\big\{ \frac{1}{\sqrt{T}}, \frac{1}{ T|D|} \big\}$. Moreover, we require roughly $\frac{1}{\Delta_{\min}^2}$ samples to decide whether $a<b$ or $b>a$ (same for other pairs). Without this information, we cannot characterize whether the input matrix has a PSNE or not, and this affects the value $V_A^*$ (which in turn affects the performance of the algorithm).
Hence, we observe it suffices to take $T \approx \min\big\{\frac{1}{\varepsilon^2},\max\big\{\frac{1}{\Delta_{\min}^2},\frac{1}{\varepsilon |D|}\big\}\big\}$. 
In the remainder of this section we make this argument rigorous and show that no smarter algorithm can improve upon this simple strategy. 

The following definition defines the set of algorithms under consideration.
\begin{definition}[$(\varepsilon,\delta)$-PAC-good]
We say an algorithm is $(\varepsilon,\delta)$-PAC-good if for all matrices $A \in \mathbb{R}^{m \times n}$ the algorithm terminates at an almost-sure finite stopping time $\tau \in \mathbb{N}$ and outputs a pair of mixed strategies $(x,y)\in \simplex_m\times \simplex_n$ such that $|V_A^*-\langle x, Ay\rangle|\leq \varepsilon$ with probability at least $1-\delta$.
\end{definition}
Our lower bounds will use this class of algorithms, and our proposed algorithm falls within this class. 


\subsection{Lower bound with respect to $|D|$}

This subsection derives a lower bound for the case when $A$ has a unique Nash equilibrium which is not a PSNE.

\begin{theorem}\label{thm:lower1}
Fix any matrix $A = [a, b; c,d ]$ that has a unique Nash equilibrium which is not a PSNE, $\varepsilon\in (0, \frac{\Delta_{\min}^2}{3|D|})$ and $\delta \in (0,1)$. 
Any $(\varepsilon,\delta)$-PAC-good algorithm that returns a pair of mixed strategies $(x,y)\in \simplex_2 \times\simplex_2$ at stopping time $\tau$ satisfies $\mathbb{E}_A[\tau] \geq \frac{ \log(1/30 \delta) }{3\varepsilon |D|}$.
\end{theorem}

The lower bound considers a class of matrices $A_\square$ parameterized by $\square \in \R$ defined as follows:
\[
A_\square = \begin{bmatrix} 
   a+\square & b \\
   c & d-\square \\
    \end{bmatrix}.\]
Clearly $A_\square = A$ when $\square=0$. Observe that if $|\square|<\Delta_{\min}$, then the matrix game defined by $A_\square$ has a unique Nash equilibrium which is not a PSNE. The proof of the theorem, found in Appendix~\ref{appendix:thm1}, follows from change of measure arguments applied to the instances defined in the following lemma.
\begin{lemma}\label{low:lem1}
Fix any $\varepsilon\in (0, \frac{\Delta_{\min}^2}{3|D|})$ and let $\Delta=\sqrt{3\varepsilon|D|}$. For any pair of mixed strategies $(x',y')\in \simplex_2 \times \simplex_2$ we have \[\displaystyle
    \max_{B\in\{A_{-\Delta},A_0,A_{\Delta}\}} |V_B^*-\langle x', By' \rangle|\geq \tfrac{3\varepsilon}{2}.\]
\end{lemma}
\noindent Unlike many lower bounds for multi-armed bandits that rely on a number of binary hypothesis tests being decided correctly (c.f., \cite{kaufmann2016complexity}), to prove the lower bound of this setting it is not possible to find a satisfying hypothesis test with fewer than three hypotheses due to the peculiar min-max behavior of the objective.

\subsection{Lower bound with respect to $\Delta_{\min}$ and $\varepsilon$}
Consider a matrix $A = [a, b; c,d ]$ that has a unique Nash equilibrium which is not a PSNE. Without loss of generality assume that $D>0$ and $\Delta_{\min}=a-b$. Let us also assume that $a-c\geq d-b$.
This subsection derives a lower bound with respect to $\Delta_{\min}$ and $\varepsilon$.

\begin{theorem}\label{thm:lower2}
Consider the matrix $A$ and fix any $\varepsilon>0$ and $\delta \in (0,1)$. 
Any $(\varepsilon,\delta)$-PAC-good algorithm that returns a pair of mixed strategies $(x,y)\in \simplex_2 \times\simplex_2$ at stopping time $\tau$ satisfies $\mathbb{E}_A[\tau] \geq \min\left\{\frac{\log(1/30 \delta) }{36\varepsilon^2 },\frac{\log(1/30 \delta) }{36\Delta_{\min}^2 }\right\}$.
\end{theorem}
The lower bound considers a class of matrices $A_\square$ parameterized by $\square \in \R$ defined as follows:
\[
A_\square = \begin{bmatrix} 
   a+\square & b-\square \\
   c+\square & d-\square \\
    \end{bmatrix}.\]
Clearly $A_\square=A$ when $\square=0$. The proof of the theorem, found in Appendix~\ref{appendix:thm2}, 
follows from change of measure arguments applied to the instances defined in the following lemma.

\begin{lemma}\label{low2:lem1}
Fix any $\varepsilon>0$. Let $\Delta=6\max\{\varepsilon,\Delta_{\min}\}$. For any pair of mixed strategies $(x',y')\in \simplex_2 \times \simplex_2$ we have $\displaystyle
    \max_{B\in\{A_{-\Delta},A_0,A_{\Delta}\}} |V_B^*-\langle x', By' \rangle| > \varepsilon.$
\end{lemma}

We can now combine Theorems~\ref{thm:lower1} and \ref{thm:lower2} to obtain the claimed result at the beginning of this section. 
Indeed, if $\varepsilon \in (0, \tfrac{\Delta_{\min}^2}{3|D|})$  then $\frac{ \log(1/30 \delta) }{3\varepsilon |D|} > \frac{\log(1/30 \delta) }{\Delta_{\min}^2 }$ and so by Theorem~\ref{thm:lower1} we have 
\begin{align*}
    \mathbb{E}_A[\tau] &\geq \tfrac{ \log(\tfrac{1}{30\delta}) }{3\varepsilon |D|}\\
    & = \max\left\{ \tfrac{ \log(\tfrac{1}{30\delta}) }{3\varepsilon |D|}, \tfrac{\log(\tfrac{1}{30\delta}) }{\Delta_{\min}^2 } \right\}\\
    &\geq \min\left\{ \tfrac{\log(\frac{1}{30\delta}) }{36\varepsilon^2 }, \max\left\{ \tfrac{ \log(\frac{1}{30\delta}) }{3\varepsilon |D|}, \tfrac{\log(\tfrac{1}{30\delta}) }{\Delta_{\min}^2 } \right\} \right\}.
\end{align*}

On the other hand, if $\varepsilon \geq \tfrac{\Delta_{\min}^2}{3|D|}$ then  $\frac{ \log(1/30 \delta) }{3\varepsilon |D|} \leq \frac{\log(1/30 \delta) }{\Delta_{\min}^2 }$ and so by Theorem~\ref{thm:lower2} we have 
\begin{align*}
    \mathbb{E}_A[\tau] &\geq \min\left\{\tfrac{\log(\tfrac{1}{30\delta}) }{36\varepsilon^2 },\tfrac{\log(\frac{1}{30\delta}) }{36\Delta_{\min}^2 }\right\}\\
    &= \tfrac{1}{36}  \min\left\{\tfrac{\log(\tfrac{1}{30\delta}) }{\varepsilon^2 }, \max\left\{ \tfrac{ \log(\tfrac{1}{30\delta}) }{3\varepsilon |D|}, \tfrac{\log(\tfrac{1}{30\delta}) }{\Delta_{\min}^2 }\right\} \right\}.
\end{align*}
Thus, for all $\varepsilon > 0$ we have \[\mathbb{E}_A[\tau] \geq  \min\left\{\tfrac{1 }{\varepsilon^2 }, \max\left\{ \tfrac{1 }{3\varepsilon |D|}, \tfrac{1 }{\Delta_{\min}^2 }\right\} \right\} \tfrac{\log(1/30 \delta)}{36}.\]

\subsection{Lower bound for games with multiple Nash Equilibria}
Consider a matrix $A = [a, a; c,d ]$ such that $a>c$, $a<d$.  Let us also assume that $a-c\geq d-a$.
Observe that the matrix game on $A$ has multiple Nash Equilibria and this game is a degenerate version of a matrix game with unique Nash Equilibrium which is not a PSNE (as $\Delta_{\min}=0$). This subsection derives a lower bound for the matrix game on $A$ with respect to $\varepsilon$.

\begin{theorem}\label{thm:multiple}
Consider the matrix $A$ and fix any $\varepsilon>0$ and $\delta \in (0,1)$. 
Any $(\varepsilon,\delta)$-PAC-good algorithm that returns a pair of mixed strategies $(x,y)\in \simplex_2 \times\simplex_2$ at stopping time $\tau$ satisfies $\mathbb{E}_A[\tau] \geq \frac{\log(1/30 \delta) }{36\varepsilon^2 }$.
\end{theorem}
The lower bound considers a class of matrices $A_\square$ parameterized by $\square \in \R$ defined as follows:
\[
A_\square = \begin{bmatrix} 
   a+\square & a-\square \\
   c+\square & d-\square \\
    \end{bmatrix}.\]
Clearly $A_\square=A$ when $\square=0$. The proof of the theorem, found in Appendix~\ref{appendix:thm:multiple}, follows from change of measure arguments applied to the instances defined in the following lemma.

\begin{lemma}\label{lem1:multiple}
Fix any $\varepsilon>0$. Let $\Delta=6\varepsilon$. For any pair of mixed strategies $(x',y')\in \simplex_2 \times \simplex_2$ we have $\displaystyle
    \max_{B\in\{A_{-\Delta},A_0,A_{\Delta}\}} |V_B^*-\langle x', By' \rangle| > \varepsilon.$
\end{lemma}

\subsection{Upper bound for $\varepsilon$--good solution}

As discussed above, Lemma~\ref{thm:lower1} describes a minimax optimal strategy. This subsection is dedicated to Algorithm~\ref{alg-ucb-1} that achieves an instance-dependent sample complexity for the special case of $m=n=2$.
Algorithm~\ref{alg-ucb-1} first samples the elements of $A$ until we can conclude whether $A$ has a PSNE or not. 
If $A$ has a PSNE, then we return it. 
If $A$ does not have a PSNE, we further sample each element of $A$ for $\tilde O(\frac{\log(\frac{1}{\delta})}{\varepsilon \tilde D})$ times and return the Nash equilibrium of the empirical matrix $\bar A$. 
Here $\tilde D$ is an empirical estimate of $|D|$. 
If no prior condition is met, the algorithm terminates in the worst case at iteration $t=T=\frac{8 \log(16/\delta)}{\varepsilon^2}$ and outputs an $\varepsilon$--good solution with high probability.
The full sample complexity guarantees of the algorithm are described in the following theorem whose proof is in Appendix \ref{appendix:thm3}.

\begin{theorem}\label{thm:alg1}
Fix any $\varepsilon > 0$ and $\delta \in (0,1)$. 
With probability at least $1-\delta$, Algorithm~\ref{alg-ucb-1} returns an $\varepsilon$--good solution after at most $n_0$ samples such that 
\begin{itemize}[itemsep=1pt,topsep=0pt]
    \item $n_0= c_1\cdot\min\big\{\frac{\log(1/\delta)}{\varepsilon^2},\max\big\{\frac{\log(\frac{1}{\varepsilon\delta})}{\Delta_{\min}^2},\frac{\log(\frac{1}{\varepsilon\delta})}{\varepsilon |D|}\big\}\big\}$ if the matrix game induced by $A$ has a unique Nash equilibrium which is not a PSNE, and
    \item $n_0=c_2\cdot\min\big\{\frac{\log(1/\delta)}{\varepsilon^2},\frac{\log(1/(\varepsilon\delta))}{\Delta_{\min}^2}\big\}$ if the matrix game induced by $A$ has a PSNE,
\end{itemize}
where $c_1,c_2$ are absolute constants.
\end{theorem}

\begin{algorithm}[t!]
\caption{Find an $\varepsilon$--good solution}
\begin{algorithmic}[1]
\STATE $T\gets\frac{8 \log (16/\delta) }{\varepsilon^2}$ 
\FOR{time step $t=1,2,\ldots,T$} 
\STATE Sample each element $(i,j)$ once and update the empirical means $\bar A_{ij}$.
\STATE $\Delta \gets \sqrt{2\log(\frac{16T}{\delta}) / t}$
\STATE $\tilde\Delta_{\min}\gets \min\{|\bar A_{11}-\bar A_{12}|,|\bar A_{21}-\bar A_{22}|\}|\bar A_{11}-\bar A_{21}|,|\bar A_{12}-\bar A_{22}|\}$
\STATE $\tilde D\gets |\bar A_{11}-\bar A_{12}-\bar A_{21}+\bar A_{22}|$
\IF{$1\leq \frac{\tilde \Delta_{\min}+2\Delta}{\tilde \Delta_{\min}-2\Delta}\leq \frac{3}{2}$ and the matrix game $\bar A$ has a PSNE }\label{alg1:con1}
\STATE Return  the PSNE of  $\bar A$.
\ELSIF{$1\leq \frac{\tilde \Delta_{\min}+2\Delta}{\tilde \Delta_{\min}-2\Delta}\leq \frac{3}{2}$ and $\tilde D< 10\varepsilon$}\label{alg1:con2} 
\STATE Sample each element $(i,j)$ for $T-t$ times.
\STATE Return the Nash equilibrium of $\bar A$.
\ELSIF{$1\leq \frac{\tilde \Delta_{\min}+2\Delta}{\tilde \Delta_{\min}-2\Delta}\leq \frac{3}{2}$ and $\tilde D\geq 10\varepsilon$}\label{alg1:con3} 
\STATE $N\gets \frac{80 \log (\frac{16T}{\delta}) }{\varepsilon \tilde D }$
\IF{$N>T-t$}\label{alg1:con4}
\STATE Sample each element $(i,j)$ for $T-t$ times
\STATE Return the Nash equilibrium of $\bar A$.
\ENDIF
\STATE Sample each element $(i,j)$ for $N$ times.
\STATE Return the Nash equilibrium of $\bar A$.
\ENDIF
\ENDFOR
\STATE Return the Nash equilibrium of $\bar A$. \label{alg1:con5}
\end{algorithmic}
\label{alg-ucb-1}
\end{algorithm}

\section{{Results for $\varepsilon$--Nash Equilibrium of $2\times 2$ Matrix Games}}\label{sec:epsilon_Nash}
This section is devoted to instance-dependent sample complexity bounds for identifying an $\varepsilon$--Nash equilibrium  $(x,y)$. Recall that both $\la x,Ay\ra\geq \la x',Ay\ra-\vep$ and $\la x,Ay'\ra \geq \la x,Ay\ra -\vep$ hold for all $(x',y')\in \simplex_2\times \simplex_2$.
For a matrix $A = [a, b; c,d ]$, our bounds will be given in terms of the following instance-dependent quantities:
\begin{equation*}
    \begin{aligned}
    D&=a-b-c+d, \quad \Delta_{m_2}&=\max\{\min\{|a-b|, |d-c|\},\min\{|a-c|,|d-b|\}\}.
    \end{aligned}
\end{equation*}
The sample complexity of identifying an $\varepsilon$--Nash equilibrium is $\min\big\{\frac{1}{\varepsilon^2},\max\big\{\frac{1}{\Delta_{\min}^2},\frac{\Delta_{m_2}^2}{\varepsilon^2 D^2}\big\}\big\} \log(1/\delta)$ up to log factors. To motivate this bound, consider the matrix $A=[1+\varepsilon^{0.5}, 1; 0,1+\varepsilon^{0.5}]$ where $0<\varepsilon<1$. Then $\min\big\{\frac{1}{\varepsilon^2},\max\big\{\frac{1}{\Delta_{\min}^2},\frac{\Delta_{m_2}^2}{\varepsilon^2 D^2}\big\}\big\}\approx \frac{1}{\varepsilon}$ which is significantly better than the trivial bound of $\frac{1}{\varepsilon^2}$. 
On the other hand, for the matrix $B=[1,0; 0,1]$ we have that  $\min\big\{\frac{1}{\varepsilon^2},\max\big\{\frac{1}{\Delta_{\min}^2},\frac{\Delta_{m_2}^2}{\varepsilon^2 D^2}\big\}\big\}\approx \frac{1}{\varepsilon^2}$ which is significantly worse than the bound of $\frac{1}{\varepsilon}$ that we achieved for identifying an $\varepsilon$--good solution before. This shows that finding an $\varepsilon$--Nash equilibrium can require many more samples than finding an $\varepsilon$--good solution. This is not unexpected as every $\varepsilon$--Nash equilibrium is also an $\varepsilon$--good solution.

To provide some intuition about where these quantities come from, suppose we measured each entry of $A$ exactly $T$ times and compiled the empirical means into a matrix $\widehat{A}$. 
If we let $(\widehat{x},\widehat{y})$ be the Nash equilibrium for $\widehat{A}$, respectively, then we show (cf.~Appendix~\ref{appendix:thm5}) that we roughly have 
\begin{equation*}
    \begin{aligned}
    \max_{x'\in\simplex_2}x'^\top A \widehat{y} - \widehat{x}^\top A \widehat{y}\leq \tfrac{\Delta_{m_2}}{\sqrt{T}|D|}\quad\text{and}\quad
    \widehat{x}^\top A \widehat{y}-\min_{y'\in\simplex_2}\widehat{x}^\top A y'  \leq \tfrac{\Delta_{m_2}}{\sqrt{T}|D|}.
    \end{aligned}
\end{equation*} Moreover, in the previous section we showed that roughly $1/\Delta_{\min}^2$ samples are required to distinguish between various alternatives. 
Hence, we observe it suffices to take $T \approx \min\big\{\frac{1}{\varepsilon^2},\max\big\{\frac{1}{\Delta_{\min}^2},\frac{\Delta_{m_2}^2}{\varepsilon^2 D^2}\big\}\big\}$. 
In the remainder of this section we make this argument rigorous and show that no smarter algorithm can improve upon this simple strategy. 

The following  is the definition of the set of algorithms under consideration.
\begin{definition}[$(\varepsilon,\delta)$-PAC-Nash]
We say an algorithm is $(\varepsilon,\delta)$-PAC-Nash if for all induced by matrices $A \in \mathbb{R}^{m \times n}$, the algorithm terminates at an almost--sure finite stopping time $\tau \in \mathbb{N}$ and outputs a pair of mixed strategies $(x,y)\in \simplex_m\times \simplex_n$ such that $\la x,Ay\ra\geq \la x',Ay\ra-\vep$ and $\la x,Ay'\ra \geq \la x,Ay\ra -\vep$ hold for all $(x',y')\in \simplex_m\times \simplex_n$ with probability at least $1-\delta$.
\end{definition}
Our lower bounds will use this class of algorithms, and the proposed algorithm falls within this class. 

\subsection{Lower bound for finding $\varepsilon$--Nash equilibrium}
This subsection derives a lower bound for the case when $A$ has a unique Nash equilibrium which is not a PSNE.

\begin{theorem}\label{thm:lower3}
Fix any matrix $A = [a, b; c,d ]$ that has a unique Nash equilibrium which is not a PSNE, $\varepsilon>0$ and $\delta \in (0,1)$. 
Any $(\varepsilon,\delta)$-PAC-Nash algorithm that returns a pair of mixed strategies $(x,y)\in \simplex_2 \times\simplex_2$ at stopping time $\tau$ satisfies $\mathbb{E}_A[\tau] \geq \frac{\Delta_{m_2}^2 \log(1/30 \delta) }{9\varepsilon^2D^2}$.
\end{theorem}

Without loss of generality assume that $D>0$ and $\Delta_{m_2}=a-b$. The lower bound considers a class of matrices $A_\square$ parameterized by $\square \in \R$ defined as follows:
\[
A_\square = \begin{bmatrix} 
   a+\square & b+\square \\
   c-\square & d-\square \\
    \end{bmatrix}.\]
Clearly $A_\square=A$ when $\square=0$. The proof of the theorem, found in Appendix~\ref{appendix:thm4}, follows from change of measure arguments applied to the instances defined in the following lemma.
\begin{lemma}\label{low3:lem1}
Fix any $\varepsilon>0$. For any pair of mixed strategies $(x',y')\in \simplex_2 \times \simplex_2$, there exists a matrix $B\in\{A_0,A_{\Delta},A_{-\Delta}\}$ such that $(x',y')$ is not an $\varepsilon$-Nash equilibrium of $B$ where $\Delta:= \frac{3\varepsilon D}{\Delta_{m_2}}$.
\end{lemma}
Recall that any lower bound for an $\varepsilon$--good solution also holds for $\varepsilon$--Nash equilibrium. Hence,  combining Theorems~\ref{thm:lower2} and \ref{thm:lower3} we obtain the claimed result at the beginning of this section. Note that the lower bound of Theorem \ref{thm:lower1} is redundant as $\frac{1}{\varepsilon|D|}\leq \frac{\Delta_{m_2}^2}{\varepsilon^2|D|^2}$ when $\varepsilon\in(0,\frac{\Delta_{\min}^2}{3|D|})$.

\subsection{Upper bound for $\varepsilon$-Nash equilibrium}

\begin{algorithm}[h]
\caption{Find an $\varepsilon$--Nash equilibrium}
\begin{algorithmic}[1]
\STATE $T\gets\frac{8 \log (16/\delta) }{\varepsilon^2}$ 
\FOR{round $t=1,2,\ldots,T$} 
\STATE Sample each element $(i,j)$ once and update the empirical means $\bar A_{ij}$.
\STATE $\Delta \gets \sqrt{2\log(\frac{16T}{\delta}) / (t)}$.
\STATE $\tilde\Delta_{\min}\gets \min\{|\bar A_{11}-\bar A_{12}|,|\bar A_{21}-\bar A_{22}|\}|\bar A_{11}-\bar A_{21}|,|\bar A_{12}-\bar A_{22}|\}$
\STATE $\tilde \Delta_{m_2}\gets \max\{\min\{|\bar A_{11}-\bar A_{12}|,|\bar A_{21}-\bar A_{22}|\},\min\{|\bar A_{11}-\bar A_{21}|,|\bar A_{12}-\bar A_{22}|\}\}$
\STATE $\tilde D\gets |\bar A_{11}-\bar A_{12}-\bar A_{21}+\bar A_{22}|$
\IF{$1\leq \frac{\tilde \Delta_{\min}+2\Delta}{\tilde \Delta_{\min}-2\Delta}\leq \frac{3}{2}$ and the game induced by $\bar A$ has a pure strategy Nash}\label{alg2:con1}
\STATE Return the Nash equilibrium associated with the PSNE of $\bar A$.
\ELSIF{$1\leq \frac{\tilde \Delta_{\min}+2\Delta}{\tilde \Delta_{\min}-2\Delta}\leq \frac{3}{2}$ and $\tilde \Delta_{m_2}\geq \tilde D/8$} \label{alg2:con2}
\STATE Sample each element $(i,j)$ for $T-t$ times.
\STATE Return the Nash equilibrium of $\bar A$.
\ELSIF{$1\leq \frac{\tilde \Delta_{\min}+2\Delta}{\tilde \Delta_{\min}-2\Delta}\leq \frac{3}{2}$ and $\tilde \Delta_{m_2}< \tilde D/8$} \label{alg2:con3}
\STATE $N\gets \frac{200\tilde\Delta^2_{m_2}\log (\frac{16T}{\delta})}{\varepsilon^2 \tilde D^2}$
\IF{$N>T-t$}\label{alg2:con4}
\STATE Sample each element $(i,j)$ for $T-t$ times.
\STATE Return the Nash equilibrium of $\bar A$.
\ENDIF
\STATE $\Delta_1 \gets \sqrt{2\log(\frac{16T}{\delta}) /(N+t)}$
\STATE Sample each element $(i,j)$ for $N$ times.
\STATE $i_1\gets \arg\min_i |\bar A_{i1}-\bar A_{i2}|$ and $i_2 \gets \{1,2\}\setminus \{i_1\}$
\STATE $j_1\gets \arg\min_i |\bar A_{1j}-\bar A_{2j}|$ and $j_2 \gets \{1,2\}\setminus \{j_1\}$
\STATE $B_{i_1j_1}\gets \bar A_{i_1j_1}$, $B_{i_2j_2}\gets \bar A_{i_2j_2}$, $B_{i_1j_2}\gets \bar A_{i_1j_2}-2\Delta_1$, $B_{i_2j_1}\gets \bar A_{i_2j_1}+2\Delta_1$
\STATE Return the Nash equilibrium of $B$.
\ENDIF
\ENDFOR
\STATE Return the Nash equilibrium of $\bar A$.  \label{alg2:con5}
\end{algorithmic}
\label{alg-ucb-2}
\end{algorithm}


Next, we characterize the instance-dependent sample complexity of Algorithm~\ref{alg-ucb-2} for the special case of $2\times 2$ matrix.
Algorithm~\ref{alg-ucb-2} first samples the elements of $A$ until we can conclude whether the game induced by $A$ has a PSNE or not. If the game does, then we return it. If the matrix game induced by $A$ does not have a PSNE, we further sample each element of $A$ for $\tilde O(\frac{\tilde\Delta^2_{m_2}\log (\frac{1}{\delta})}{\varepsilon^2 \tilde D^2})$ times and return the Nash equilibrium of a matrix $B$ that we get by slightly modifying the empirical matrix $\bar A$. Here $\tilde D$ and $\tilde\Delta_{m_2}$ are the empirical estimates of $|D|$ and $\Delta_{m_2}$, respectively. If no prior condition is met, the algorithm terminates in the worst case at iteration $t=T=\frac{8 \log(16/\delta)}{\varepsilon^2}$ and outputs an $\varepsilon$--Nash equilibrium with high probability by Lemma~\ref{thm:lower1}. The full sample complexity guarantees of the algorithm are described in the following theorem the proof of which is in Appendix~\ref{appendix:thm5}.

\begin{theorem}\label{thm:alg2}
 Fix any $\varepsilon > 0$ and $\delta \in (0,1)$. 
With probability at least $1-\delta$ Algorithm~\ref{alg-ucb-2} returns an $\varepsilon$--Nash equilibrium after at most $n_0$ samples where 
\begin{itemize} 
    \item $n_0= c_1\cdot\min\big\{\frac{\log(\tfrac{1}{\varepsilon \delta})}{\varepsilon^2},\max\big\{\frac{\log(\tfrac{1}{\varepsilon \delta})}{\Delta_{\min}^2},\frac{\Delta_{m_2}^2\log(\tfrac{1}{\varepsilon \delta})}{\varepsilon^2 D^2}\big\}\big\}$ if the matrix game induced by $A$ has a unique Nash equilibrium which is not a PSNE, and
    \item $n_0=c_2\cdot\min\big\{\frac{\log(1/\delta)}{\varepsilon^2},\frac{\log(1/\varepsilon\delta)}{\Delta_{\min}^2}\big\}$ otherwise,
\end{itemize}
where $c_1,c_2$ are absolute constants.
\end{theorem}

\section{{Results for $n\times 2$ Matrix Games}}
This section is devoted to instance-dependent sample complexity bounds for identifying an $\varepsilon$--good solution and an $\varepsilon$--Nash equilibrium in a $n\times 2$ matrix game that has a unique Nash equilibrium. For a matrix $A\in \mathbb{R}^{n\times 2}$, the bounds will be given in terms of instance-dependent quantities $\Delta_{\min}$ and $\Delta_g$.  
For $\Delta_{\min}$, the natural extension from  the $2 \times 2$ to the $n\times 2$ case is 
\begin{equation*}
    \begin{aligned}
    \Delta_{\min}= \min\{\min_i\{| A_{i1}- A_{i2}|\},\min\limits_{j,k:j\neq k}\{| A_{j1}- A_{k1}|\},\min\limits_{j,k:j\neq k}\{| A_{j2}-A_{k2}|\}\}.
    \end{aligned}
\end{equation*}
To define $\Delta_g$, observe that if the matrix game induced by $A$ has a unique Nash equilibrium $(x^*,y^*)$, then $|\supp(x^*)|=|\supp(y^*)|$ \cite{bohnenblust1950solutions}. 
Suppose that $A$ has a unique Nash equilibrium $(x^*,y^*)$ such that $\supp(x^*)=\{i_1,i_2\}$. Let  $A_{i_11}>A_{i_12}$ and $A_{i_21}<A_{i_22}$ and define
    \[\Delta_g:= \min\limits_{i\in[n]\setminus\{i_1,i_2\}} r_i\cdot(V^*_{A}-\langle y^*,( A_{i1}, A_{i2})\rangle),\]
where $r_i=\tfrac{| A_{i_11}- A_{i_12}|+| A_{i_21}- A_{i_22}|}{| A_{i_11}- A_{i_12}|+| A_{i_21}- A_{i_22}|+| A_{i1}- A_{i2}|}$.
It is not hard to see that $\min_{i\in[n]\setminus\{i_1,i_2\}}V^*_{A}-\langle y^*,( A_{i1}, A_{i2})\rangle>0$ which implies $\Delta_{g}>0$. 

To provide some intuition for the origin of $\Delta_g$, consider a class of matrices $A_\square=[a,b;c-\square,d-\square;e+\square,f+\square]$ parameterized by $\square \in \R$. In Appendix \ref{apendix:low4:lem1}, we show that for $\Delta=c_0\cdot \Delta_g$ where $c_0$ is an absolute constant, the matrices $A_0$, $A_\Delta$ and $A_{2\Delta}$ have different supports for their respective Nash equilibrium. This implies that we require roughly $1/\Delta_{\g}^2$ samples to determine the support of the Nash equilibrium of $A$. Without this information, we cannot determine an $\varepsilon$--good solution with high probability as the support of the Nash equilibrium affects the value $V_A^*$. The same holds true for finding an $\varepsilon$--Nash equilibirum as every $\varepsilon$--Nash equilibrium is also an $\varepsilon$--good solution. Moreover, to obtain any meaningful upper bound, we require an empirical estimate of $\Delta_g$ to be close to $\Delta_g$ and this is possible when we re-scale the gaps $V^*_{A}-\langle y^*,( A_{i1}, A_{i2})\rangle$ by a factor of $r_i$. In the remainder of this section we make this argument rigorous and further show that roughly $1/\Delta_{\g}^2$ samples suffices to find the support of the Nash equilibrium of $A$. Once the support is identified, we can use the algorithms derived for the $2\times2$ case in the previous two sections.



\subsection{Lower bound with respect to $\Delta_g$}
Consider any matrix game \[A = \bmat{a & b\\ c & d\\ e & f }\] that has a unique Nash equilibrium $(x^*,y^*)$ which is not a PSNE. Without loss of generality assume that $\supp(x^*)=\supp(y^*)=\{1,2\}$. Let us also assume that $a>b,a>c,a>e,d>b,d>c,f>e,f>b$. Observe that in this case we have the following:
\begin{equation*}
     \Delta_g= \frac{(|a-b|+|c-d|)(V^*_{A}-\langle y^*,( e, f)\rangle)}{|a-b|+|c-d|+|e-f|}
\end{equation*}
Let $D_1:=a-b-c+d$ and $D_2:=a-b-e+f$. Let $\Delta:=\frac{(d-b)D_2-(f-b)D_1}{D_1+D_2}$ and $\lambda:=\min\{\frac{(a-b)\Delta}{D_1},\frac{(a-b)\Delta}{D_2}\}$. 
This subsection derives a lower bound for $\varepsilon$--good solution for the matrix game $A$.

\begin{theorem}\label{thm:lower4}
Consider the matrix $A$ and fix any $\varepsilon\in (0, \frac{\lambda}{4})$ and $\delta \in (0,1)$. 
Any $(\varepsilon,\delta)$-PAC-good algorithm that returns a pair of mixed strategies $(x,y)\in \simplex_3 \times\simplex_2$ at stopping time $\tau$ satisfies $\mathbb{E}_A[\tau] \geq \frac{ \log(1/30 \delta) }{4\Delta_{g}^2}$.
\end{theorem}

The lower bound considers a class of matrices $A_\square$ parameterized by $\square \in \R$ defined as follows:
\[
A_\square = \begin{bmatrix} 
   a & b \\
   c-\square & d-\square \\
   e+\square & f+\square \\
    \end{bmatrix}.\]
Clearly $A_\square=A$ when $\square=0$. The proof of the theorem, found in  Appendix~\ref{appendix:thm6}, follows from change of measure arguments applied to the instances defined in the following lemma.

\begin{lemma}\label{low4:lem1}
Fix any $\varepsilon\in(0,\frac{\lambda}{4})$. For any pair of mixed strategies $(x',y')\in \simplex_3 \times \simplex_2$,
we have \[\displaystyle
    \max_{B\in\{A_{0},A_{\Delta},A_{2\Delta}\}} |V_B^*-\langle x', By' \rangle|> \varepsilon.\]
%
\end{lemma}

\subsection{Finding the support in $n\times 2$ matrix games}
\begin{algorithm}[h]
\caption{Find the equilibrium support for a $n\times 2$ matrix}
\begin{algorithmic}[1]
\STATE $T\gets\frac{8 \log (8n/\delta) }{\varepsilon^2}$ 
\FOR{round $t=1,2,\ldots,T$} 
\STATE Sample each element $(i,j)$ once and update the empirical means $\bar A_{ij}$.
\STATE $\Delta \gets \sqrt{2\log(\frac{8nT}{\delta}) / (t)}$.
\STATE $\tilde\Delta_{\min}$ $\gets$ $\min$ $\{\min_i\{|\bar A_{i1}-\bar A_{i2}|\},$ $\min\limits_{j,k:j\neq k}\{|\bar A_{j1}-\bar A_{k1}|\},$ $\min\limits_{j,k:j\neq k}\{|\bar A_{j2}-\bar A_{k2}|\}\}$
\IF{$1\leq \frac{\tilde \Delta_{\min}+2\Delta}{\tilde \Delta_{\min}-2\Delta}\leq \frac{3}{2}$ and the matrix game $\bar A$ has a PSNE }\label{alg3:con1}
\STATE Return the PSNE of $\bar A$.
\ELSIF{$1\leq \frac{\tilde \Delta_{\min}+2\Delta}{\tilde \Delta_{\min}-2\Delta}\leq \frac{3}{2}$ and $\bar A$ does not have a pure strategy Nash} \label{alg3:con2}
\STATE $\forall i\in[n]$, remove the row $i$ of $\bar A$ if there is a row $j$ such that $\bar A_{j1}>\bar A_{i1}$ and $\bar A_{j2}>\bar A_{i2}$.
\FOR{round $t'=t+1,t+2,\ldots,T$}
\STATE Sample each element $(i,j)$ once and update the empirical means $\bar A_{ij}$.
\STATE $\Delta' \gets \sqrt{2\log(\frac{8nT}{\delta}) / (t')}$
\STATE $(x',y')\gets$ Nash equilibrium of $\bar A$
\STATE If $t'=T$, return the Nash equilibrium of $\bar A$. \label{alg3:con3}
\IF{$|\supp(x')|= 2$}\label{alg3:con4}
\STATE $\{i_1,i_2\}\gets \supp(x')$
\STATE For all rows $i$,\\ $\tilde{r_i}\gets \frac{| \bar A_{i_11}- \bar A_{i_12}|+| \bar A_{i_21}- \bar A_{i_22}|}{|\bar A_{i_11}-\bar A_{i_12}|+|\bar A_{i_21}- \bar A_{i_22}|+|\bar A_{i1}-\bar A_{i2}|}$
\STATE $\tilde \Delta_g\gets \min\limits_{i:i\notin \{i_1,i_2\}}\tilde{r_i}\cdot (V^*_{\bar A}-\langle y',(\bar A_{i1},\bar A_{i2})\rangle)$
\STATE If $\tilde \Delta_g\geq 4\Delta'$, then Return $\{\{i_1,i_2\},\{1,2\}\}$ as the support of the Nash equilibrium.\label{alg3:con5}
\ENDIF
\ENDFOR
\ENDIF
\ENDFOR
\STATE Return the Nash equilibrium of $\bar A$.  \label{alg3:con6}
\end{algorithmic}
\label{alg-ucb-3}
\end{algorithm}

Next, we characterize the instance-dependent sample complexity of 
Algorithm~\ref{alg-ucb-3} which  finds the support of the unique Nash equilibrium, having cardinality at most two, for the matrix game $A\in \mathbb{R}^{n\times 2}$.
Algorithm~\ref{alg-ucb-3} first samples the elements of $A$ until we can conclude whether the game induced by $A$ has a PSNE or not. If the game has a PSNE, then we return it. If the game induced by $A$ does not have a PSNE, we further sample the elements of $A$ until $\tilde \Delta_g$ is sufficiently large and return the support of the Nash equilibrium of the empirical matrix $\bar A$. Here $\tilde \Delta_g$ is an empirical estimate of $\Delta_g$. If no prior condition is met, the algorithm terminates in the worst case at iteration $t=T=\frac{8 \log(8n/\delta)}{\varepsilon^2}$ and outputs an $\varepsilon$--Nash equilibrium with high probability by Lemma~\ref{lem:trivial:NE}.
The full sample complexity guarantees of the algorithm are described in the following theorem the proof of which is in Appendix~\ref{appendix:thm7}.

\begin{theorem}\label{thm:alg3}
Fix any $\varepsilon > 0$ and $\delta \in (0,1)$. Let $T=\frac{8 \log (8n/\delta) }{\varepsilon^2}$. Consider a game defined by the matrix $A\in \mathbb{R}^{n\times 2}$ with a unique Nash equilibrium $(x^*,y^*)\in\simplex_n\times\simplex_2$. The following hold, where $c_1,c_2,c_3,c_4$ are absolute constants.
\begin{itemize}
    \item If $A$ has a PSNE and $\frac{800\log (\frac{8nT}{ \delta})}{ \Delta_{\min}^2}\leq T$, then  with probability at least $1-\delta$, Algorithm~\ref{alg-ucb-3} samples each element of $A$ for $n_0$ times and returns a PSNE where $n_0= c_1\cdot\frac{\log (\frac{n}{\varepsilon \delta})}{ \Delta_{\min}^2}$.
    \item If $A$ has a PSNE and $\frac{800\log (\frac{8nT}{ \delta})}{ \Delta_{\min}^2}> T$, then  with probability at least $1-\delta$, Algorithm~\ref{alg-ucb-3} samples each element of $A$ for $n_0$ times and either returns a PSNE or an $\varepsilon$--Nash equilibrium where $n_0= c_2\cdot\frac{\log(n/\delta)}{\varepsilon^2}$.
    \item If $A$ does not have a PSNE and $\max\big\{\frac{800\log (\frac{8nT}{ \delta})}{ \Delta_{\min}^2},\frac{722\log (\frac{8nT}{ \delta})}{ \Delta_{g}^2}\big\}<T$, then  with probability at least $1-\delta$, Algorithm~\ref{alg-ucb-3} samples each element of $A$ for $n_0$ times and returns $\supp(x^*)$ and $\supp(y^*)$ where $n_0= c_3\cdot\max\big\{\frac{\log (\frac{n}{\varepsilon \delta})}{ \Delta_{\min}^2},\frac{\log (\frac{nT}{ \delta})}{ \Delta_{g}^2}\big\}$.
    \item If $A$ does not have a PSNE and $\max\big\{\frac{800\log (\frac{8nT}{ \delta})}{ \Delta_{\min}^2},\frac{722\log (\frac{8nT}{ \delta})}{ \Delta_{g}^2}\big\}\geq T$, then  with probability at least $1-\delta$, Algorithm~\ref{alg-ucb-3} samples each element of $A$ for $n_0$ times and either returns $\supp(x^*)$ and $\supp(y^*)$ or an $\varepsilon$--Nash equilibrium where $n_0= c_4\cdot \frac{\log(n/\delta)}{\varepsilon^2}$.
\end{itemize}
\end{theorem}

If Algorithm $\ref{alg-ucb-3}$ returns $\supp(x^*)$ and $\supp(y^*)$ such that $|\supp(x^*)|=|\supp(y^*)|=2$, then  Algorithms $\ref{alg-ucb-1}$ and $\ref{alg-ucb-2}$ can be run on the $2\times 2$ sub-matrix formed by $\supp(x^*)$ and $\supp(y^*)$ and return, with high probability, an $\varepsilon$--good solution and $\varepsilon$--Nash equilibrium, respectively.

%% file: discussion.tex
\section{{Conclusion and Open Questions}}
To the best of our knowledge, this work provides the first instance-dependent sample complexity results for zero-sum normal form games. We have completely characterized the instance dependent sample complexity of finding Nash Equilibrium in $2\times 2$ matrix games. In addition, we have shed some light on the case of $n\times 2$ matrix games.
These results shed light on the properties of a game that make its dynamics converge quickly or slowly. 
The implications of this line of results could be new algorithms designed to take advantage of easy games, where previous minimax optimal algorithms may not.
This more nuanced understanding of instance-dependent sample complexity may also influence mechanism design since our results describe specifically how one could speed up convergence of players to a Nash equilibrium. 

However, our work leaves many questions unresolved as well as revealing new ones. 
The most obvious direction--extending our results to general $(n \times m) \in \mathbb{N} \times \mathbb{N}$--is also one of the most challenging. 
First, unlike our $n\times 2$ case in which the size of the support is trivially at most $k=2$, it is unclear how to identify the true size $k$ of the support of the Nash equilibrium in general, and then how to identify the $k \times k$ sub-matrix within the game matrix.
Second, there does not exist a closed-form expression for the Nash equilibrium of general $(n\times m) \in \mathbb{N} \times \mathbb{N}$ matrix games.
We exploit the existence of the closed-form solution in our $2\times 2$ analysis in many ways, including deriving alternative instances for lower bounds, and also understanding the right notions of gap by considering a perturbation of the optimal solution.  Due to Cramer's rule there exists a closed-form expression for the Nash equilibrium of general $(n\times n) $ matrix game, however, we would have to analyze how determinants of a matrix behave under minor perturbations in order to establish meaningful upper bounds. 

Besides larger game matrices, there are other very natural questions to pursue. 
Given the instance-dependent quantities we introduced in this work, how do these generalize to general-sum games and can our lower bound strategies be extended? 
What is the sample complexity of identifying other kinds of equilibria, such as an $\varepsilon$ (coarse) correlated equilibrium?
Finally, can we derive instance-dependent regret bounds for computationally efficient strategies?

%% file: ack.tex
\section*{ACKNOWLEDGEMENTS}
The authors are supported by NSF RI Award 
\# 1907907.

%% file: appendix.tex
\section{Properties of Matrix games}\label{appendix:properties:matrix}
Let $e_j^k$ denote a $k$-dimensional vector such that its $j$-th component is $1$ and the rest of the components are $0$.
We now state some well known properties of Nash Equilibrium (in short NE) of Matrix games.
\begin{enumerate}
    \item (\citet{karlin2017game}) If $(x^*,y^*)$ is a Nash equilibrium of a matrix game $A\in\mathbb{R}^{m\times n}$, then $\langle x^*, Ay^*\rangle=V^*_A$ .
    \item (\citet{karlin2017game}) Let $(x^*,y^*)$ be a Nash equilibrium of a matrix game $A\in\mathbb{R}^{m\times n}$. Then the following holds true:
    \begin{itemize}
        \item For any $i\in \supp(x^*)$, $\langle e_i^m, Ay^*\rangle =V^*_A$. Similarly, for any $j\in \supp(y^*), \langle x^*, Ae_j^n\rangle=V^*_A$.
        \item  For any $i\notin \supp(x^*)$, $\langle e_i^m, Ay^*\rangle \leq V^*_A$. Similarly, for any $j\notin \supp(y^*), \langle x^*, Ae_j^n\rangle\geq V^*_A$.
    \end{itemize}
    \item (\citet{bohnenblust1950solutions}) Consider a matrix game on $A\in\mathbb{R}^{m\times n}$ that has a unique Nash equilibrium $(x^*,y^*)$. Then the following holds true:
    \begin{itemize}
        \item $|\supp(x^*)|=|\supp(y^*)|$
        \item For any $i\in \supp(x^*)$, $\langle e_i^m, Ay^*\rangle =V^*_A$. Similarly, for any $j\in \supp(y^*), \langle x^*, Ae_j^n\rangle=V^*_A$.
        \item  For any $i\notin \supp(x^*)$, $\langle e_i^m, Ay^*\rangle < V^*_A$. Similarly, for any $j\notin \supp(y^*), \langle x^*, Ae_j^n\rangle> V^*_A$.
    \end{itemize}
    \end{enumerate}
Next we present some useful properties of Matrix games.

\begin{proposition}\label{prop:matrix:dmin}
If $\Delta_{\min}=0$, then the matrix game on $A=[a,b;c,d]$ has a PSNE.
\end{proposition}
\begin{proof}
W.l.o.g let us assume that $a=b$. If $a\geq c$, then $(1,1)$ is a PSNE of $A$. If $a\geq d$, then $(1,2)$ is a PSNE of $A$. If $a\leq c\leq d$, then $(2,1)$ is a PSNE of $A$. If $a\leq d\leq c$, then $(2,2)$ is a PSNE of $A$. 
\end{proof}
\begin{proposition}\label{prop:matrix:saddle}
If the matrix game on $A=[a,b;c,d]$ has a Nash equilibrium $(x^*,y^*)$ such that $\min\{|\supp(x^*)|,|\supp(y^*)|\}=1$, then the matrix game on $A$ has a PSNE.
\end{proposition}
\begin{proof}
If $\supp(x^*)=\{i\}$ and $\supp(y^*)=\{j\}$, then $(i,j)$ is a PSNE of $A$. Now w.l.o.g let us assume that $\supp(x^*)=\{1\}$ and $\supp(y^*)=\{1,2\}$. This implies that $a=b$. Hence due to Proposition \ref{prop:matrix:dmin}, matrix game on $A$ has a PSNE.
\end{proof}
\begin{proposition}\label{prop:matrix:nosaddle}
If the matrix game on $A=[a,b;c,d]$ does not have a PSNE, then either $a<b$, $a<c$, $d<b$ and $d<c$ or $a>b$, $a>c$, $d>b$ and $d>c$.
\end{proposition}
\begin{proof}
Due to Proposition \ref{prop:matrix:dmin}, we have $\Delta_{\min}>0$. Let us first assume that $a>b$. Then $d>c$ otherwise $A$ has a PSNE. Similarly we have $a>c$ otherwise $(2,1)$ would be a PSNE of $A$. Also we have $b<d$ otherwise $(1,2)$ would be a PSNE of $A$. Similarly we can show that if $a<b$, then $a<c$, $d<b$ and $d<c$.
\end{proof}
\begin{proposition}\label{prop:matrix:main}
\begin{enumerate}
    \item The matrix game on $A=[a,b;c,d]$ has a unique Nash equilibrium which is not a PSNE if and only if one of the following condition holds true:
    \begin{itemize}
        \item $a<b$, $a<c$, $d<b$ and $d<c$
        \item $a>b$, $a>c$, $d>b$ and $d>c$
    \end{itemize}
    \item Consider a matrix game on $A=[a,b;c,d]$ that has a unique Nash equilibrium $(x^*,y^*)$ which is not a PSNE. Then $x^*=\left(\frac{d-c}{D},\frac{a-b}{D}\right)$, $y^*=\left(\frac{d-b}{D},\frac{a-c}{D}\right)$ and $V_A^*=\frac{ad-bc}{D}$ where $D=a-b-c+d$.
\end{enumerate}
\end{proposition}
\begin{proof}
Let us assume that $a<b$, $a<c$, $d<b$ and $d<c$. Then the matrix game on $A=[a,b;c,d]$ does not have a PSNE. Hence due to Proposition \ref{prop:matrix:saddle}, any Nash equilibrium $(x^*,y^*)$ of the matrix game on $A$ has $\supp(x^*)=\supp(y^*)=\{1,2\}$. Let $((x,1-x),(y,1-y))$ be a Nash equilibrium of matrix game on $A$. Then it must satisfy the following equations:
\begin{align*}
    ax+(1-x)c&=bx+(1-x)d\\
    ay+(1-y)b&=cy+(1-y)d
\end{align*}
The above equations have a unique solution which is $x=\frac{d-c}{a-b-c+d}$ and $y=\frac{d-b}{a-b-c+d}$. In a similar way we can show that if $a>b$, $a>c$, $d>b$ and $d>c$, then the matrix game on $A$ has unique Nash equilibrium $((x,1-x),(y,1-y))$ where $x=\frac{d-c}{a-b-c+d}$ and $y=\frac{d-b}{a-b-c+d}$. We also have $V_A^*=\frac{d-c}{D}\cdot a+\frac{a-b}{D}\cdot c=\frac{ad-bc}{D}$ where $D=a-b-c+d$.

Next let us assume that the matrix game on $A=[a,b;c,d]$ has a unique Nash equilibrium which is not a PSNE. Then it does not have a PSNE. Due to Proposition \ref{prop:matrix:nosaddle}, either $a<b$, $a<c$, $d<b$ and $d<c$ or $a>b$, $a>c$, $d>b$ and $d>c$.
\end{proof}

\begin{proposition}
Consider a matrix $A=[a,b;c,d]$. If $\Delta_{\min}>0$, then the input matrix $A$ has a unique Nash Equilibrium. 
\end{proposition}
\begin{proof}
If $A$ does not have a PSNE, then due to Propositions \ref{prop:matrix:nosaddle} and \ref{prop:matrix:main}, we get that the matrix $A$ has a unique Nash Equilibrium which is not a PSNE.

Let $\Delta_{\min}>0$. Let us assume that $A$ has a PSNE. W.l.o.g. let the element $(1,1)$ be a PSNE. Due to the definition of PSNE, we have $a<b$ and $a>c$. If $d<b$, then $(1,1)$ is the unique Nash Equilibrium of $A$ as $a<b$ and strategy $1$ strictly dominates strategy $2$ for the row player. Similarly if $d>c$, then $(1,1)$ is the unique Nash Equilibrium of $A$ as $a>c$ and strategy $1$ strictly dominates strategy $2$ for the column player. The final case $d<c$ and $d>b$ is not possible otherwise we would have $d>b>a>c$ which is contradictory. Hence $A$ has a unique Nash Equilibrium.
\end{proof}

\begin{proposition}
The matrix game on $A=[a,b;c,d]$ has a unique Nash equilibrium which is not a PSNE if and only if the matrix game on $A$ does not have a PSNE.
\end{proposition}
\begin{proof}
If the matrix game on $A=[a,b;c,d]$ has a unique Nash equilibrium which is not a PSNE, then due to Proposition \ref{prop:matrix:main} either $a<b$, $a<c$, $d<b$ and $d<c$ or $a>b$, $a>c$, $d>b$ and $d>c$. This implies that $A$ does not have a PSNE.

If $A$ does not have a PSNE, then due to Proposition \ref{prop:matrix:nosaddle} either $a<b$, $a<c$, $d<b$ and $d<c$ or $a>b$, $a>c$, $d>b$ and $d>c$. This along with Proposition \ref{prop:matrix:main} implies that the matrix game on $A$ has a unique Nash equilibrium which is not a PSNE.
\end{proof}
\begin{proposition}
Any $\varepsilon$-Nash equilibrium of a matrix game $A$ is also an $\varepsilon$-good solution of the matrix game $A$.
\end{proposition}
\begin{proof}
Let $(x^*,y^*)$ be a Nash equilibrium of $A$ and $(x,y)$ be an $\varepsilon$-Nash equilibrium of $A$. Recall that $V_A^*=\langle x^*, Ay^*\rangle$. Now we have the following:
\begin{align*}
    \langle x^*, Ay^*\rangle &\geq \langle x, Ay^*\rangle \tag{as $(x^*,y^*)$ is a Nash equilibrium }\\
    & \geq \langle x,Ay\rangle-\varepsilon \tag{as $(x,y)$ is an $\varepsilon$-Nash equilibrium }
\end{align*}

Similarly we have the following:
\begin{align*}
    \langle x^*, Ay^*\rangle &\leq \langle x^*, Ay\rangle \tag{as $(x^*,y^*)$ is a Nash equilibrium }\\
    & \leq \langle x,Ay\rangle+\varepsilon \tag{as $(x,y)$ is an $\varepsilon$-Nash equilibrium }
\end{align*}
Hence $(x,y)$ is also an $\varepsilon$-good solution. 
\end{proof}

\section{Minimax sample complexity, Proof of Lemma~\ref{lem:trivial:NE}}\label{appendix:minimax}

\begin{proof}
First note that
\begin{align*}
    \P\left( \bigcup_{i=1}^m \bigcup_{j=1}^n  \{ |\bar{A}_{ij} - A_{ij}| \geq \varepsilon/2 \} \right) \leq \sum_{i=1}^m \sum_{j=1}^n \P( |\bar{A}_{ij} - A_{ij}| \geq \varepsilon/2 ) \leq \sum_{i=1}^m \sum_{j=1}^n \frac{\delta}{ mn} = \delta
\end{align*}
where the last inequality follows from a sub-Gaussian tail bound on our $1$-sub-Gaussian observations. The sub-Gaussian tail bound, also known as Hoeffding bound, is as follows.
\begin{lemma}[sub-Gaussian tail bound]
Let $X_1,X_2,\ldots,X_n$ be i.i.d samples  from a $1$-sub-Gaussian distribution with mean $\mu$. Then we have the following:
\begin{equation*}
    \mathbb{P}\left[\left|\frac{1}{n}\cdot\sum_{i=1}^nX_i-\mu\right|\geq \sqrt{\frac{2\log(2/\delta)}{n}}\right]\leq \delta
\end{equation*}
\end{lemma}

Thus, in what follows assume $|\bar{A}_{ij} - A_{ij}| \leq \varepsilon/2$ for all $(i,j) \in [m]\times [n]$.
For any $x'\in \simplex_m$ we have
\begin{align*}
\langle x, Ay \rangle &= \langle x, \bar{A} y \rangle + \sum_{i,j} (A_{ij}-\bar{A}_{i,j}) x_i y_j \\
&\geq \langle x,\bar Ay \rangle -\frac{\varepsilon}{2}  \\
& \geq  \langle x',\bar Ay \rangle-\frac{\varepsilon}{2} \tag{as $(x,y)$ is a NE of $\bar A$}\\
& \geq  \langle x', Ay \rangle - \varepsilon 
\end{align*}
Similarly, for any $y'\in \simplex_n$ we have
\begin{align*}
\langle x, Ay \rangle &\leq \langle x,\bar Ay \rangle +\frac{\varepsilon}{2}\\
& \leq  \langle x,\bar Ay' \rangle+\frac{\varepsilon}{2} \tag{as $(x,y)$ is a NE of $\bar A$}\\
& \leq  \langle x, Ay' \rangle + \varepsilon 
\end{align*}
which completes the proof.
\end{proof}

\section{Proof of $\varepsilon$-good solution Upper Bound}\label{appendix:thm3}
We establish the sample complexity and the correctness of the Algorithm \ref{alg-ucb-1} by proving the Theorem \ref{thm:alg1}.

\begin{proof}[Proof of Theorem \ref{thm:alg1}]
Let $\bar A_{ij,t}$ denote the empirical mean of $A_{ij}$ at time step $t$. Let us begin by defining two events:
\begin{align*}
    G:=&\bigcap_{t=1}^T\bigcap_{i=1}^2 \bigcap_{j=1}^2 \{ |A_{ij}-\bar A_{ij,t}|\leq \sqrt{\tfrac{2\log({16T}/{\delta})}{t}} \} \\
    E:=&\bigcap_{i=1}^2 \bigcap_{j=1}^2 \{ |A_{ij}-\bar A_{ij,T}|\leq \sqrt{\tfrac{2\log({16}/{\delta})}{T}} \} 
\end{align*}
A union bound and sub-Gaussian-tail bound demonstrates that $\P( G^c \cup E^c ) \leq \P(G^c) + \P(E^c) \leq \delta$.
Consequently, events $E$ and $G$ hold simultaneously with probability at least $1-\delta$, so in what follows, assume they hold. 

If $A$ has a PSNE and if the condition of line \ref{alg1:con1} of  Algorithm~\ref{alg-ucb-1} is satisfied, then we identify an $\varepsilon$-good solution in $\frac{800\log (\frac{16T}{ \delta})}{ \Delta_{\min}^2}$ time steps due to Lemma \ref{lem:alg1:tmin} and Corollary \ref{cor:alg1:saddle}. On the other hand, if $A$ has a PSNE but the for loop completes after $t=T$ iterations,  then we identify an $\varepsilon$-good solution in $T=\frac{8 \log (16/\delta) }{\varepsilon^2}$ time steps due to Lemma \ref{lem:trivial:NE}. Note that in this case, $T<\frac{800\log (\frac{16T}{ \delta})}{ \Delta_{\min}^2}$ due to Lemma \ref{lem:alg1:tmin}. Hence, if $A$ has a PSNE, we identify an $\varepsilon$-good solution in $O\left(\min\left\{\frac{\log(1/\delta)}{\varepsilon^2},\frac{\log(T/\delta)}{\Delta_{\min}^2}\right\}\right)$ time steps.

Let us now assume for the rest of the proof that $A$ has a unique NE which is not a PSNE. 
If the condition of line \ref{alg1:con2} of Algorithm~\ref{alg-ucb-1} is satisfied, then we identify an $\varepsilon$-good solution in $T=\frac{8 \log (16/\delta) }{\varepsilon^2}$ time steps due to Lemma \ref{lem:trivial:NE}. 
Now observe that in this case $T=O\left(\min\left\{\frac{\log(1/\delta)}{\varepsilon^2},\frac{\log(T/\delta)}{\varepsilon|D|}\right\}\right)$ due to Lemma \ref{lem:alg1:D<eps}. On the other hand, if the for loop completes after $t=T$ iterations,  then we identify an $\varepsilon$-good solution in $T=\frac{8 \log (16/\delta) }{\varepsilon^2}$ time steps due to Lemma \ref{lem:trivial:NE}. Note that in this case, $T<\frac{800\log (\frac{16T}{ \delta})}{ \Delta_{\min}^2}$ due to Lemma \ref{lem:alg1:tmin}.

Now let us  assume for the rest of the proof that the condition in the line \ref{alg1:con3} is satisfied. Then due to Lemma \ref{lem:alg1:D:ratio}, we have $\frac{64 \log (\frac{16}{\delta}) }{\varepsilon |D| } \leq N \leq \frac{96 \log (\frac{16}{\delta}) }{\varepsilon |D| } $. If the condition in the line \ref{alg1:con4} is satisfied, then we identify an $\varepsilon$-good solution in $T=\frac{8 \log (16/\delta) }{\varepsilon^2}$ time steps due to Lemma \ref{lem:trivial:NE}. Now observe that in this case $T=O\left(\min\left\{\frac{\log(1/\delta)}{\varepsilon^2},\max\left\{\frac{\log(T/\delta)}{\Delta_{\min}^2},\frac{\log(T/\delta)}{\varepsilon |D|}\right\}\right\}\right)$ as $T< \frac{800\log (\frac{16T}{ \delta})}{ \Delta_{\min}^2}+N$. If the condition in the line \ref{alg1:con4} is not satisfied, then we identify an $\varepsilon$-good solution due to Lemma \ref{lem:alg1:goodsol}. In this case, let the number of times we are required to sample each element be $n_0$. Then $n_0\leq T$ and $n_0\leq \frac{800\log (\frac{16T}{ \delta})}{ \Delta_{\min}^2}+ \frac{96 \log (\frac{16}{\delta}) }{\varepsilon |D| } $. Hence, $n_0=O\left(\min\left\{\frac{\log(1/\delta)}{\varepsilon^2},\max\left\{\frac{\log(T/\delta)}{\Delta_{\min}^2},\frac{\log(T/\delta)}{\varepsilon |D|}\right\}\right\}\right)$.
\end{proof}

\subsection{Consequential lemmas of Algorithm~\ref{alg-ucb-1}'s conditional statements}
Recall the definitions of events $E$ and $G$.
We first present a few lemmas that deal with empirical estimates and instance dependent parameters like $\tilde \Delta_{\min},\tilde D, \Delta_{\min}$ and $|D|$. Whenever we fix a time step $t\leq T$ and discuss the parameters like $\tilde  \Delta_{\min}, \tilde D$ and $\Delta$, we consider those values that have been assigned to these parameters during the time step $t$.
We begin with upper bounding $ |\Delta_{\min}-\tilde \Delta_{\min}|$ in the following lemma.
\begin{lemma}
Fix a time step $t\leq T$. If the event $G$ holds, then we have the following:
\begin{equation*}
    |\Delta_{\min}-\tilde \Delta_{\min}|\leq 2\Delta
\end{equation*}
\end{lemma}
\begin{proof}
Let us assume that the event $G$ holds. Then for every element $(i,j)$, we have  $|A_{ij}-\bar A_{ij}|\leq \Delta$. Then we have $\left||A_{ij}-A_{i'j'}|-|\bar A_{ij}-\bar A_{i'j'}|\right|\leq 2\Delta$ for any $i,j,i',j'$. By repeatedly applying the Lemma \ref{lem:deviation}, we get $|\Delta_{\min}-\tilde \Delta_{\min}|\leq 2\Delta$.
\end{proof}

The following lemma upper bounds the number of time steps required to satisfy the condition $1\leq \frac{\tilde \Delta_{\min}+2\Delta}{\tilde \Delta_{\min}-2\Delta}\leq \frac{3}{2}$.
\begin{lemma}\label{lem:alg1:tmin}
Let $t$ be the time step when the condition $1\leq \frac{\tilde \Delta_{\min}+2\Delta}{\tilde \Delta_{\min}-2\Delta}\leq \frac{3}{2}$ holds true for the first time. If the event $G$ holds, then $t\leq \frac{800\log (\frac{16T}{ \delta})}{ \Delta_{\min}^2}$.
\end{lemma}
\begin{proof}
Consider the time step $t= \frac{800\log (\frac{16T}{ \delta})}{ \Delta_{\min}^2}$. Let us assume that the event $G$ holds. Then for every element $(i,j)$, we have  $|A_{ij}-\bar A_{ij}|\leq \Delta=\sqrt{\frac{2\log(\frac{16T}{\delta})}{t}}=\frac{\Delta_{\min}}{20}$. Now observe that $\tilde \Delta_{\min}+2\Delta\leq \Delta_{\min}+4\Delta = \frac{6\Delta_{\min}}{5}$. Similarly, we have $\tilde \Delta_{\min}-2\Delta\geq \Delta_{\min}-4\Delta \geq \frac{4\Delta_{\min}}{5}$. Hence, we have $1\leq \frac{\tilde \Delta_{\min}+2\Delta}{\tilde \Delta_{\min}-2\Delta}\leq \frac{3}{2}$. 

\end{proof}

The following lemma bounds the ratio $\frac{\tilde \Delta_{\min}}{\Delta_{\min}}$.
\begin{lemma}\label{lem:alg1:dmin:ratio}
Let $t$ be the time step when the condition $1\leq \frac{\tilde \Delta_{\min}+2\Delta}{\tilde \Delta_{\min}-2\Delta}\leq \frac{3}{2}$ holds true for the first time. If the event $G$ holds, then $\frac{5}{6}\leq \frac{\tilde \Delta_{\min}}{\Delta_{\min}} \leq \frac{5}{4}$ at the time step $t$.
\end{lemma}
\begin{proof}
Let us assume that the event $G$ holds. Then for every element $(i,j)$, we have  $|A_{ij}-\bar A_{ij}|\leq \Delta$. As $\frac{\tilde \Delta_{\min}+2\Delta}{\tilde \Delta_{\min}-2\Delta}\leq \frac{3}{2}$, we have $\Delta\leq \frac{\tilde \Delta_{\min}}{10}$.
Now observe that $\frac{\tilde \Delta_{\min}}{\Delta_{\min}}\leq \frac{\tilde\Delta_{\min}}{\tilde\Delta_{\min}-2\Delta}\leq \frac{\tilde\Delta_{\min}}{4\tilde\Delta_{\min}/5}= \frac{5}{4}$. Next observe that $\frac{\tilde \Delta_{\min}}{\Delta_{\min}}\geq \frac{\tilde \Delta_{\min}}{\tilde \Delta_{\min}+2\Delta}\geq\frac{\tilde \Delta_{\min}}{6\tilde \Delta_{\min}/5}= \frac{5}{6}$.
\end{proof}

\noindent
The following lemma and the subsequent corollary relates the empirical matrix $\bar A$ to the input matrix $A$.
\begin{lemma}\label{lem:alg1:gap}
Let $t$ be the time step when the condition $1\leq \frac{\tilde \Delta_{\min}+2\Delta}{\tilde \Delta_{\min}-2\Delta}\leq \frac{3}{2}$ holds true for the first time. If the event $G$ holds, then at any time step $t_0$ such that  $t\leq t_0\leq T$, we have the following:
\begin{itemize}
    \item If $A_{ij_1}>A_{ij_2}$, then $\bar A_{ij_1}>\bar A_{ij_2}$ 
    \item If $A_{i_1j}>A_{i_2j}$, then $\bar A_{i_1j}>\bar A_{i_2j}$
    \item If $\bar A_{ij_1}>\bar A_{ij_2}$, then $ A_{ij_1}> A_{ij_2}$ 
    \item If $\bar A_{i_1j}>\bar A_{i_2j}$, then $ A_{i_1j}> A_{i_2j}$
\end{itemize}
\end{lemma}
\begin{proof}
As $\frac{\tilde \Delta_{\min}+2\Delta}{\tilde \Delta_{\min}-2\Delta}\leq \frac{3}{2}$, we have $\Delta\leq \frac{\tilde \Delta_{\min}}{10}$. Due to Lemma \ref{lem:alg1:dmin:ratio}, we have $\Delta\leq \frac{\Delta_{\min}}{8}$. As event $G$ holds, for any element $(i,j)$, we have $|A_{ij}-\bar A_{ij}|\leq \sqrt{\frac{2\log(\frac{16T}{\delta})}{t_0}}\leq \Delta $.

If $A_{ij_1}>A_{ij_2}$, we have the following:
\begin{align*}
    \bar A_{ij_1} & \geq A_{ij_1}-\Delta \\
    & \geq A_{ij_2}+\Delta_{\min}-\Delta \tag{as $A_{ij_1}- A_{ij_2}\geq \Delta_{\min}$}\\
    & > A_{ij_2} + \Delta \tag{as $\Delta\leq \frac{ \Delta_{\min}}{8}$}\\
    & \geq \bar A_{ij_2} \tag{as event $G$ holds}\\
\end{align*}

If $A_{i_1j}>A_{i_2j}$, we have the following:
\begin{align*}
    \bar A_{i_1j} & \geq A_{i_1j}-\Delta \\
    & \geq A_{i_2j}+\Delta_{\min}-\Delta \tag{as $A_{i_1j}- A_{i_2j}\geq \Delta_{\min}$}\\
    & > A_{i_2j} + \Delta \tag{as $\Delta\leq \frac{ \Delta_{\min}}{8}$}\\
    & \geq \bar A_{i_2j} \tag{as event $G$ holds}\\
\end{align*}

If $\bar A_{ij_1}>\bar A_{ij_2}$, we have the following:
\begin{align*}
    A_{ij_1} & \geq \bar A_{ij_1}-\Delta \\
    & \geq \bar A_{ij_2}+\tilde\Delta_{\min}-\Delta \tag{as $\bar A_{ij_1}-\bar A_{ij_2}\geq \tilde\Delta_{\min}$}\\
    & >\bar A_{ij_2} + \Delta \tag{as $\Delta\leq \frac{\tilde \Delta_{\min}}{10}$}\\
    & \geq A_{ij_2} \tag{as event $G$ holds}\\
\end{align*}

If $\bar A_{i_1j}>\bar A_{i_1j}$, we have the following:
\begin{align*}
    A_{i_1j} & \geq \bar A_{i_1j}-\Delta \\
    & \geq \bar A_{i_2j}+\tilde\Delta_{\min}-\Delta \tag{as $\bar A_{i_1j}-\bar A_{i_2j}\geq \tilde\Delta_{\min}$}\\
    & >\bar A_{i_2j} + \Delta \tag{as $\Delta\leq \frac{\tilde \Delta_{\min}}{10}$}\\
    & \geq A_{i_2j} \tag{as event $G$ holds}\\
\end{align*}
\end{proof}
\begin{corollary}\label{cor:alg1:saddle}
Let $t$ be the time step when the condition $1\leq \frac{\tilde \Delta_{\min}+2\Delta}{\tilde \Delta_{\min}-2\Delta}\leq \frac{3}{2}$ holds true for the first time. If the event $G$ holds, then at any time step $t_0$ such that  $t\leq t_0\leq T$, we have the following:
\begin{itemize}
    \item $(i,j)$ is PSNE of $A$ if and only if $(i,j)$ is a PSNE of $\bar A$.    
    \item $A$ does not have a PSNE if and only if $\bar A$ does not have a PSNE.
\end{itemize}
\end{corollary}
The following lemma bounds the ratio $\frac{\tilde D}{|D|}$.
\begin{lemma}\label{lem:alg1:D:ratio}
Let $t$ be the time step when the condition $1\leq \frac{\tilde \Delta_{\min}+2\Delta}{\tilde \Delta_{\min}-2\Delta}\leq \frac{3}{2}$ holds true for the first time. If the event $G$ holds and $A$ has a unique Equilibrium which is not a PSNE, then $\frac{5}{6}\leq \frac{\tilde D}{|D|} \leq \frac{5}{4}$ at the time step $t$.
\end{lemma}
\begin{proof}
Let us assume that the event $G$ holds. Then for every element $(i,j)$, we have  $|A_{ij}-\bar A_{ij}|\leq \Delta$. As $\frac{\tilde \Delta_{\min}+2\Delta}{\tilde \Delta_{\min}-2\Delta}\leq \frac{3}{2}$ and $2\tilde\Delta_{\min}\leq \tilde D$, we have $\Delta\leq \frac{\tilde\Delta_{\min}}{10}\leq \frac{\tilde D}{20}$.
Now observe that $\frac{\tilde D}{|D|}\leq \frac{\tilde D}{\tilde D-4\Delta}\leq \frac{\tilde D}{4\tilde D/5}= \frac{5}{4}$. Next observe that $\frac{\tilde D}{|D|}\geq \frac{\tilde D}{\tilde D+4\Delta}\geq\frac{\tilde D}{6\tilde D/5}= \frac{5}{6}$.
\end{proof}

The following two lemmas bound $|D|$ when certain conditions in the algorithm \ref{alg-ucb-1} hold true.
\begin{lemma}\label{lem:alg1:D<eps}
If the condition in the line \ref{alg1:con2} of the algorithm \ref{alg-ucb-1} holds true and event $G$ holds, then $|D|<12 \varepsilon$
\end{lemma}
\begin{proof}
Due to Lemma \ref{lem:alg1:D:ratio}, we have $|D|\leq \frac{6\tilde D}{5}<12\varepsilon$.
\end{proof}

\begin{lemma}\label{lem:alg1:D>eps}
If the condition in the line \ref{alg1:con3} of the algorithm \ref{alg-ucb-1} holds true and event $G$ holds, then $|D|\geq 8 \varepsilon$
\end{lemma}
\begin{proof}
Due to Lemma \ref{lem:alg1:D:ratio}, we have $|D|\geq \frac{4\tilde D}{5}\geq 8\varepsilon$.
\end{proof}
For any $A\in \mathbb{R}^{m\times n}$, let $A(\Delta):=\{B\in \mathbb{R}^{n\times n}:\max_{i,j}|(B-A)_{ij}|\leq \Delta\}$.
We now present the main lemma that establishes the correctness of the algorithm \ref{alg-ucb-1} when the input matrix $A$ does not have a PSNE.
\begin{lemma}\label{lem:alg1:goodsol}
If the condition in the line \ref{alg1:con3} of the algorithm \ref{alg-ucb-1} holds true and event $G$ holds, then Nash equilibrium of the empirical matrix $\bar A$ is also an $\varepsilon$-good solution of the input matrix $A$.
\end{lemma}
\begin{proof}
$\bar A\in A\left(\sqrt{\frac{\varepsilon \tilde D}{40}}\right)$ as event $G$ holds true. Due to Lemma \ref{lem:alg1:D:ratio}, we have $\tilde D\leq \frac{5|D|}{4}$. Hence, $\bar A\in A\left(\sqrt{\frac{\varepsilon |D|}{32}}\right)$ which in turn implies that $A\in \bar A\left(\sqrt{\frac{\varepsilon |D|}{32}}\right)$. Now we show that $\frac{\sqrt{\varepsilon |D|}}{4\sqrt 2}< \frac{|\bar A_{11}-\bar A_{12}-\bar A_{21}+\bar A_{22}|}{12}$.
\begin{align*}
     \frac{|\bar A_{11}-\bar A_{12}-\bar A_{21}+\bar A_{22}|}{12}&\geq \frac{| A_{11}- A_{12}- A_{21}+ A_{22}|-4\cdot \sqrt{\frac{\varepsilon |D|}{32}}}{12} \tag{as $\bar A\in A\left(\sqrt{\frac{\varepsilon |D|}{32}}\right)$}\\
     &= \frac{|D|-\sqrt{\varepsilon|D|/2}}{12}\\
     &\geq \frac{\sqrt{8\varepsilon|D|}-\sqrt{\varepsilon|D|/2}}{12} \tag{as $|D|\geq 8\varepsilon$}\\
     & >\frac{\sqrt{\varepsilon |D|}}{4\sqrt 2}
\end{align*}
Let $(x,y)$ be the Nash equilibrium of $\bar A$. Now by applying Lemma \ref{2matrix-1} we have that $|V_{A}^*-\langle x, A y\rangle|\leq \varepsilon$. We can apply lemma \ref{2matrix-1} as $A\in \bar A\left(\frac{\sqrt{\varepsilon |D|}}{4\sqrt 2}\right)$, $\frac{\sqrt{\varepsilon |D|}}{4\sqrt 2}< \frac{|\bar A_{11}-\bar A_{12}-\bar A_{21}+\bar A_{22}|}{12}$ and $\bar A$ has a unique Nash equilibrium which is not a PSNE (due to Corollary \ref{cor:alg1:saddle}).
\end{proof}

\subsection{Technical Lemmas for Upper Bound}
In this section, we present few technical lemmas that are used to establish the upper bound on the sample complexity of finding $\varepsilon$-good solution.

Let $A\in \mathbb{R}^{n\times n}$. Recall that $A(\Delta):=\{B\in \mathbb{R}^{n\times n}:\max_{i,j}|(B-A)_{ij}|\leq \Delta\}$ and for any $n$-dimensional vector $v$, $v(i)$ denotes its $i$-th component. Now we present the following lemma, where we relate $V_A^*$, the Nash equilibrium of $A$, to $V_{B}^*$ where $B\in A(\Delta)$.
\begin{lemma}\label{prop:st1}
Consider a matrix $A\in \mathbb{R}^{n\times n}$ with unique Nash equilibrium $(x^*,y^*)$ which is not a PSNE. Then for any $B\in A(\Delta)$ that has a unique Nash equilibrium $(x,y)$ which is not a PSNE, we have the following:
\begin{equation*}
   V_B^*=\langle x^*, By^* \rangle+\sum_{j=1}^n y^*(j)\sum_{i=1}^n\theta_{i}\Delta_{ij} 
\end{equation*}
where $\Delta_{ij}:=B_{ij}-A_{ij}$ and $\theta_i=x(i)-x^*(i)$.
\end{lemma}
\begin{proof}
For all $i\in[n]$, let $x_i=x(i)$, $y_i=y(i)$, $x^*_i=x^*(i)$ and $y^*_i=y^*(i)$. 

First observe that $\langle x^*, By^* \rangle=V_A^*+\sum_{i,j}x^*_iy^*_j\Delta_{ij}$. Also observe that $\sum_{i=1}^n\theta_i=0$. Let $B_j$ denote the $j$-th column of $B$. Let $V_j:=\langle x, B_j \rangle$. Now we have the following:
\begin{align*}
    \langle x, B_j \rangle&= \sum_{i=1}^n [x_i^*A_{ij}+x_i^*\Delta_{ij}+\theta_iA_{ij}+\theta_{i}\Delta_{ij}]\\
    &=V_A^*+\sum_{i=1}^nx_i^*\Delta_{ij}+\sum_{i=1}^n\theta_iA_{ij}+\sum_{i=1}^n\theta_{i}\Delta_{ij} \tag{ as $\sum_{i=1}^n x_i^*A_{ij}=V_A^*$}
\end{align*}

Let $V=(V_1,\ldots,V_n)$. Since $\supp(x)=\supp(y)=[n]$, therefore we have for all $j\in[n]$, $V_j=V_B^*$. Now we have the following:
\begin{align}
    V_B^*&=\langle V, y^*\rangle \nonumber \\
    &= \sum_{j=1}^n y_j^*V_A^*+\sum_{j=1}^n y_j^*\sum_{i=1}^nx_i^*\Delta_{ij}+\sum_{j=1}^n y_j^*\sum_{i=1}^n\theta_iA_{ij}+\sum_{j=1}^n y_j^*\sum_{i=1}^n\theta_{i}\Delta_{ij}\nonumber \\
    &= V_A^*+ \sum_{i=1}^n\sum_{j=1}^nx_i^*y_j^*\Delta_{ij}+\sum_{i=1}^n\theta_i\sum_{j=1}^n y_j^*A_{ij}+\sum_{j=1}^n y_j^*\sum_{i=1}^n\theta_{i}\Delta_{ij}\nonumber \\
    &= \langle x^*, By^* \rangle+V_A^*\sum_{i=1}^n\theta_i+\sum_{j=1}^n y_j^*\sum_{i=1}^n\theta_{i}\Delta_{ij}\tag{ as $\sum_{j=1}^n y_j^*A_{ij}=V_A^*$}\\
    &= \langle x^*, By^* \rangle+\sum_{j=1}^n y_j^*\sum_{i=1}^n\theta_{i}\Delta_{ij}\nonumber
\end{align}
\end{proof}

Let us define two matrices $A_1$ and $A_2$ as follows:
\[
A_1 = \begin{bmatrix} 
   a & b \\
   c & d \\
    \end{bmatrix}\]
    \[
A_2 = \begin{bmatrix} 
   a+\Delta_{11} & b+\Delta_{12} \\
   c+\Delta_{21} & d+\Delta_{22} \\
    \end{bmatrix}\]
Let $\Delta=\max_{i,j}|\Delta_{ij}|$. Now we present the following lemma, where we upper bound $|V_{A_2}^*-\langle x^*, A_2 y^*\rangle|$ where $(x^*,y^*)$ is the Nash equilibrium of $A_1$.
\begin{lemma}\label{2matrix-1}
Let $A_1$ and $A_2$ have a unique NE which is not a PSNE. Let $(x^*,y^*)$ be the NE of the matrix game $A_1$. Let $\Delta\leq |a-b-c+d|/12$. Then we have the following:
\begin{equation*}
    |V_{A_2}^*-\langle x^*, A_2 y^*\rangle|\leq \frac{16\Delta^2}{|D|}\leq \frac{32\Delta^2}{|D'|} 
\end{equation*}
where $D:=a-b-c+d$, $D':=a-b-c+d+\Delta_{11}-\Delta_{12}+\Delta_{22}-\Delta_{21}$.
\end{lemma}    
\begin{proof}
Let $(x^*,y^*)=((x,1-x),(y,1-y))$ be the NE of the matrix game $A_1$ where $x=\frac{d-c}{a-b-c+d}$ and $y=\frac{d-b}{a-b-c+d}$. Let $((x',1-x'),(y',1-y'))$ be the NE of the matrix game $A_2$ where $x'=\frac{d-c+\Delta_{22}-\Delta_{21}}{a-b-c+d+\Delta_{11}-\Delta_{12}+\Delta_{22}-\Delta_{21}}$ and $y'=\frac{d-b+\Delta_{22}-\Delta_{12}}{a-b-c+d+\Delta_{11}-\Delta_{12}+\Delta_{22}-\Delta_{21}}$ For convenience, let $N:=d-c$ and $D:=a-b-c+d$. Hence, we have $x=\frac{|N|}{|D|}$ and $x'=\frac{N+\Delta_{22}-\Delta_{21}}{D+\Delta_{11}-\Delta_{12}+\Delta_{22}-\Delta_{21}}$. Now we will upper bound $x'$ as follows:
\begin{align*}
x'&\leq \frac{|N|+2\Delta}{|D|-4\Delta}\\
&=\left(\frac{|N|}{|D|}+\frac{2\Delta}{|D|}\right)\left(1-\frac{4\Delta}{|D|}\right)^{-1}\\
& \leq \left(\frac{|N|}{|D|}+\frac{2\Delta}{|D|}\right)\left(1+\frac{6\Delta}{|D|}\right)\tag{as $\frac{1}{1-z}\leq1+\frac{3z}{2}$ when $0\leq z\leq\frac{1}{3}$}\\
& = \frac{|N|}{|D|}+\frac{2\Delta}{|D|} + \frac{|N|}{|D|}\cdot\frac{6\Delta}{|D|} + \frac{12\Delta^2}{|D|^2}\\
& \leq \frac{|N|}{|D|}+\frac{9\Delta}{D}\tag{as $\frac{|N|}{|D|}\leq 1$ and $\frac{\Delta}{|D|}\leq \frac{1}{12}$}
\end{align*}
Next we will lower $x'$ as follows:
\begin{align*}
x'&\geq \frac{|N|-2\Delta}{|D|+4\Delta}\\
&=\left(\frac{|N|}{|D|}-\frac{2\Delta}{|D|}\right)\left(1+\frac{4\Delta}{|D|}\right)^{-1}\\
& \geq \left(\frac{|N|}{|D|}-\frac{2\Delta}{|D|}\right)\left(1-\frac{4\Delta}{|D|}\right)\tag{as $\frac{1}{1+z}>1-z$ when $z>0$}\\
& = \frac{|N|}{|D|}-\frac{2\Delta}{|D|} - \frac{|N|}{|D|}\cdot\frac{4\Delta}{|D|} + \frac{8\Delta^2}{|D|^2}\\
& \geq \frac{|N|}{|D|}-\frac{6\Delta}{|D|} \tag{as $\frac{|N|}{|D|}\leq 1$}
\end{align*}
Hence we have $|x'-x|\leq \frac{8\Delta}{D}$. Due to Lemma \ref{prop:st1}, we have the following:
\begin{align*}
    |V_B^*-\langle x^*, By^* \rangle|&\leq \sum_{j=1}^2 y^*(j)\sum_{i=1}^2|\theta_{i}\Delta_{ij}|\\
    &\leq \sum_{j=1}^2 y^*(j)\sum_{i=1}^2\frac{8\Delta^2}{|D|}\\
    &= \frac{18\Delta^2}{|D|}\sum_{j=1}^2 y^*(j)\\
    &= \frac{18\Delta^2}{|D|}\\
    & \leq \frac{32\Delta^2}{|D'|} \tag{as $|D'|\geq |D|-4\Delta\geq2|D|/3$}
\end{align*}
\end{proof}

The next lemma states some basic inequalities that will be used frequently in the analysis that follows. 
\begin{lemma}\label{lem:deviation}
Let $a,\bar a,b,\bar b,\Delta'$ be positive real numbers. Let $|a-\bar a|\leq \Delta'$ and $|b-\bar b|\leq \Delta'$. Then we have the following:
\begin{itemize}
    \item $|a+b-(\bar a+\bar b)|\leq 2\Delta'$
    \item $|\min\{a,b\}-\min\{\bar a,\bar b\}|\leq \Delta'$
    \item $|\max\{a,b\}-\max\{\bar a,\bar b\}|\leq \Delta'$
\end{itemize}
\begin{proof}
First observe that $|a-\bar a+b-\bar b|\leq |a-\bar a|+|b-\bar b|\leq 2\Delta'$. 

Next w.l.o.g let us assume that $\min\{a,b\}=a$. If $\min\{\bar a,\bar b\}=\bar a$, then we have $|\min\{a,b\}-\min\{\bar a,\bar b\}|=|a-\bar a|\leq \Delta'$. If $\min\{\bar a,\bar b\}=\bar b$, then $\bar b\leq \bar a\leq a+\Delta'$ and $\bar b\geq b-\Delta'\geq a-\Delta'$. Hence, in this case also we have $|\min\{a,b\}-\min\{\bar a,\bar b\}|=|a-\bar b|\leq \Delta'$.

Finally w.l.o.g let us assume that $\max\{a,b\}=b$. If $\max\{\bar a,\bar b\}=\bar b$, then we have $|\max\{a,b\}-\max\{\bar a,\bar b\}|=|b-\bar b|\leq \Delta'$. If $\max\{\bar a,\bar b\}=\bar a$, then $\bar a\geq \bar b\geq b-\Delta'$ and $\bar a\leq a+\Delta'\leq b+\Delta'$. Hence, in this case also we have $|\max\{a,b\}-\max\{\bar a,\bar b\}|=|b-\bar a|\leq \Delta'$.
\end{proof}
\end{lemma}

\section{Proof of $\varepsilon$-good solution lower bound with respect to $D$}\label{appendix:thm1}
Before finishing the proof of the theorem \ref{thm:lower1}, we begin with the proof of Lemma~\ref{low:lem1}
\begin{proof}
W.l.o.g let us assume that $D>0$. 
For $\square \in \{-\Delta,0,\Delta\}$ it can be shown that 
\begin{align*}
V_{A_\square}^* &=\frac{ad-bc}{D}+\frac{d-a}{D}\square-\frac{\square^2}{D}\\
\langle x',A_\square y'\rangle &=\frac{ad-bc}{D}+\frac{d-a}{D}\square+\frac{\alpha \beta+(\alpha+\beta)\square}{D}
\end{align*}
where $x'=(\frac{d-c+\alpha}{D},\frac{a-b-\alpha}{D})\in \simplex_2$ and $y'=(\frac{d-b+\beta}{D},\frac{a-c-\beta}{D})\in \simplex_2$ for any $\alpha \in [c-d, a-b]$ and $\beta \in [b-d ,a-c]$. 
Note that this parameterization ensures the range of $x',y'$ is equal to $\simplex_2$. We refer the reader to the Appendix \ref{apendix:low:lem1} for the detailed calculations.

We will now show that regardless of what values $(\alpha,\beta)$ take (equivalently, regardless of what values $(x',y')$ take), there is at least one of the three alternative matrices has error of at least $3\varepsilon/2$.   
If $|\alpha \beta| \geq \Delta^2/2$, then $|V_{A_0}^*-\langle x', A_0y' \rangle|=\frac{|\alpha \beta|}{D}\geq \frac{\Delta^2}{2D}$. 
If $|\alpha \beta| < \Delta^2/2$ and $\alpha+\beta\geq 0$, then $\langle x',A_\Delta y'\rangle-V_{A_\Delta}^*=\frac{\Delta^2+\alpha \beta+(\alpha+\beta)\Delta}{D}\geq \frac{\Delta^2}{2D}$. Similarly, if $|\alpha \beta| < \Delta^2/2$ and if $\alpha + \beta< 0$, then $\langle x',A_{-\Delta} y'\rangle-V_{A_{-\Delta}}^*=\frac{\Delta^2+\alpha \beta-(\alpha+\beta)\Delta}{D}\geq \frac{\Delta^2}{2D}$. Hence, we proved that for any $(x',y') \in \simplex_2 \times \simplex_2$, there exists a matrix $B\in\{A_{-\Delta},A_0,A_{\Delta}\}$ such that
\begin{equation*}
    |V_B^*-\langle x', By' \rangle|\geq \frac{\Delta^2}{2D}=\frac{3\varepsilon}{2}
\end{equation*}



\end{proof}

\subsection{Proof of Theorem \ref{thm:lower1}}
Let $\nu_{i,j}^A = \mathcal{N}(A_{ij},1)$ be the distribution of an observation when playing pair $(i,j)$ with matrix $A$. 
Let $\P_A$ denote the probability law of the internal randomness of the algorithm and random observations.
If an algorithm is $(\varepsilon,\delta)$-PAC-good and outputs a solution $(\widehat{x},\widehat{y})$ then $\min_A \P_A(|V_A^*-\langle \widehat{x}, A \widehat{y} \rangle|\leq \varepsilon) \geq 1-\delta$.
We will show that if an algorithm is $(\varepsilon,\delta)$-PAC-good then it can also accomplish a particular hypothesis.
We will conclude by noting that any procedure that can accomplish the hypothesis test must take the claimed sample complexity. 

For any pair of mixed strategies $(\widehat{x},\widehat{y})$ output by the procedure at the stopping time $\tau$, define 
\begin{align*}
    \phi = \{A_{-\Delta},A_{0},A_{\Delta}\} \setminus \arg\max_{B \in \{A_{-\Delta},A_{0},A_{\Delta}\}} |V_{B}^*-\langle \widehat{x}, B \widehat{y} \rangle|,
\end{align*}
breaking ties arbitrarily in the maximum so that $\phi \in  \{A_{-\Delta},A_{0}\}\cup \{A_{-\Delta},A_{\Delta}\}\cup \{A_{0},A_{\Delta}\}$.
Note that
\begin{align}\label{eqn:correctness_delta}
    \P_{A_0}( A_0 \in \phi ) \geq \P_{A_0}( A_0 \in \phi,  |V_{A_0}^*-\langle \widehat{x},A_0 \widehat{y} \rangle| \leq \varepsilon) = \P_{A_0}(   |V_{A_0}^*-\langle \widehat{x}, A_0 \widehat{y} \rangle| \leq \varepsilon) \geq 1-\delta
\end{align}
where the equality follows from the Lemma~\ref{low:lem1}: at least one of the three matrices must have a loss of at least $3 \varepsilon /2$, but $A_0$ has a loss of at most $\varepsilon$, thus $A_0 \in \phi$.
Now because
\begin{align*}
    2 \max\{ \P_{A_0}( \phi = \{A_0,A_{-\Delta}\} ) , \P_{A_0}( \phi = \{A_0,A_{\Delta}\} ) \} &\geq \P_{A_0}( \phi = \{A_0,A_{-\Delta}\} ) + \P_{A_0}( \phi = \{A_0,A_{\Delta}\} ) \\
    &= \P_{A_0}( A_0 \in \phi ) \geq 1-\delta
\end{align*}
we have that $\P_{A_0}( \phi = \{A_0,A_{-\Delta}\} ) \geq \frac{1-\delta}{2}$ or $\P_{A_0}( \phi = \{A_0,A_{\Delta} \}) \geq \frac{1-\delta}{2}$.
Let's assume the former (the latter case is handled identically).
By the same argument as \eqref{eqn:correctness_delta} we have that $\P_{A_\Delta}( \phi = \{A_0,A_{-\Delta}\} ) \leq \delta$.

For a stopping time $\tau$, let $N_{i,j}(\tau)$ denote the number of times $(i,j)$ is sampled. Recalling that $\nu_{i,j}^A = \mathcal{N}(A_{ij},1)$, we have by Lemma 1 of \citet{kaufmann2016complexity} that
\begin{align*}
\mathbb{E}_{A_0}[ N_{1,1}(\tau) ] KL( \nu_{1,1}^{A_0}, \nu_{1,1}^{A_\Delta} ) + \mathbb{E}_{A_0}[ N_{2,2}(\tau) ] KL( \nu_{2,2}^{A_0}, \nu_{2,2}^{A_\Delta} ) \geq d( \P_{A_0}( \phi = \{A_0,A_{-\Delta}\} ), \P_{A_\Delta}( \phi = \{A_0,A_{-\Delta}\} ) )
\end{align*}
where $KL( \nu_{1,1}^{A_0}, \nu_{1,1}^{A_\Delta} ) = KL( \nu_{2,2}^{A_0}, \nu_{2,2}^{A_\Delta} ) = \Delta^2/2$ and $d(p,q) = p \log(\frac{p}{q}) + (1-p) \log(\frac{1-p}{1-q})$.
Since 
\begin{align*}
    d( \P_{A_0}( \phi = \{A_0,A_{-\Delta}\} ), \P_{A_\Delta}( \phi = \{A_0,A_{-\Delta}\} ) ) &\geq d( \tfrac{1-\delta}{2},\delta) \\
    &= \tfrac{1-\delta}{2} \log( \tfrac{1-\delta}{2 \delta} ) + \tfrac{1+\delta}{2} \log( \tfrac{1+\delta}{2(1-\delta)} ) \\
    &= \tfrac{1}{2} \log( \tfrac{1+\delta}{4 \delta} ) -  \tfrac{\delta}{2} \log( \tfrac{(1-\delta)^2}{ \delta (1+\delta)} ) \\
    &\geq \tfrac{1}{2} \log( \tfrac{1+\delta}{4 \delta} ) -  1/8 > \tfrac{1}{2} \log(1/30\delta)
\end{align*} 
and $\tau = N_{1,1}(\tau)+N_{1,2}(\tau)+N_{2,1}(\tau)+N_{2,2}(\tau)$ we conclude that
\begin{align*}
    \mathbb{E}_{A_0}[ \tau ] \geq \frac{ \log(1/30 \delta) }{\Delta^2} = \frac{ \log(1/30 \delta) }{3\varepsilon |D|}
\end{align*}
as claimed.
\subsection{Calculations for Lemma \ref{low:lem1}}\label{apendix:low:lem1}

\begin{proposition}\label{low1:prop1}
For $\square\in\{-\Delta,0,\Delta\}$, $V_{A_\square}^*=\frac{ad-bc}{D}+\frac{d-a}{D}\square-\frac{\square^2}{D}$
\end{proposition}
\begin{proof}
$V_{A_\square}^*=\frac{(a+\square)(d-\square)-bc}{a+\square-b-c+d-\square}=\frac{ad-bc}{D}+\frac{d-a}{D}\square-\frac{\square^2}{D}$
\end{proof}


Recall that $x'=(\frac{d-c+\alpha}{D},\frac{a-b-\alpha}{D})$ and $y'=(\frac{d-b+\beta}{D},\frac{a-c-\beta}{D})$. Now we present the following proposition.
\begin{proposition}\label{low1:prop2}
For $\square\in\{-\Delta,0,\Delta\}$, $\langle x',A_\square y'\rangle=\frac{ad-bc}{D}+\frac{d-a}{D}\square+\frac{\alpha\beta+(\alpha+\beta)\square}{D}$
\end{proposition}
\begin{proof}
Let $V_1=\langle x',(a+\square,c) \rangle$ and $V_2=\langle x',(b,d-\square) \rangle$. First we have the following.
\begin{align*}
    V_1 &= \frac{d-c+\alpha}{D}\cdot (a+\square)+\frac{a-b-\alpha}{D}\cdot c\\
    &= \frac{d-c}{D}\cdot a+\frac{a-b}{D}\cdot c+\frac{d-c+\alpha}{D}\square+\frac{a-c}{D}\cdot \alpha\\
    &= \frac{ad-bc}{D}+\frac{d-c+\alpha}{D}\square+\frac{a-c}{D}\cdot \alpha
\end{align*}
Similarly, we have the following.
\begin{align*}
    V_2 &= \frac{d-c+\alpha}{D}\cdot b +\frac{a-b-\alpha}{D}\cdot (d-\square)\\
    &= \frac{d-c}{D}\cdot b+\frac{a-b}{D}\cdot d-\frac{a-b-\alpha}{D}\square+\frac{b-d}{D}\cdot \alpha\\
    &= \frac{ad-bc}{D}-\frac{a-b-\alpha}{D}\square+\frac{b-d}{D}\cdot \alpha
\end{align*}
Now observe that $\langle x',A_\square y'\rangle=\langle y',(V_1,V_2)\rangle$. Now we have the following:
\begin{align}
    \langle y',(V_1,V_2)\rangle&=\left\langle y', \left(\frac{ad-bc}{D},\frac{ad-bc}{D}\right)\right\rangle+\frac{d-b+\beta}{D}\cdot\frac{d-c+\alpha}{D}\cdot\square-\frac{a-c-\beta}{D}\cdot\frac{a-b-\alpha}{D}\cdot\square \nonumber\\
    &\quad+\frac{d-b+\beta}{D}\cdot\frac{a-c}{D}\cdot \alpha+\frac{a-c-\beta}{D}\cdot\frac{b-d}{D}\cdot \alpha\nonumber\\
    &=\frac{ad-bc}{D}+\frac{(d-b)(d-c)-(a-c)(a-b)}{D^2}\square+\frac{(d-c+\alpha)\beta+(a-b-\alpha)\beta}{D^2}\cdot \square \nonumber \\
    &\quad+\frac{(d-b)\alpha+(a-c)\alpha}{D^2}\cdot \square +\frac{(d-b)(a-c)\alpha-(a-c)(d-b)\alpha}{D^2}+\frac{a-b-c+d}{D^2}\cdot \alpha\beta \nonumber\\
    &= \frac{ad-bc}{D}+\frac{d-a}{D}\square+\frac{(\alpha+\beta)\square+\alpha\beta}{D}\label{prop:low1:st1}
\end{align}
We get (\ref{prop:low1:st1}) as $D=a-b-c+d$ and $(d-b)(d-c)-(a-c)(a-b)=(d-a)(a-b-c+d)$.
\end{proof}

\section{Proof of $\varepsilon$-good lower bound with respect to $\varepsilon$, $\Delta_{\min}$}\label{appendix:thm2}
Before finishing the proof of the Theorem \ref{thm:lower2}, we begin with the proof of Lemma~\ref{low2:lem1}
\begin{proof}
Let us first consider the case when $d-c> 2\Delta$. Observe that $V_{A_0}^*=\frac{ad-bc}{D}$, $V_{A_\Delta}^*=\frac{ad-bc}{D}+\frac{(d-b)-(a-c)}{D}\cdot\Delta$ and $V_{A_{-\Delta}}^*=a-\Delta$. For any $\alpha\in [\frac{c-d}{D},\frac{a-b}{D}]$ and $\beta\in[\frac{b-d}{D},\frac{a-c}{D}]$, let $x'=(\frac{d-c}{D}+\alpha,\frac{a-b}{D}-\alpha)$ and $y'=(\frac{d-b}{D}+\beta,\frac{a-c}{D}-\beta)$. Note that this parameterization ensures the range of $x',y'$ is equal to $\simplex_2$. It can be shown that $\langle x',A_\square y'\rangle=\frac{ad-bc}{D}+\frac{(d-b)-(a-c)}{D}\square+2\square \beta + D\alpha\beta$. We refer the reader to the Appendix \ref{apendix:low2:lem1} for the detailed calculations.

We will now show that regardless of what values $(\alpha,\beta)$ take (equivalently, regardless of what values $(x',y')$ take), there is at least one of the three alternative matrices has error of more than $\varepsilon$. If $|D\alpha\beta| > \varepsilon$, then $|V_{A_{0}}^*-\langle x', A_{0}y' \rangle|=|D\alpha\beta|> \varepsilon$. Let $f(\Delta)=\frac{(d-b)-(a-c)}{D}\Delta+2\Delta \beta$. If $|D\alpha\beta| \leq \varepsilon$ and $f(\Delta)\geq \frac{\Delta}{2}$, then $\langle x',A_{\Delta} y'\rangle-V_{A_{\Delta}}^*=D\alpha\beta+f(\Delta)-\frac{(d-b)-(a-c)}{D}\cdot\Delta\geq-\varepsilon+\frac{\Delta}{2}>\varepsilon $. Similarly, if $|D\alpha\beta| \leq \varepsilon$ and $f(\Delta)< \frac{\Delta}{2}$, then $V_{A_{-\Delta}}^*-\langle x',A_{-\Delta} y'\rangle=\frac{(a-b)(a-c)}{D}-\Delta-D\alpha\beta+f(\Delta)<\Delta_{\min}-\Delta+\varepsilon+\frac{\Delta}{2}\leq -\varepsilon$. Hence, we proved that for any $(x',y') \in \simplex_2 \times \simplex_2$, there exists a matrix $B\in\{A_{0},A_{\Delta},A_{2\Delta}\}$ such that the following holds:
\begin{equation*}
    |V_B^*-\langle x', By' \rangle|>\varepsilon
\end{equation*}

Next we consider the case when $d-c\leq 2\Delta$. Observe that $V_{A_0}^*=\frac{ad-bc}{D}$, $V_{A_\Delta}^*=d-\Delta$ and $V_{A_{-\Delta}}^*=a-\Delta$. For any $\alpha\in [\frac{c-d}{D},\frac{a-b}{D}]$ and $\beta\in[\frac{b-d}{D},\frac{a-c}{D}]$, let $x'=(\frac{d-c}{D}+\alpha,\frac{a-b}{D}-\alpha)$ and $y'=(\frac{d-b}{D}+\beta,\frac{a-c}{D}-\beta)$. Note that this parameterization ensures the range of $x',y'$ is equal to $\simplex_2$. Recall that $\langle x',A_\square y'\rangle=\frac{ad-bc}{D}+\frac{(d-b)-(a-c)}{D}\square+2\square \beta + D\alpha\beta$.

We will now show that regardless of what values $(\alpha,\beta)$ take (equivalently, regardless of what values $(x',y')$ take), there is at least one of the three alternative matrices has error of more than $\varepsilon$. If $|D\alpha\beta| > \varepsilon$, then $|V_{A_{0}}^*-\langle x', A_{0}y' \rangle|=|D\alpha\beta|> \varepsilon$. Let $f(\Delta)=\frac{(d-b)-(a-c)}{D}\Delta+2\Delta \beta$. If $|D\alpha\beta| \leq \varepsilon$ and $f(\Delta)\geq \frac{\Delta}{2}$, then we have the following:
\begin{align*}
    \langle x',A_{\Delta} y'\rangle-V_{A_{\Delta}}^*&=D\alpha\beta+f(\Delta)-\frac{(d-b)(d-c)}{D}+\Delta\\
    &\geq-\varepsilon+\frac{\Delta}{2}-\frac{d-c}{2}+\Delta \tag{as $d-b\leq D/2$}\\
    &\geq -\varepsilon+\frac{\Delta}{2}\tag{as $d-c\leq 2\Delta$}\\
    &>\varepsilon 
\end{align*}
Similarly, if $|D\alpha\beta| \leq \varepsilon$ and $f(\Delta)< \frac{\Delta}{2}$, then $V_{A_{-\Delta}}^*-\langle x',A_{-\Delta} y'\rangle=\frac{(a-b)(a-c)}{D}-\Delta-D\alpha\beta+f(\Delta)<\Delta_{\min}-\Delta+\varepsilon+\frac{\Delta}{2}\leq -\varepsilon$. Hence, we proved that for any $(x',y') \in \simplex_2 \times \simplex_2$, there exists a matrix $B\in\{A_{0},A_{\Delta},A_{-\Delta}\}$ such that the following holds:
\begin{equation*}
    |V_B^*-\langle x', By' \rangle|>\varepsilon
\end{equation*}
\end{proof}

\subsection{Proof of Theorem \ref{thm:lower2}}
Let $\nu_{i,j}^A = \mathcal{N}(A_{ij},1)$ be the distribution of an observation when playing pair $(i,j)$ with matrix $A$. 
Let $\P_A$ denote the probability law of the internal randomness of the algorithm and random observations.
If an algorithm is $(\varepsilon,\delta)$-PAC-good and outputs a solution $(\widehat{x},\widehat{y})$ then $\min_A \P_A(|V_A^*-\langle \widehat{x}, A \widehat{y} \rangle|\leq \varepsilon) \geq 1-\delta$.
We will show that if an algorithm is $(\varepsilon,\delta)$-PAC-good then it can also accomplish a particular hypothesis.
We will conclude by noting that any procedure that can accomplish the hypothesis test must take the claimed sample complexity. 

For any pair of mixed strategies $(\widehat{x},\widehat{y})$ output by the procedure at the stopping time $\tau$, define 
\begin{align*}
    \phi = \{A_{-\Delta},A_{0},A_{\Delta}\} \setminus \arg\max_{B \in \{A_{-\Delta},A_{0},A_{\Delta}\}} |V_{B}^*-\langle \widehat{x}, B \widehat{y} \rangle|,
\end{align*}
breaking ties arbitrarily in the maximum so that $\phi \in  \{A_{-\Delta},A_{0}\}\cup \{A_{-\Delta},A_{\Delta}\}\cup \{A_{0},A_{\Delta}\}$.
Note that
\begin{align}\label{eqn:correctness_delta:2}
    \P_{A_0}( A_0 \in \phi ) \geq \P_{A_0}( A_0 \in \phi,  |V_{A_0}^*-\langle \widehat{x},A_0 \widehat{y} \rangle| \leq \varepsilon) = \P_{A_0}(   |V_{A_0}^*-\langle \widehat{x}, A_0 \widehat{y} \rangle| \leq \varepsilon) \geq 1-\delta
\end{align}
where the equality follows from the Lemma~\ref{low2:lem1}: at least one of the three matrices must have a loss of more than $\varepsilon$, but $A_0$ has a loss of at most $\varepsilon$, thus $A_0 \in \phi$.
Now because
\begin{align*}
    2 \max\{ \P_{A_0}( \phi = \{A_0,A_{-\Delta}\} ) , \P_{A_0}( \phi = \{A_0,A_{\Delta}\} ) \} &\geq \P_{A_0}( \phi = \{A_0,A_{-\Delta}\} ) + \P_{A_0}( \phi = \{A_0,A_{\Delta}\} ) \\
    &= \P_{A_0}( A_0 \in \phi ) \geq 1-\delta
\end{align*}
we have that $\P_{A_0}( \phi = \{A_0,A_{-\Delta}\} ) \geq \frac{1-\delta}{2}$ or $\P_{A_0}( \phi = \{A_0,A_{\Delta} \}) \geq \frac{1-\delta}{2}$.
Let's assume the former (the latter case is handled identically).
By the same argument as \eqref{eqn:correctness_delta:2} we have that $\P_{A_\Delta}( \phi = \{A_0,A_{-\Delta}\} ) \leq \delta$.

For a stopping time $\tau$, let $N_{i,j}(\tau)$ denote the number of times $(i,j)$ is sampled. Recalling that $\nu_{i,j}^A = \mathcal{N}(A_{ij},1)$, we have by Lemma 1 of \citet{kaufmann2016complexity} that
\begin{align*}
\sum_{i=1}^2\sum_{j=1}^2\mathbb{E}_{A_0}[ N_{i,j}(\tau) ] KL( \nu_{i,j}^{A_0}, \nu_{i,j}^{A_\Delta} ) \geq d( \P_{A_0}( \phi = \{A_0,A_{-\Delta}\} ), \P_{A_\Delta}( \phi = \{A_0,A_{-\Delta}\} ) )
\end{align*}
where for any $i,j\in\{1,2\}$, $KL( \nu_{i,j}^{A_0}, \nu_{i,j}^{A_\Delta} ) = \Delta^2/2$ and $d(p,q) = p \log(\frac{p}{q}) + (1-p) \log(\frac{1-p}{1-q})$.
Since 
\begin{align*}
    d( \P_{A_0}( \phi = \{A_0,A_{-\Delta}\} ), \P_{A_\Delta}( \phi = \{A_0,A_{-\Delta}\} ) ) &\geq d( \tfrac{1-\delta}{2},\delta) \\
    &= \tfrac{1-\delta}{2} \log( \tfrac{1-\delta}{2 \delta} ) + \tfrac{1+\delta}{2} \log( \tfrac{1+\delta}{2(1-\delta)} ) \\
    &= \tfrac{1}{2} \log( \tfrac{1+\delta}{4 \delta} ) -  \tfrac{\delta}{2} \log( \tfrac{(1-\delta)^2}{ \delta (1+\delta)} ) \\
    &\geq \tfrac{1}{2} \log( \tfrac{1+\delta}{4 \delta} ) -  1/8 > \tfrac{1}{2} \log(1/30\delta)
\end{align*} 
and $\tau = N_{1,1}(\tau)+N_{1,2}(\tau)+N_{2,1}(\tau)+N_{2,2}(\tau)$ we conclude that
\begin{align*}
    \mathbb{E}_{A_0}[ \tau ] \geq \frac{ \log(1/30 \delta) }{\Delta^2} = \min\left\{\frac{\log(1/30 \delta) }{36\varepsilon^2 },\frac{\log(1/30 \delta) }{36\Delta_{\min}^2 }\right\}
\end{align*}
as claimed.

\subsection{Calculations for Lemma \ref{low2:lem1}}\label{apendix:low2:lem1}
Recall that $x'=(\frac{d-c}{D}+\alpha,\frac{a-b}{D}-\alpha)$ and $y'=(\frac{d-b}{D}+\beta,\frac{a-c}{D}-\beta)$. Let $(A_\square)^r_i$ denote the $i$-th row of $A_\square$ and $(A_\square)^c_j$ denote the $j$-th column of $A_\square$. 

First, observe that $\langle x', (A_\square)^c_1\rangle= \frac{d-c}{D}\cdot a + \frac{d-c}{D}\cdot \square+a\alpha+\square \alpha+\frac{a-b}{D}\cdot c + \frac{a-b}{D}\cdot \square-c\alpha-\square \alpha=\frac{ad-bc}{D}+\square+(a-c)\alpha$. Similarly, we have  $\langle x', (A_\square)^c_2\rangle= \frac{d-c}{D}\cdot b - \frac{d-c}{D}\cdot \square+b\alpha-\square \alpha+\frac{a-b}{D}\cdot d - \frac{a-b}{D}\cdot \square-d\alpha+\square \alpha=\frac{ad-bc}{D}-\square+(b-d)\alpha$.

Now we present the following proposition.
\begin{proposition}\label{low2:prop1}
$\langle x',A_\square y'\rangle=\frac{ad-bc}{D}+\frac{(d-b)-(a-c)}{D}\square+2\square \beta + D\alpha\beta$
\end{proposition}
\begin{proof}
Let $V_1=\langle x', (A_\square)^c_1\rangle$ and $V_2=\langle x', (A_\square)^c_2\rangle$. Now observe that $\langle x',A_\square y'\rangle=\langle y',(V_1,V_2)\rangle$. Now we have the following:
\begin{align}
    \langle y',(V_1,V_2)\rangle&=\left\langle y', \left(\frac{ad-bc}{D},\frac{ad-bc}{D}\right)\right\rangle+\left(\frac{d-b}{D}+\beta\right)\cdot\square-\left(\frac{a-c}{D}-\beta\right)\cdot\square \nonumber\\
    &\quad+\left(\frac{d-b}{D}+\beta\right)\cdot(a-c)\alpha-\left(\frac{a-c}{D}-\beta\right)\cdot(d-b)\alpha\nonumber\\
    &=\frac{ad-bc}{D}+\frac{(d-b)-(a-c)}{D}\square+2\square \beta \nonumber \\
    &\quad+\frac{(d-b)(a-c)-(a-c)(d-b)}{D}\cdot \alpha +(a-b-c+d)\alpha\beta \nonumber\\
    &= \frac{ad-bc}{D}+\frac{(d-b)-(a-c)}{D}\square+2\square \beta + D\alpha\beta\tag{as $D=a-b-c+d$}
\end{align}
\end{proof}

\section{Proof of $\varepsilon$-good lower bound for games with multiple Nash Equilibria}\label{appendix:thm:multiple}
Before finishing the proof of the Theorem \ref{thm:multiple}, we begin with the proof of Lemma~\ref{lem1:multiple}
\begin{proof}
Let us first consider the case when $D=d-c> 2\Delta$. Observe that $V_{A_0}^*=a$, $V_{A_\Delta}^*=a+\frac{(d-a)-(a-c)}{D}\cdot\Delta$ and $V_{A_{-\Delta}}^*=a-\Delta$. For any $\alpha\in [-1,0]$ and $\beta\in[\frac{a-d}{D},\frac{a-c}{D}]$, let $x'=(1+\alpha,-\alpha)$ and $y'=(\frac{d-a}{D}+\beta,\frac{a-c}{D}-\beta)$. Note that this parameterization ensures the range of $x',y'$ is equal to $\simplex_2$. It can be shown that $\langle x',A_\square y'\rangle=a+\frac{(d-a)-(a-c)}{D}\square+2\square \beta + D\alpha\beta$. We refer the reader to the Appendix \ref{appendix:lem:multiple} for the detailed calculations.

We will now show that regardless of what values $(\alpha,\beta)$ take (equivalently, regardless of what values $(x',y')$ take), there is at least one of the three alternative matrices has error of more than $\varepsilon$. If $|D\alpha\beta| > \varepsilon$, then $|V_{A_{0}}^*-\langle x', A_{0}y' \rangle|=|D\alpha\beta|> \varepsilon$. Let $f(\Delta)=\frac{(d-a)-(a-c)}{D}\Delta+2\Delta \beta$. If $|D\alpha\beta| \leq \varepsilon$ and $f(\Delta)\geq \frac{\Delta}{2}$, then $\langle x',A_{\Delta} y'\rangle-V_{A_{\Delta}}^*=D\alpha\beta+f(\Delta)-\frac{(d-a)-(a-c)}{D}\cdot\Delta\geq-\varepsilon+\frac{\Delta}{2}>\varepsilon $. Similarly, if $|D\alpha\beta| \leq \varepsilon$ and $f(\Delta)< \frac{\Delta}{2}$, then $V_{A_{-\Delta}}^*-\langle x',A_{-\Delta} y'\rangle=-\Delta-D\alpha\beta+f(\Delta)<-\Delta+\varepsilon+\frac{\Delta}{2}\leq -\varepsilon$. Hence, we proved that for any $(x',y') \in \simplex_2 \times \simplex_2$, there exists a matrix $B\in\{A_{0},A_{\Delta},A_{2\Delta}\}$ such that the following holds:
\begin{equation*}
    |V_B^*-\langle x', By' \rangle|>\varepsilon
\end{equation*}

Next we consider the case when $d-c\leq 2\Delta$. Observe that $V_{A_0}^*=a$, $V_{A_\Delta}^*=d-\Delta$ and $V_{A_{-\Delta}}^*=a-\Delta$. For any $\alpha\in [-1,0]$ and $\beta\in[\frac{a-d}{D},\frac{a-c}{D}]$, let $x'=(1+\alpha,-\alpha)$ and $y'=(\frac{d-a}{D}+\beta,\frac{a-c}{D}-\beta)$. Note that this parameterization ensures the range of $x',y'$ is equal to $\simplex_2$. Recall that $\langle x',A_\square y'\rangle=a+\frac{(d-a)-(a-c)}{D}\square+2\square \beta + D\alpha\beta$.

We will now show that regardless of what values $(\alpha,\beta)$ take (equivalently, regardless of what values $(x',y')$ take), there is at least one of the three alternative matrices has error of more than $\varepsilon$. If $|D\alpha\beta| > \varepsilon$, then $|V_{A_{0}}^*-\langle x', A_{0}y' \rangle|=|D\alpha\beta|> \varepsilon$. Let $f(\Delta)=\frac{(d-a)-(a-c)}{D}\Delta+2\Delta \beta$. If $|D\alpha\beta| \leq \varepsilon$ and $f(\Delta)\geq \frac{\Delta}{2}$, then we have the following:
\begin{align*}
    \langle x',A_{\Delta} y'\rangle-V_{A_{\Delta}}^*&=D\alpha\beta+f(\Delta)-\frac{(d-a)(d-c)}{D}+\Delta\\
    &\geq-\varepsilon+\frac{\Delta}{2}-\frac{d-c}{2}+\Delta \tag{as $d-a\leq D/2$}\\
    &\geq -\varepsilon+\frac{\Delta}{2}\tag{as $d-c\leq 2\Delta$}\\
    &>\varepsilon 
\end{align*}
Similarly, if $|D\alpha\beta| \leq \varepsilon$ and $f(\Delta)< \frac{\Delta}{2}$, then $V_{A_{-\Delta}}^*-\langle x',A_{-\Delta} y'\rangle=-\Delta-D\alpha\beta+f(\Delta)<-\Delta+\varepsilon+\frac{\Delta}{2}\leq -\varepsilon$. Hence, we proved that for any $(x',y') \in \simplex_2 \times \simplex_2$, there exists a matrix $B\in\{A_{0},A_{\Delta},A_{-\Delta}\}$ such that the following holds:
\begin{equation*}
    |V_B^*-\langle x', By' \rangle|>\varepsilon
\end{equation*}
\end{proof}

\subsection{Proof of Theorem \ref{thm:lower2}}
Let $\nu_{i,j}^A = \mathcal{N}(A_{ij},1)$ be the distribution of an observation when playing pair $(i,j)$ with matrix $A$. 
Let $\P_A$ denote the probability law of the internal randomness of the algorithm and random observations.
If an algorithm is $(\varepsilon,\delta)$-PAC-good and outputs a solution $(\widehat{x},\widehat{y})$ then $\min_A \P_A(|V_A^*-\langle \widehat{x}, A \widehat{y} \rangle|\leq \varepsilon) \geq 1-\delta$.
We will show that if an algorithm is $(\varepsilon,\delta)$-PAC-good then it can also accomplish a particular hypothesis.
We will conclude by noting that any procedure that can accomplish the hypothesis test must take the claimed sample complexity. 

For any pair of mixed strategies $(\widehat{x},\widehat{y})$ output by the procedure at the stopping time $\tau$, define 
\begin{align*}
    \phi = \{A_{-\Delta},A_{0},A_{\Delta}\} \setminus \arg\max_{B \in \{A_{-\Delta},A_{0},A_{\Delta}\}} |V_{B}^*-\langle \widehat{x}, B \widehat{y} \rangle|,
\end{align*}
breaking ties arbitrarily in the maximum so that $\phi \in  \{A_{-\Delta},A_{0}\}\cup \{A_{-\Delta},A_{\Delta}\}\cup \{A_{0},A_{\Delta}\}$.
Note that
\begin{align}\label{eqn:correctness_delta:multiple}
    \P_{A_0}( A_0 \in \phi ) \geq \P_{A_0}( A_0 \in \phi,  |V_{A_0}^*-\langle \widehat{x},A_0 \widehat{y} \rangle| \leq \varepsilon) = \P_{A_0}(   |V_{A_0}^*-\langle \widehat{x}, A_0 \widehat{y} \rangle| \leq \varepsilon) \geq 1-\delta
\end{align}
where the equality follows from the Lemma~\ref{lem1:multiple}: at least one of the three matrices must have a loss of more than $\varepsilon$, but $A_0$ has a loss of at most $\varepsilon$, thus $A_0 \in \phi$.
Now because
\begin{align*}
    2 \max\{ \P_{A_0}( \phi = \{A_0,A_{-\Delta}\} ) , \P_{A_0}( \phi = \{A_0,A_{\Delta}\} ) \} &\geq \P_{A_0}( \phi = \{A_0,A_{-\Delta}\} ) + \P_{A_0}( \phi = \{A_0,A_{\Delta}\} ) \\
    &= \P_{A_0}( A_0 \in \phi ) \geq 1-\delta
\end{align*}
we have that $\P_{A_0}( \phi = \{A_0,A_{-\Delta}\} ) \geq \frac{1-\delta}{2}$ or $\P_{A_0}( \phi = \{A_0,A_{\Delta} \}) \geq \frac{1-\delta}{2}$.
Let's assume the former (the latter case is handled identically).
By the same argument as \eqref{eqn:correctness_delta:multiple} we have that $\P_{A_\Delta}( \phi = \{A_0,A_{-\Delta}\} ) \leq \delta$.

For a stopping time $\tau$, let $N_{i,j}(\tau)$ denote the number of times $(i,j)$ is sampled. Recalling that $\nu_{i,j}^A = \mathcal{N}(A_{ij},1)$, we have by Lemma 1 of \citet{kaufmann2016complexity} that
\begin{align*}
\sum_{i=1}^2\sum_{j=1}^2\mathbb{E}_{A_0}[ N_{i,j}(\tau) ] KL( \nu_{i,j}^{A_0}, \nu_{i,j}^{A_\Delta} ) \geq d( \P_{A_0}( \phi = \{A_0,A_{-\Delta}\} ), \P_{A_\Delta}( \phi = \{A_0,A_{-\Delta}\} ) )
\end{align*}
where for any $i,j\in\{1,2\}$, $KL( \nu_{i,j}^{A_0}, \nu_{i,j}^{A_\Delta} ) = \Delta^2/2$ and $d(p,q) = p \log(\frac{p}{q}) + (1-p) \log(\frac{1-p}{1-q})$.
Since 
\begin{align*}
    d( \P_{A_0}( \phi = \{A_0,A_{-\Delta}\} ), \P_{A_\Delta}( \phi = \{A_0,A_{-\Delta}\} ) ) &\geq d( \tfrac{1-\delta}{2},\delta) \\
    &= \tfrac{1-\delta}{2} \log( \tfrac{1-\delta}{2 \delta} ) + \tfrac{1+\delta}{2} \log( \tfrac{1+\delta}{2(1-\delta)} ) \\
    &= \tfrac{1}{2} \log( \tfrac{1+\delta}{4 \delta} ) -  \tfrac{\delta}{2} \log( \tfrac{(1-\delta)^2}{ \delta (1+\delta)} ) \\
    &\geq \tfrac{1}{2} \log( \tfrac{1+\delta}{4 \delta} ) -  1/8 > \tfrac{1}{2} \log(1/30\delta)
\end{align*} 
and $\tau = N_{1,1}(\tau)+N_{1,2}(\tau)+N_{2,1}(\tau)+N_{2,2}(\tau)$ we conclude that
\begin{align*}
    \mathbb{E}_{A_0}[ \tau ] \geq \frac{ \log(1/30 \delta) }{\Delta^2} = \frac{\log(1/30 \delta) }{36\varepsilon^2 }
\end{align*}
as claimed.

\subsection{Calculations for Lemma \ref{lem1:multiple}}\label{appendix:lem:multiple}
Recall that $x'=(1+\alpha,-\alpha)$ and $y'=(\frac{d-a}{D}+\beta,\frac{a-c}{D}-\beta)$. Let $(A_\square)^r_i$ denote the $i$-th row of $A_\square$ and $(A_\square)^c_j$ denote the $j$-th column of $A_\square$. 

First, observe that $\langle x', (A_\square)^c_1\rangle= 1\cdot a + 1\cdot \square+a\alpha+\square \alpha-c\alpha-\square \alpha=a+\square+(a-c)\alpha$. Similarly, we have  $\langle x', (A_\square)^c_2\rangle= 1\cdot a - 1\cdot \square+a\alpha-\square \alpha-d\alpha+\square \alpha=a-\square+(a-d)\alpha$.

Now we present the following proposition.
\begin{proposition}\label{low2:prop1}
$\langle x',A_\square y'\rangle=a+\frac{(d-a)-(a-c)}{D}\square+2\square \beta + D\alpha\beta$
\end{proposition}
\begin{proof}
Let $V_1=\langle x', (A_\square)^c_1\rangle$ and $V_2=\langle x', (A_\square)^c_2\rangle$. Now observe that $\langle x',A_\square y'\rangle=\langle y',(V_1,V_2)\rangle$. Now we have the following:
\begin{align}
    \langle y',(V_1,V_2)\rangle&=\langle y', (a,a)\rangle+\left(\frac{d-a}{D}+\beta\right)\cdot\square-\left(\frac{a-c}{D}-\beta\right)\cdot\square \nonumber\\
    &\quad+\left(\frac{d-a}{D}+\beta\right)\cdot(a-c)\alpha-\left(\frac{a-c}{D}-\beta\right)\cdot(d-a)\alpha\nonumber\\
    &=a+\frac{(d-a)-(a-c)}{D}\square+2\square \beta \nonumber \\
    &\quad+\frac{(d-a)(a-c)-(a-c)(d-a)}{D}\cdot \alpha +(a-c+d-a)\alpha\beta \nonumber\\
    &= a+\frac{(d-a)-(a-c)}{D}\square+2\square \beta + D\alpha\beta\tag{as $D=d-c$}
\end{align}
\end{proof}

\section{Proof of $\varepsilon$-Nash equilibrium Upper Bound}\label{appendix:thm5}

We establish the sample complexity and the correctness of the Algorithm \ref{alg-ucb-2} by proving the Theorem \ref{thm:alg2}.

\begin{proof}[Proof of Theorem \ref{thm:alg2}]
Let $\bar A_{ij,t}$ denote the empirical mean of $A_{ij}$ at time step $t$. Let us begin by defining two events:
\begin{align*}
    G:=&\bigcap_{t=1}^T\bigcap_{i=1}^2 \bigcap_{j=1}^2 \{ |A_{ij}-\bar A_{ij,t}|\leq \sqrt{\tfrac{2\log({16T}/{\delta})}{t}} \} \\
    E:=&\bigcap_{i=1}^2 \bigcap_{j=1}^2 \{ |A_{ij}-\bar A_{ij,T}|\leq \sqrt{\tfrac{2\log({16}/{\delta})}{T}} \} 
\end{align*}
A union bound and sub-Gaussian-tail bound demonstrates that $\P( G^c \cup E^c ) \leq \P(G^c) + \P(E^c) \leq \delta$.
Consequently, events $E$ and $G$ hold simultaneously with probability at least $1-\delta$, so in what follows, assume they hold. 

If $A$ has a PSNE and if the condition in the line \ref{alg2:con1} of the algorithm \ref{alg-ucb-2} is satisfied, then we identify an $\varepsilon$-Nash equilibrium in $\frac{800\log (\frac{16T}{ \delta})}{ \Delta_{\min}^2}$ time steps due to Lemma \ref{lem:alg2:tmin} and Corollary \ref{cor:alg2:saddle}. On the other hand, if $A$ has a PSNE but the for loop completes after $t=T$ iterations,  then we identify an $\varepsilon$-good solution in $T=\frac{8 \log (16/\delta) }{\varepsilon^2}$ time steps due to Lemma \ref{lem:trivial:NE}. Note that in this case, $T<\frac{800\log (\frac{16T}{ \delta})}{ \Delta_{\min}^2}$ due to Lemma \ref{lem:alg2:tmin}. Hence, if $A$ has a PSNE, we identify an $\varepsilon$-Nash equilibrium in $O\left(\min\left\{\frac{\log(1/\delta)}{\varepsilon^2},\frac{\log(T/\delta)}{\Delta_{\min}^2}\right\}\right)$ time steps.

Let us assume for the rest of the proof that $A$ has a unique Nash equilibrium which is not a PSNE. If the condition in the line \ref{alg2:con2} of the algorithm \ref{alg-ucb-2} is satisfied, then we identify an $\varepsilon$-Nash equilibrium in $T=\frac{8 \log (16/\delta) }{\varepsilon^2}$ time steps due to Lemma \ref{lem:trivial:NE}. Now observe that in this case $T=O\left(\min\left\{\frac{\log(1/\delta)}{\varepsilon^2},\frac{\Delta_{m_2}^2\log(T/\delta)}{\varepsilon^2 D^2}\right\}\right)$ due to Lemma \ref{lem:alg2:dm2/D:one}. On the other hand, if the for loop completes after $t=T$ iterations,  then we identify an $\varepsilon$-good solution in $T=\frac{8 \log (16/\delta) }{\varepsilon^2}$ time steps due to Lemma \ref{lem:trivial:NE}. Note that in this case, $T<\frac{800\log (\frac{16T}{ \delta})}{ \Delta_{\min}^2}$ due to Lemma \ref{lem:alg2:tmin}.

Now let us  assume for the rest of the proof that the condition in the line \ref{alg2:con3} is satisfied.  Then due to Lemma \ref{lem:alg2:dm2/D:both}, we have $\frac{800\Delta^2_{m_2}\log (\frac{16T}{\delta})}{9\varepsilon^2 |D|^2} \leq N \leq \frac{450\Delta^2_{m_2}\log (\frac{16T}{\delta})}{\varepsilon^2 |D|^2} $. If the condition in the line \ref{alg2:con4} is satisfied, then we identify an $\varepsilon$-Nash equilibrium in $T=\frac{8 \log (16/\delta) }{\varepsilon^2}$ time steps due to Lemma \ref{lem:trivial:NE}. Now observe that in this case $T=O\left(\min\left\{\frac{\log(1/\delta)}{\varepsilon^2},\max\left\{\frac{\log(T/\delta)}{\Delta_{\min}^2},\frac{\Delta_{m_2}^2\log(T/\delta)}{\varepsilon^2 D^2}\right\}\right\}\right)$ as $T< \frac{800\log (\frac{16T}{ \delta})}{ \Delta_{\min}^2}+N$. If the condition in the line \ref{alg2:con4} is not satisfied, then we identify an $\varepsilon$-Nash equilibrium due to Lemma \ref{lem:alg2:main}. In this case, let the number of times we are required to sample each element be $n_0$. Then $n_0\leq T$ and $n_0\leq \frac{800\log (\frac{16T}{ \delta})}{ \Delta_{\min}^2}+ \frac{450\Delta^2_{m_2}\log (\frac{16T}{\delta})}{\varepsilon^2 |D|^2} $. Hence, $n_0= O\left(\min\left\{\frac{\log(1/\delta)}{\varepsilon^2},\max\left\{\frac{\log(T/\delta)}{\Delta_{\min}^2},\frac{\Delta_{m_2}^2\log(T/\delta)}{\varepsilon^2 D^2}\right\}\right\}\right)$.
\end{proof}

\subsection{Consequential lemmas of Algorithm~\ref{alg-ucb-2}'s conditional statements}
Recall the definitions of events $E$ and $G$.
We first present a few lemmas that deal with empirical estimates and instance dependent parameters like $\tilde \Delta_{\min},\tilde D,\tilde\Delta_{m_2},\Delta_{m_2}, \Delta_{\min}$ and $|D|$. Whenever we fix a time step $t\leq T$ and discuss the parameters like $\tilde  \Delta_{\min}, \tilde D, \tilde\Delta_{m_2}$ and $\Delta$, we consider those values that have been assigned to these parameters during the time step $t$.

We begin with upper bounding $ |\Delta_{\min}-\tilde \Delta_{\min}|$ and $|\Delta_{m_2}-\tilde \Delta_{m_2}|$ in the following lemma.
\begin{lemma}\label{lem:alg2:delta}
Fix a time step $t\leq T$. If the event $G$ holds, then we have the following:
\begin{itemize}
    \item $|\Delta_{\min}-\tilde \Delta_{\min}|\leq 2\Delta$
    \item $|\Delta_{m_2}-\tilde \Delta_{m_2}|\leq 2\Delta$
\end{itemize}
\end{lemma}
\begin{proof}
As the event $G$ holds true, we have $\left||A_{ij}-A_{i'j'}|-|\bar A_{ij}-\bar A_{i'j'}|\right|\leq 2\Delta$ for any $i,j,i',j'$. By repeatedly apply Lemma \ref{lem:deviation}, we get $|\Delta_{\min}-\tilde \Delta_{\min}|\leq 2\Delta$ and $|\Delta_{m_2}-\tilde \Delta_{m_2}|\leq 2\Delta$.
\end{proof}
The following lemma upper bounds the number of time steps required to satisfy the condition $1\leq \frac{\tilde \Delta_{\min}+2\Delta}{\tilde \Delta_{\min}-2\Delta}\leq \frac{3}{2}$.
\begin{lemma}\label{lem:alg2:tmin}
Let $t$ be the time step when the condition $1\leq \frac{\tilde \Delta_{\min}+2\Delta}{\tilde \Delta_{\min}-2\Delta}\leq \frac{3}{2}$ holds true for the first time. If the event $G$ holds, then $t\leq \frac{800\log (\frac{16T}{ \delta})}{ \Delta_{\min}^2}$.
\end{lemma}
\begin{proof}
Consider the time step $t= \frac{800\log (\frac{16T}{ \delta})}{ \Delta_{\min}^2}$. Let us assume that the event $G$ holds. Then for every element $(i,j)$, we have  $|A_{ij}-\bar A_{ij}|\leq \Delta=\sqrt{\frac{2\log(\frac{16T}{\delta})}{t}}=\frac{\Delta_{\min}}{20}$. Now observe that $\tilde \Delta_{\min}+2\Delta\leq \Delta_{\min}+4\Delta = \frac{6\Delta_{\min}}{5}$. Similarly, we have $\tilde \Delta_{\min}-2\Delta\geq \Delta_{\min}-4\Delta \geq \frac{4\Delta_{\min}}{5}$. Hence, we have $1\leq \frac{\tilde \Delta_{\min}+2\Delta}{\tilde \Delta_{\min}-2\Delta}\leq \frac{3}{2}$. 

\end{proof}
The following two lemmas bound the ratios $\frac{\tilde \Delta_{\min}}{\Delta_{\min}}$ and $\frac{\tilde \Delta_{m_2}}{\Delta_{m_2}}$.
\begin{lemma}\label{lem:alg2:dmin:ratio}
Let $t$ be the time step when the condition $1\leq \frac{\tilde \Delta_{\min}+2\Delta}{\tilde \Delta_{\min}-2\Delta}\leq \frac{3}{2}$ holds true for the first time. If the event $G$ holds, then $\frac{5}{6}\leq \frac{\tilde \Delta_{\min}}{\Delta_{\min}} \leq \frac{5}{4}$ at the time step $t$.
\end{lemma}
\begin{proof}
Let us assume that the event $G$ holds. Then for every element $(i,j)$, we have  $|A_{ij}-\bar A_{ij}|\leq \Delta$. As $\frac{\tilde \Delta_{\min}+2\Delta}{\tilde \Delta_{\min}-2\Delta}\leq \frac{3}{2}$, we have $\Delta\leq \frac{\tilde \Delta_{\min}}{10}$.
Now observe that $\frac{\tilde \Delta_{\min}}{\Delta_{\min}}\leq \frac{\tilde\Delta_{\min}}{\tilde\Delta_{\min}-2\Delta}\leq \frac{\tilde\Delta_{\min}}{4\tilde\Delta_{\min}/5}= \frac{5}{4}$. Next observe that $\frac{\tilde \Delta_{\min}}{\Delta_{\min}}\geq \frac{\tilde \Delta_{\min}}{\tilde \Delta_{\min}+2\Delta}\geq\frac{\tilde \Delta_{\min}}{6\tilde \Delta_{\min}/5}=\frac{5}{6}$.
\end{proof}

\begin{lemma}\label{lem:alg2:dm2:ratio}
Let $t$ be the time step when the condition $1\leq \frac{\tilde \Delta_{\min}+2\Delta}{\tilde \Delta_{\min}-2\Delta}\leq \frac{3}{2}$ holds true for the first time. If the event $G$ holds, then $\frac{5}{6}\leq \frac{\tilde \Delta_{m_2}}{\Delta_{m_2}} \leq \frac{5}{4}$ at the time step $t$.
\end{lemma}
\begin{proof}
Let us assume that the event $G$ holds true. Then for every element $(i,j)$, we have  $|A_{ij}-\bar A_{ij}|\leq \Delta$. As $\frac{\tilde \Delta_{\min}+2\Delta}{\tilde \Delta_{\min}-2\Delta}\leq \frac{3}{2}$, we have $\Delta\leq \frac{\tilde \Delta_{\min}}{10}\leq \frac{\tilde \Delta_{m_2}}{10}$.
Now observe that $\frac{\tilde \Delta_{m_2}}{\Delta_{m_2}}\leq \frac{\tilde\Delta_{m_2}}{\tilde\Delta_{m_2}-2\Delta}\leq \frac{\tilde\Delta_{m_2}}{4\tilde\Delta_{m_2}/5}= \frac{5}{4}$. Next observe that $\frac{\tilde \Delta_{m_2}}{\Delta_{m_2}}\geq \frac{\tilde \Delta_{m_2}}{\tilde \Delta_{m_2}+2\Delta}\geq\frac{\tilde \Delta_{m_2}}{6\tilde \Delta_{m_2}/5}= \frac{5}{6}$.
\end{proof}
The following lemma and the subsequent corollary relates the empirical matrix $\bar A$ to the input matrix $A$.
\begin{lemma}\label{lem:alg2:gap}
Let $t$ be the time step when the condition $1\leq \frac{\tilde \Delta_{\min}+2\Delta}{\tilde \Delta_{\min}-2\Delta}\leq \frac{3}{2}$ holds true for the first time. If the event $G$ holds, then at any time step $t_0$ such that  $t\leq t_0\leq T$, we have the following:
\begin{itemize}
    \item If $A_{ij_1}>A_{ij_2}$, then $\bar A_{ij_1}>\bar A_{ij_2}$ 
    \item If $A_{i_1j}>A_{i_2j}$, then $\bar A_{i_1j}>\bar A_{i_2j}$
    \item If $\bar A_{ij_1}>\bar A_{ij_2}$, then $ A_{ij_1}> A_{ij_2}$ 
    \item If $\bar A_{i_1j}>\bar A_{i_2j}$, then $ A_{i_1j}> A_{i_2j}$
\end{itemize}
\end{lemma}
\begin{proof}
As $\frac{\tilde \Delta_{\min}+2\Delta}{\tilde \Delta_{\min}-2\Delta}\leq \frac{3}{2}$, we have $\Delta\leq \frac{\tilde \Delta_{\min}}{10}$. Due to Lemma \ref{lem:alg2:dmin:ratio}, we have $\Delta\leq \frac{\Delta_{\min}}{8}$. As event $G$ holds, for any element $(i,j)$, we have $|A_{ij}-\bar A_{ij}|\leq \sqrt{\frac{2\log(\frac{16T}{\delta})}{t_0}}\leq \Delta $.

If $A_{ij_1}>A_{ij_2}$, we have the following:
\begin{align*}
    \bar A_{ij_1} & \geq A_{ij_1}-\Delta \\
    & \geq A_{ij_2}+\Delta_{\min}-\Delta \tag{as $A_{ij_1}- A_{ij_2}\geq \Delta_{\min}$}\\
    & > A_{ij_2} + \Delta \tag{as $\Delta\leq \frac{ \Delta_{\min}}{8}$}\\
    & \geq \bar A_{ij_2} \tag{as event $G$ holds}\\
\end{align*}

If $A_{i_1j}>A_{i_2j}$, we have the following:
\begin{align*}
    \bar A_{i_1j} & \geq A_{i_1j}-\Delta \\
    & \geq A_{i_2j}+\Delta_{\min}-\Delta \tag{as $A_{i_1j}- A_{i_2j}\geq \Delta_{\min}$}\\
    & > A_{i_2j} + \Delta \tag{as $\Delta\leq \frac{ \Delta_{\min}}{8}$}\\
    & \geq \bar A_{i_2j} \tag{as event $G$ holds}\\
\end{align*}

If $\bar A_{ij_1}>\bar A_{ij_2}$, we have the following:
\begin{align*}
    A_{ij_1} & \geq \bar A_{ij_1}-\Delta \\
    & \geq \bar A_{ij_2}+\tilde\Delta_{\min}-\Delta \tag{as $\bar A_{ij_1}-\bar A_{ij_2}\geq \tilde\Delta_{\min}$}\\
    & >\bar A_{ij_2} + \Delta \tag{as $\Delta\leq \frac{\tilde \Delta_{\min}}{10}$}\\
    & \geq A_{ij_2} \tag{as event $G$ holds}\\
\end{align*}

If $\bar A_{i_1j}>\bar A_{i_1j}$, we have the following:
\begin{align*}
    A_{i_1j} & \geq \bar A_{i_1j}-\Delta \\
    & \geq \bar A_{i_2j}+\tilde\Delta_{\min}-\Delta \tag{as $\bar A_{i_1j}-\bar A_{i_2j}\geq \tilde\Delta_{\min}$}\\
    & >\bar A_{i_2j} + \Delta \tag{as $\Delta\leq \frac{\tilde \Delta_{\min}}{10}$}\\
    & \geq A_{i_2j} \tag{as event $G$ holds}\\
\end{align*}
\end{proof}
\begin{corollary}\label{cor:alg2:saddle}
Let $t$ be the time step when the condition $1\leq \frac{\tilde \Delta_{\min}+2\Delta}{\tilde \Delta_{\min}-2\Delta}\leq \frac{3}{2}$ holds true for the first time. If the event $G$ holds, then at any time step $t_0$ such that  $t\leq t_0\leq T$, we have the following:
\begin{itemize}
    \item $(i,j)$ is PSNE of $A$ if and only if $(i,j)$ is a PSNE of $\bar A$.    
    \item $A$ does not have a PSNE if and only if $\bar A$ does not have a PSNE.
\end{itemize}
\end{corollary}

The following lemma bounds the ratio $\frac{\tilde D}{|D|}$.
\begin{lemma}\label{lem:alg2:D:ratio}
Let $t$ be the time step when the condition $1\leq \frac{\tilde \Delta_{\min}+2\Delta}{\tilde \Delta_{\min}-2\Delta}\leq \frac{3}{2}$ holds true for the first time. If the event $G$ holds and $A$ has a unique Equilibrium which is not a PSNE, then $\frac{5}{6}\leq \frac{\tilde D}{|D|} \leq \frac{5}{4}$ at the time step $t$.
\end{lemma}
\begin{proof}
Let us assume that the event $G$ holds. Then for every element $(i,j)$, we have  $|A_{ij}-\bar A_{ij}|\leq \Delta$. As $\frac{\tilde \Delta_{\min}+2\Delta}{\tilde \Delta_{\min}-2\Delta}\leq \frac{3}{2}$ and $2\tilde\Delta_{\min}\leq \tilde D$, we have $\Delta\leq \frac{\tilde\Delta_{\min}}{10}\leq \frac{\tilde D}{20}$.
Now observe that $\frac{\tilde D}{|D|}\leq \frac{\tilde D}{\tilde D-4\Delta}\leq \frac{\tilde D}{4\tilde D/5}= \frac{5}{4}$. Next observe that $\frac{\tilde D}{|D|}\geq \frac{\tilde D}{\tilde D+4\Delta}\geq\frac{\tilde D}{6\tilde D/5}= \frac{5}{6}$.
\end{proof}

The following two lemmas bound the ratio $\frac{\Delta_{m_2}}{|D|}$ when certain conditions in the algorithm \ref{alg-ucb-2} hold true.
\begin{lemma}\label{lem:alg2:dm2/D:one}
If the condition in the line \ref{alg2:con2} of the algorithm \ref{alg-ucb-2} holds true and event $G$ holds, then $\frac{\Delta_{m_2}}{|D|}\geq\frac{1}{12}$
\end{lemma}
\begin{proof}
Due to Lemma \ref{lem:alg2:dm2:ratio}, we have $\Delta_{m_2}\geq \frac{4\tilde \Delta_{m_2}}{5}$. Due to Lemma \ref{lem:alg2:D:ratio}, we have $|D|\leq \frac{6\tilde D}{5}$. Hence, we have $\frac{\Delta_{m_2}}{D}\geq \frac{2\tilde\Delta_{m_2}}{3\tilde D}\geq \frac{1}{12}$. We get the latter inequality as the condition in the line \ref{alg2:con2} holds true.
\end{proof}

\begin{lemma}\label{lem:alg2:dm2/D:both}
If the condition in the line \ref{alg2:con3} of the algorithm \ref{alg-ucb-2} holds true and event $G$ holds, then $\frac{2\Delta_{m_2}}{3|D|}\leq \frac{\tilde\Delta_{m_2}}{\tilde D}\leq\frac{3\Delta_{m_2}}{2|D|}$.
\end{lemma}
\begin{proof}
Due to Lemma \ref{lem:alg2:dm2:ratio}, we have $\frac{5\Delta_{m_2}}{6}\leq \tilde \Delta_{m_2}\leq \frac{5\Delta_{m_2}}{4}$. Due to Lemma \ref{lem:alg2:D:ratio}, we have $\frac{5|D|}{6}\leq \tilde D\leq \frac{5|D|}{4}$. Hence, we have $\frac{2\Delta_{m_2}}{3|D|}\leq \frac{\tilde\Delta_{m_2}}{\tilde D}\leq\frac{3\Delta_{m_2}}{2|D|}$.
\end{proof}
We now present the main lemma that establishes the correctness of the algorithm \ref{alg-ucb-2} when the input matrix $A$ does not have a PSNE.
\begin{lemma}\label{lem:alg2:main}
If the condition in the line \ref{alg2:con3} of the algorithm \ref{alg-ucb-2} holds true and event $G$ holds, then Nash equilibrium of the matrix $B$ is also an $\varepsilon$-Nash equilibrium of $A$
\end{lemma}
\begin{proof}
First observe that $\Delta_1\leq \sqrt{2\log(\frac{16T}{\delta}) /(\frac{200\tilde\Delta^2_{m_2}\log (\frac{16T}{\delta})}{\varepsilon^2 \tilde D^2})}=\frac{\varepsilon\tilde D}{10\tilde \Delta_{m_2}}$. Due to Lemma \ref{lem:alg2:dm2/D:both}, we then have $\Delta_1\leq \frac{3\varepsilon |D|}{20 \Delta_{m_2}}$. Let $\Delta_{ij}:=A_{ij}-B_{ij}$ for all $i,j$. As event $G$ holds true and due to the construction of the matrix $B$, we have $|\Delta_{ij}|\leq 3\Delta_1<\frac{\varepsilon |D|}{2\Delta_{m_2}}$. 

As $\frac{\tilde \Delta_{\min}+2\Delta}{\tilde \Delta_{\min}-2\Delta}\leq \frac{3}{2}$ and $\Delta_1\leq\Delta$, we have $\Delta_1\leq \frac{\tilde \Delta_{\min}}{10}$. Due to Lemma \ref{lem:alg2:dmin:ratio}, we then have $\Delta_1\leq \frac{\Delta_{\min}}{8}$. As the condition in the line \ref{alg2:con3} holds true and due to Lemma \ref{lem:alg2:dm2/D:both}, we have $\frac{\Delta_{m_2}}{|D|}\leq \frac{3\tilde\Delta_{m_2}}{2\tilde D}<\frac{3}{16}$.

Now we show that $B$ does not have PSNE. Recall that $B_{i_1j_1}= \bar A_{i_1j_1}$, $B_{i_2j_2}= \bar A_{i_2j_2}$, $B_{i_1j_2}= \bar A_{i_1j_2}-2\Delta_1$ and $B_{i_2j_1}= \bar A_{i_2j_1}+2\Delta_1$. Due to Corollary \ref{cor:alg2:saddle}, $\bar A$ does not have a PSNE. Hence, it suffices to show that $2\Delta_1<\min\{|\bar A_{11}-\bar A_{12}|,|\bar A_{21}-\bar A_{22}|\}|\bar A_{11}-\bar A_{21}|,|\bar A_{12}-\bar A_{22}|\}$. As event $G$ holds and due to Lemma \ref{lem:alg2:delta},  we have $\min\{|\bar A_{11}-\bar A_{12}|,|\bar A_{21}-\bar A_{22}|\}|\bar A_{11}-\bar A_{21}|,|\bar A_{12}-\bar A_{22}|\}\geq\Delta_{\min}-2\Delta_1$. As $\Delta_1\leq \frac{\Delta_{\min}}{8}$, we have $\Delta_{\min}-2\Delta_1>2\Delta_1$. Hence, $B$ does not have a PSNE.

Let $(x^*,y^*):=((x_1^*,x_2^*),(y_1^*,y_2^*))$ be the Nash equilibrium of $B$. Let $D_B=|B_{11}-B_{12}-B_{21}+B_{22}|$. Observe that $D_B=|\bar A_{11}-\bar A_{12}-\bar A_{21}+\bar A_{22}|\geq |D|-4\Delta_1$ as event $G$ holds true. Now we have the following:
\begin{align*}
    x_{i_2}^*&=\frac{|B_{i_11}-B_{i_12}|}{D_B}\\
    & \leq \frac{|\bar A_{i_11}-\bar A_{i_12}|+2\Delta_1}{D_B} \tag{Due to the construction of $B$}\\
    & \leq \frac{\tilde\Delta_{m_2}+2\Delta_1}{D_B} \tag{Due to the choice of $i_1$ in Algorithm \ref{alg-ucb-2}}\\
    & \leq \frac{\Delta_{m_2}+4\Delta_1}{|D|-4\Delta_1} \tag{as event $G$ holds}\\
    & \leq \frac{\Delta_{m_2}+\Delta_{m_2}/2}{|D|-|D|/4} \tag{as $\Delta_1\leq \frac{\Delta_{\min}}{8}\leq |D|/4$}\\
    & = \frac{2\Delta_{m_2}}{|D|}\\
    & <\frac{3}{8} \tag{as $\frac{\Delta_{m_2}}{|D|}<\frac{3}{16}$}
\end{align*}

Similarly we have the following:
\begin{align*}
    y_{j_2}^*&=\frac{|B_{1j_1}-B_{2j_1}|}{D_B}\\
    & \leq \frac{|\bar A_{1j_1}-\bar A_{2j_1}|+2\Delta_1}{D_B} \tag{Due to the construction of $B$}\\
    & \leq \frac{\tilde\Delta_{m_2}+2\Delta_1}{D_B} \tag{Due to the choice of $i_1$ in Algorithm \ref{alg-ucb-2}}\\
    & \leq \frac{\Delta_{m_2}+4\Delta_1}{|D|-4\Delta_1} \tag{as event $G$ holds}\\
    & \leq \frac{\Delta_{m_2}+\Delta_{m_2}/2}{|D|-|D|/4} \tag{as $\Delta_1\leq \frac{\Delta_{\min}}{8}\leq |D|/4$}\\
    & = \frac{2\Delta_{m_2}}{|D|}\\
    & <\frac{3}{8} \tag{as $\frac{\Delta_{m_2}}{|D|}<\frac{3}{16}$}
\end{align*}

Hence, we have shown that $i_1=\arg\max_ix_i^*$ and $j_1=\arg\max_i y_i^*$. 

Now we have $B_{i_1j_1}=\bar A_{i_1j_1} \leq A_{i_1j_1}+\Delta_1=B_{i_1j_1}+\Delta_{i_1j_1}+\Delta_1$. We also have $B_{i_2j_1}=\bar A_{i_2j_1}+2\Delta_1\geq A_{i_2j_1}+\Delta_1=B_{i_2j_1}+\Delta_{i_2j_1}+\Delta_1$. Hence we have $\Delta_{i_1j_1}\geq -\Delta_1 \geq \Delta_{i_2j_1}$.

Similarly, we have $B_{i_1j_1}=\bar A_{i_1j_1} \geq A_{i_1j_1}-\Delta_1=B_{i_1j_1}+\Delta_{i_1j_1}-\Delta_1$. We also have $B_{i_1j_2}=\bar A_{i_1j_2}-2\Delta_1\leq A_{i_1j_2}-\Delta_1=B_{i_1j_2}+\Delta_{i_1j_2}-\Delta_1$. Hence we have $\Delta_{i_1j_1}\leq \Delta_1 \leq \Delta_{i_1j_2}$.

Hence, all the conditions of the Lemma \ref{lem:improved:NE} is satisfied by the matrices $A$ and $B$. Hence, we can apply Lemma \ref{lem:improved:NE} and conclude that $(x^*,y^*)$ is an $\varepsilon$-Nash equilibrium of $A$.
\end{proof}

\subsection{Technical Lemma for Upper Bound}
In this section, we present an important technical lemma that is used to establish the upper bound on the sample complexity of finding $\varepsilon$-Nash equilibrium.

Let us first define two matrices $A_1$ and $A_2$ as follows:
\[
A_1 = \begin{bmatrix} 
   a & b \\
   c & d \\
    \end{bmatrix}
A_2 = \begin{bmatrix} 
   a+\Delta_{11} & b+\Delta_{12} \\
   c+\Delta_{21} & d+\Delta_{22} \\
    \end{bmatrix}\]
Let $(x^*,y^*):=((x_1^*,x_2^*),(y_1^*,y_2^*))$ be the unique Nash equilibrium of $A_1$. Let $\supp(x^*)=\supp(y^*)=\{1,2\}$. Let $i^*:=\arg\max_{i}x_i^*$ and $j^*:=\arg\max_jy_j^*$. Recall that $\Delta_{m_2}:=\max\{\min\{|a-b|, |d-c|\},\min\{|a-c|,|d-b|\}\}$ and  $D:=a-b-c+d$. Now we present the technical lemma.
\begin{lemma}\label{lem:improved:NE}
Let $|\Delta_{ij}|\leq\frac{\varepsilon |D|}{2\Delta_{m_2}}$ for all $i,j$. If $\Delta_{i^*j^*}\geq \Delta_{ij^*}$ for all $i$ and $\Delta_{i^*j^*}\leq \Delta_{i^*j}$ for all $j$, then $(x^*,y^*)$ is an $\varepsilon$-Nash equilibrium of matrix $A_2$.
\end{lemma}
\begin{proof}
First, observe that $\langle x^*,A_2y^* \rangle=\frac{ad-bc}{D}+\sum_{i,j}x_i^*y_j^*\Delta_{ij}$. Let $(A_2)^r_i$ denote the $i$-th row of $A_2$ and $(A_2)^c_j$ denote the $j$-th column of $A_2$. Now observe that $\langle (A_2)^r_1,y^*\rangle = \frac{ad-bc}{D}+y_1^*\Delta_{11}+y_2^*\Delta_{12}$ and $\langle (A_2)^r_2,y^*\rangle = \frac{ad-bc}{D}+y_1^*\Delta_{21}+y_2^*\Delta_{22}$. Finally, observe that $\langle (A_2)^c_1,x^*\rangle = \frac{ad-bc}{D}+x_1^*\Delta_{11}+x_2^*\Delta_{21}$ and $\langle (A_2)^c_2,x^*\rangle = \frac{ad-bc}{D}+x_1^*\Delta_{12}+x_2^*\Delta_{22}$.

W.l.o.g let us assume that $i^*=1$ and $j^*=1$. Now we have the following:
\begin{align*}
\langle (A_2)^r_1,y^*\rangle - \langle x^*,A_2y^* \rangle &= x_2^*y_1^*(\Delta_{11}-\Delta_{21})+x_2^*y_2^*(\Delta_{12}-\Delta_{22})\\
& \leq x_2^*y_1^*(|\Delta_{11}|+|\Delta_{21}|)+x_2^*y_2^*(|\Delta_{12}|+|\Delta_{22}|)\\
& \leq \frac{\varepsilon |D|}{\Delta_{m_2}}(x_2^*y_1^*+x_2^*y_2^*)\\
&= x_2^*\frac{\varepsilon |D|}{\Delta_{m_2}}\\
&= \frac{|a-b|}{|D|}\frac{\varepsilon |D|}{\Delta_{m_2}}\\
&\leq \varepsilon \tag{as $|a-b|\leq \Delta_{m_2}$}
\end{align*}
\begin{align*}
\langle (A_2)^r_2,y^*\rangle - \langle x^*,A_2y^* \rangle &= x_1^*y_1^*(\Delta_{21}-\Delta_{11})+x_1^*y_2^*(\Delta_{22}-\Delta_{12})\\
&\leq x_1^*y_2^*(\Delta_{22}-\Delta_{12}) \tag{as $\Delta_{11}\geq \Delta_{21}$}\\
& \leq x_1^*y_2^*(|\Delta_{22}|+|\Delta_{12}|)\\
& \leq \frac{\varepsilon |D|}{\Delta_{m_2}}x_1^*y_2^*\\
& \leq \frac{|a-c|}{|D|}\frac{\varepsilon |D|}{\Delta_{m_2}} \tag{as $x_1^*\leq 1$ and $y_2^*=\frac{|a-c|}{|D|}$}\\
& \leq \varepsilon \tag{as $|a-c|\leq \Delta_{m_2}$}
\end{align*}

\begin{align*}
\langle x^*,A_2y^* \rangle-\langle (A_2)^c_1,x^*\rangle &= x_1^*y_2^*(\Delta_{12}-\Delta_{11})+x_2^*y_2^*(\Delta_{22}-\Delta_{21})\\
& \leq x_1^*y_2^*(|\Delta_{12}|+|\Delta_{11}|)+x_2^*y_2^*(|\Delta_{22}|+|\Delta_{21}|)\\
& \leq \frac{\varepsilon |D|}{\Delta_{m_2}}(x_1^*y_2^*+x_2^*y_2^*)\\
&= y_2^*\frac{\varepsilon |D|}{\Delta_{m_2}}\\
&= \frac{|a-c|}{|D|}\frac{\varepsilon |D|}{\Delta_{m_2}}\\
&\leq \varepsilon \tag{as $|a-c|\leq \Delta_{m_2}$}
\end{align*}
\begin{align*}
\langle x^*,A_2y^* \rangle-\langle (A_2)^c_2,x^*\rangle &= x_1^*y_1^*(\Delta_{11}-\Delta_{12})+x_2^*y_1^*(\Delta_{21}-\Delta_{22})\\
&\leq x_2^*y_1^*(\Delta_{21}-\Delta_{22}) \tag{as $\Delta_{11}\leq \Delta_{12}$}\\
& \leq x_2^*y_1^*(|\Delta_{21}|+|\Delta_{22}|)\\
& \leq \frac{\varepsilon |D|}{\Delta_{m_2}}x_2^*y_1^*\\
& \leq \frac{|a-b|}{|D|}\frac{\varepsilon |D|}{\Delta_{m_2}} \tag{as $y_1^*\leq 1$ and $x_2^*=\frac{|a-b|}{|D|}$}\\
& \leq \varepsilon \tag{as $|a-b|\leq \Delta_{m_2}$}
\end{align*}
\end{proof}
\section{Proof of $\varepsilon$-Nash equilibrium Lower Bound}\label{appendix:thm4}
Before finishing the proof of the theorem \ref{thm:lower3}, we begin with the proof of Lemma~\ref{low3:lem1}
\begin{proof}
For any $\alpha\in [\frac{c-d}{D},\frac{a-b}{D}]$ and $\beta\in[\frac{b-d}{D},\frac{a-c}{D}]$, let $x'=(\frac{d-c}{D}+\alpha,\frac{a-b}{D}-\alpha)$ and $y'=(\frac{d-b}{D}+\beta,\frac{a-c}{D}-\beta)$. Note that this parameterization ensures the range of $x',y'$ is equal to $\simplex_2$. It can be shown that $\langle x',A_\square y'\rangle=\frac{ad-bc}{D}+\frac{(d-c)-(a-b)}{D}\square+2\square \alpha + D\alpha\beta$. Let $(A_\square)^r_i$ denote the $i$-th row of $A_\square$ and $(A_\square)^c_j$ denote the $j$-th column of $A_\square$. It can be shown that $\langle x', (A_\square)^c_1\rangle= \frac{ad-bc}{D}+\frac{(d-c)-(a-b)}{D}\square+2\square \alpha+(a-c)\alpha$ and $\langle x', (A_\square)^c_2\rangle= \frac{ad-bc}{D}+\frac{(d-c)-(a-b)}{D}\square+2\square \alpha+(b-d)\alpha$. Similarly, it can be shown that $\langle y', (A_\square)^r_1\rangle= \frac{ad-bc}{D}+\square+(a-b)\beta$ and $\langle y', (A_\square)^r_2\rangle= \frac{ad-bc}{D}-\square-(d-c)\beta$. We refer the reader to the Appendix \ref{apendix:low3:lem1} for the detailed calculations.

Now we have the following:
\begin{align*}
    \langle y',(A_\square)^r_1\rangle - \langle x', A_\square y' \rangle &= \frac{2(a-b)}{D}\square+(a-b)\beta-2\square \alpha-D\alpha\beta\\
    \langle y',(A_\square)^r_2\rangle - \langle x', A_\square y' \rangle &= -\frac{2(d-c)}{D}\square-(d-c)\beta-2\square \alpha-D\alpha\beta\\
    \langle x', A_\square y' \rangle -  \langle x',(A_\square)^c_1\rangle &= D\alpha\beta-(a-c)\alpha\\
    \langle x', A_\square y' \rangle -  \langle x',(A_\square)^c_2\rangle &= D\alpha\beta+(d-b)\alpha\\
\end{align*}

Observe that if $(x',y')$ is an $\varepsilon$-Nash equilibrium of $A_\square$, then $\langle y',(A_\square)^r_1\rangle - \langle x', A_\square y' \rangle\leq \varepsilon$, $\langle y',(A_\square)^r_2\rangle - \langle x', A_\square y' \rangle\leq \varepsilon$, $\langle x', A_\square y' \rangle -  \langle x',(A_\square)^c_1\rangle\leq \varepsilon$ and $ \langle x', A_\square y' \rangle -  \langle x',(A_\square)^c_2\rangle\leq \varepsilon$.

We will now show that regardless of what values $(\alpha,\beta)$ take (equivalently, regardless of what values $(x',y')$ take), there is at least one of the three alternative matrices for which $(x',y')$ is not an $\varepsilon$-Nash equilibrium. Let us assume that $(x',y')$ is an $\varepsilon$-Nash equilibrium of $A$, otherwise $A_0$ is the matrix for which $(x',y')$ is not an $\varepsilon$-Nash equilibrium. Now we have the following:
\begin{align*}
    (a-b)\beta-D\alpha\beta&\leq \varepsilon \\
    -(d-c)\beta-D\alpha\beta&\leq \varepsilon \\
    -(a-c)\alpha+D\alpha\beta&\leq \varepsilon \\
    (d-b)\alpha+D\alpha\beta&\leq \varepsilon \\
\end{align*}

Using the above equations, we get the following:
\begin{align*}
    -\varepsilon \leq D\alpha\beta \leq \varepsilon  \\
    \beta\geq -\frac{2\varepsilon}{d-c}\\
    \beta\leq \frac{2\varepsilon}{a-b}
\end{align*}

If $\Delta \alpha < \varepsilon$, we have the following:
\begin{align*}
\langle y',(A_\Delta)^r_1\rangle - \langle x', A_\Delta y' \rangle&=\frac{2(a-b)}{D}\Delta+(a-b)\beta-2\Delta \alpha-D\alpha\beta\\
&\geq \frac{2(a-b)}{D}\Delta+(a-b)\beta-2\Delta \alpha-\varepsilon \tag{as $-D\alpha\beta\geq -\varepsilon$}\\
&\geq \frac{2(a-b)}{D}\Delta-2\Delta \alpha-3\varepsilon \tag{as $\beta\geq -\frac{2\varepsilon}{d-c}$}\\
& > \frac{2(a-b)}{D}\Delta - 5\varepsilon \tag{as $-\Delta \alpha> -\varepsilon$} \\
&= \varepsilon \tag{as $\Delta= \frac{3\varepsilon D}{a-b}$}
\end{align*}

If $\Delta \alpha \geq \varepsilon$, we have the following:
\begin{align*}
\langle y',(A_{-\Delta})^r_2\rangle - \langle x', A_{-\Delta} y' \rangle&=\frac{2(d-c)}{D}\Delta-(d-c)\beta+2\Delta \alpha-D\alpha\beta\\
&\geq \frac{2(d-c)}{D}\Delta-(d-c)\beta+2\Delta \alpha-\varepsilon \tag{as $-D\alpha\beta\geq -\varepsilon$}\\
& \geq \frac{2(d-c)}{D}\Delta-(d-c)\beta+\varepsilon \tag{as $\Delta \alpha\geq\varepsilon$} \\
&\geq \frac{2(d-c)}{D}\Delta-\frac{2(d-c)}{a-b}\varepsilon+\varepsilon \tag{as $-\beta\geq -\frac{2\varepsilon}{a-b}$}\\
& = \frac{6(d-c)}{a-b}\varepsilon-\frac{2(d-c)}{a-b}\varepsilon+\varepsilon \tag{as $\Delta= \frac{3\varepsilon D}{a-b}$}\\
&> \varepsilon 
\end{align*}

Hence, we proved that for any $(x',y') \in \simplex_2 \times \simplex_2$, there exists a matrix $B\in\{A_0,A_{\Delta},A_{-\Delta}\}$ such that $(x',y')$ is not an $\varepsilon$-Nash equilibrium of $B$.
\end{proof}


\subsection{Proof of Theorem \ref{thm:lower3}}
Let $\nu_{i,j}^A = \mathcal{N}(A_{ij},1)$ be the distribution of an observation when playing pair $(i,j)$ with matrix $A$. 
Let $\P_A$ denote the probability law of the internal randomness of the algorithm and random observations. Let $f_A(x,y)=\max\{\max_{x'\in \simplex_2}\langle x',Ay\rangle-\langle x,Ay\rangle,\langle x,Ay\rangle-\min_{y'\in\simplex_2}\langle x,Ay'\rangle\}$.
If an algorithm is $(\varepsilon,\delta)$-PAC-Nash and outputs a solution $(\widehat{x},\widehat{y})$ then $\min_A \P_A(f_A(\widehat{x},\widehat{y})\leq \varepsilon) \geq 1-\delta$.
We will show that if an algorithm is $(\varepsilon,\delta)$-PAC-Nash then it can also accomplish a particular hypothesis.
We will conclude by noting that any procedure that can accomplish the hypothesis test must take the claimed sample complexity. 

For any pair of mixed strategies $(\widehat{x},\widehat{y})$ output by the procedure at the stopping time $\tau$, define 
\begin{align*}
    \phi = \{A_{-\Delta},A_{0},A_{\Delta}\} \setminus \arg\max_{B \in \{A_{-\Delta},A_{0},A_{\Delta}\}} f_B(\widehat{x},\widehat{y}),
\end{align*}
breaking ties arbitrarily in the maximum so that $\phi \in  \{A_{-\Delta},A_{0}\}\cup \{A_{-\Delta},A_{\Delta}\}\cup \{A_{0},A_{\Delta}\}$.
Note that
\begin{align}\label{eqn:correctness_delta:3}
    \P_{A_0}( A_0 \in \phi ) \geq \P_{A_0}( A_0 \in \phi,  f_{A_0}(\widehat{x},\widehat{y}) \leq \varepsilon) = \P_{A_0}(   f_{A_0}(\widehat{x},\widehat{y}) \leq \varepsilon) \geq 1-\delta
\end{align}
where the equality follows from the Lemma~\ref{low3:lem1}: at least one of the three matrices must have a loss of more than $\varepsilon$, but $A_0$ has a loss of at most $\varepsilon$, thus $A_0 \in \phi$.
Now because
\begin{align*}
    2 \max\{ \P_{A_0}( \phi = \{A_0,A_{-\Delta}\} ) , \P_{A_0}( \phi = \{A_0,A_{\Delta}\} ) \} &\geq \P_{A_0}( \phi = \{A_0,A_{-\Delta}\} ) + \P_{A_0}( \phi = \{A_0,A_{\Delta}\} ) \\
    &= \P_{A_0}( A_0 \in \phi ) \geq 1-\delta
\end{align*}
we have that $\P_{A_0}( \phi = \{A_0,A_{-\Delta}\} ) \geq \frac{1-\delta}{2}$ or $\P_{A_0}( \phi = \{A_0,A_{\Delta} \}) \geq \frac{1-\delta}{2}$.
Let's assume the former (the latter case is handled identically).
By the same argument as \eqref{eqn:correctness_delta:3} we have that $\P_{A_\Delta}( \phi = \{A_0,A_{-\Delta}\} ) \leq \delta$.

For a stopping time $\tau$, let $N_{i,j}(\tau)$ denote the number of times $(i,j)$ is sampled. Recalling that $\nu_{i,j}^A = \mathcal{N}(A_{ij},1)$, we have by Lemma 1 of \citet{kaufmann2016complexity} that
\begin{align*}
\sum_{i=1}^2\sum_{j=1}^2\mathbb{E}_{A_0}[ N_{i,j}(\tau) ] KL( \nu_{i,j}^{A_0}, \nu_{i,j}^{A_{2\Delta}} ) \geq d( \P_{A_0}( \phi = \{A_0,A_{-\Delta}\} ), \P_{A_\Delta}( \phi = \{A_0,A_{-\Delta}\} ) )
\end{align*}
where for any $i,j\in\{1,2\}$, $KL( \nu_{i,j}^{A_0}, \nu_{i,j}^{A_\Delta} ) = \Delta^2/2$ and $d(p,q) = p \log(\frac{p}{q}) + (1-p) \log(\frac{1-p}{1-q})$.
Since 
\begin{align*}
    d( \P_{A_0}( \phi = \{A_0,A_{-\Delta}\} ), \P_{A_\Delta}( \phi = \{A_0,A_{-\Delta}\} ) ) &\geq d( \tfrac{1-\delta}{2},\delta) \\
    &= \tfrac{1-\delta}{2} \log( \tfrac{1-\delta}{2 \delta} ) + \tfrac{1+\delta}{2} \log( \tfrac{1+\delta}{2(1-\delta)} ) \\
    &= \tfrac{1}{2} \log( \tfrac{1+\delta}{4 \delta} ) -  \tfrac{\delta}{2} \log( \tfrac{(1-\delta)^2}{ \delta (1+\delta)} ) \\
    &\geq \tfrac{1}{2} \log( \tfrac{1+\delta}{4 \delta} ) -  1/8 > \tfrac{1}{2} \log(1/30\delta)
\end{align*} 
and $\tau = N_{1,1}(\tau)+N_{1,2}(\tau)+N_{2,1}(\tau)+N_{2,2}(\tau)$ we conclude that
\begin{align*}
    \mathbb{E}_{A_0}[ \tau ] \geq \frac{ \log(1/30 \delta) }{\Delta^2} = \frac{\Delta_{m_2}^2 \log(1/30 \delta) }{9\varepsilon^2 D^2}
\end{align*}
as claimed.

\subsection{Calculations for Lemma \ref{low3:lem1}}\label{apendix:low3:lem1}
Recall that $x'=(\frac{d-c}{D}+\alpha,\frac{a-b}{D}-\alpha)$ and $y'=(\frac{d-b}{D}+\beta,\frac{a-c}{D}-\beta)$. Let $(A_\square)^r_i$ denote the $i$-th row of $A_\square$ and $(A_\square)^c_j$ denote the $j$-th column of $A_\square$. 

First, observe that $\langle y', (A_\square)^r_1\rangle= \frac{d-b}{D}\cdot a + \frac{d-b}{D}\cdot \square+a\beta+\square \beta+\frac{a-c}{D}\cdot b + \frac{a-c}{D}\cdot \square-b\beta-\square \beta=\frac{ad-bc}{D}+\square+(a-b)\beta$. Similarly, we have  $\langle y', (A_\square)^r_2\rangle= \frac{d-b}{D}\cdot c - \frac{d-b}{D}\cdot \square+c\beta-\square \beta+\frac{a-c}{D}\cdot d - \frac{a-c}{D}\cdot \square-d\beta+\square \beta=\frac{ad-bc}{D}-\square+(c-d)\beta$.

Next, observe that $\langle x', (A_\square)^c_1\rangle=\frac{d-c}{D}\cdot a + \frac{d-c}{D}\cdot \square+a\alpha+\square \alpha+\frac{a-b}{D}\cdot c - \frac{a-b}{D}\cdot \square-c\alpha+\square \alpha=\frac{ad-bc}{D}+\frac{(d-c)-(a-b)}{D}\square+(a-c)\alpha+2\square \alpha$. Similarly, we have $\langle x', (A_\square)^c_2\rangle=\frac{d-c}{D}\cdot b + \frac{d-c}{D}\cdot \square+b\alpha+\square \alpha+\frac{a-b}{D}\cdot d - \frac{a-b}{D}\cdot \square-d\alpha+\square \alpha=\frac{ad-bc}{D}+\frac{(d-c)-(a-b)}{D}\square+(b-d)\alpha+2\square \alpha$.

Now we present the following proposition.
\begin{proposition}\label{low3:prop1}
$\langle x',A_\square y'\rangle=\frac{ad-bc}{D}+\frac{(d-c)-(a-b)}{D}\square+2\square \alpha + D\alpha\beta$
\end{proposition}
\begin{proof}
Let $V_1=\langle y', (A_\square)^r_1\rangle$ and $V_2=\langle y', (A_\square)^r_2\rangle$. Now observe that $\langle x',A_\square y'\rangle=\langle x',(V_1,V_2)\rangle$. Now we have the following:
\begin{align}
    \langle x',(V_1,V_2)\rangle&=\left\langle x', \left(\frac{ad-bc}{D},\frac{ad-bc}{D}\right)\right\rangle+\left(\frac{d-c}{D}+\alpha\right)\cdot\square-\left(\frac{a-b}{D}-\alpha\right)\cdot\square \nonumber\\
    &\quad+\left(\frac{d-c}{D}+\alpha\right)\cdot(a-b)\beta-\left(\frac{a-b}{D}-\alpha\right)\cdot(d-c)\beta\nonumber\\
    &=\frac{ad-bc}{D}+\frac{(d-c)-(a-b)}{D}\square+2\square \alpha \nonumber \\
    &\quad+\frac{(d-b)(a-b)-(a-b)(d-c)}{D}\cdot \beta +(a-b-c+d)\alpha\beta \nonumber\\
    &= \frac{ad-bc}{D}+\frac{(d-c)-(a-b)}{D}\square+2\square \alpha + D\alpha\beta\tag{as $D=a-b-c+d$}
\end{align}
\end{proof}

\section{Proof of $n\times 2$ Matrix Upper Bound}\label{appendix:thm7}
We now establish the sample complexity and the correctness of the Algorithm \ref{alg-ucb-3} by proving the Theorem \ref{thm:alg3}.
\begin{proof}[Proof of Theorem \ref{thm:alg3}]
Let $\bar A_{ij,t}$ denote the empirical mean of $A_{ij}$ at time step $t$. Let us begin by defining two events:
\begin{align*}
    G:=&\bigcap_{t=1}^T\bigcap_{i=1}^n \bigcap_{j=1}^2 \{ |A_{ij}-\bar A_{ij,t}|\leq \sqrt{\tfrac{2\log({8nT}/{\delta})}{t}} \} \\
    E:=&\bigcap_{i=1}^n \bigcap_{j=1}^2 \{ |A_{ij}-\bar A_{ij,T}|\leq \sqrt{\tfrac{2\log({8n}/{\delta})}{T}} \} 
\end{align*}
A union bound and sub-Gaussian-tail bound demonstrates that $\P( G^c \cup E^c ) \leq \P(G^c) + \P(E^c) \leq \delta$.
Consequently, events $E$ and $G$ hold simultaneously with probability at least $1-\delta$, so in what follows, assume they hold. 

If $A$ has a PSNE and if the condition in the line \ref{alg3:con1} of the algorithm \ref{alg-ucb-3} is satisfied, then we identify a PSNE in $\frac{800\log (\frac{8nT}{ \delta})}{ \Delta_{\min}^2}$ time steps due to Lemma \ref{lem:alg3:tmin} and Corollary \ref{cor:alg3:saddle}. On the other hand, if $A$ has a PSNE but the outer for loop completes after $t=T$ iterations,  then we identify an $\varepsilon$-good solution in $T=\frac{8 \log (8n/\delta) }{\varepsilon^2}$ time steps due to Lemma \ref{lem:trivial:NE}. Note that in this case, $T<\frac{800\log (\frac{8nT}{ \delta})}{ \Delta_{\min}^2}$ due to Lemma \ref{lem:alg3:tmin}. 

Let us assume for the rest of the proof that $A$ does not have a PSNE. If the outer for loop completes after $t=T$ iterations,  then we identify an $\varepsilon$-good solution in $T=\frac{8 \log (8n/\delta) }{\varepsilon^2}$ time steps due to Lemma \ref{lem:trivial:NE}. Note that in this case, $T<\frac{800\log (\frac{8nT}{ \delta})}{ \Delta_{\min}^2}$ due to Lemma \ref{lem:alg3:tmin}.

Now let us  assume for the rest of the proof that the condition in the line \ref{alg3:con2} is satisfied. If the condition in the line \ref{alg3:con3} is satisfied, then we identify an $\varepsilon$-Nash equilibrium in $T=\frac{8 \log (8nT/\delta) }{\varepsilon^2}$ time steps due to Lemma \ref{lem:trivial:NE}. Now observe that in this case $T\leq \max\left\{\frac{800\log (\frac{8nT}{ \delta})}{ \Delta_{\min}^2},\frac{722\log (\frac{8nT}{ \delta})}{ \Delta_{g}^2}\right\}$ due to Lemma \ref{lem:alg3:tmin} and Lemma \ref{lem:alg3:dg}. If the condition in the line \ref{alg3:con5} is satisfied, then we identify $\supp(x^*)$ and $\supp(y^*)$ due to Lemma \ref{lem:alg3:main}. In this case, let the number of times we are required to sample each element be $n_0$. Then $n_0\leq \max\left\{\frac{800\log (\frac{8nT}{ \delta})}{ \Delta_{\min}^2},\frac{722\log (\frac{8nT}{ \delta})}{ \Delta_{g}^2}\right\}+1$ due to Lemma \ref{lem:alg3:tmin} and Lemma \ref{lem:alg3:dg}.
\end{proof}

\subsection{Consequential lemmas of Algorithm~\ref{alg-ucb-3}'s conditional statements}
Recall the definitions of events $E$ and $G$. We first present few lemmas which deal with empirical estimates and instance dependent parameters like $\tilde \Delta_{\min},\tilde \Delta_g, \Delta_{\min}$ and $\Delta_g$ . Whenever we fix a time step $t\leq T$ and discuss the parameters like $\tilde  \Delta_{\min}, \tilde \Delta_g,\Delta'$ and $\Delta$, we consider those values that have been assigned to these parameters during the time step $t$. We begin with upper bounding $ |\Delta_{\min}-\tilde \Delta_{\min}|$ in the following lemma.
\begin{lemma}
Fix a time step $t\leq T$. If the event $G$ holds, then we have the following:
\begin{equation*}
    |\Delta_{\min}-\tilde \Delta_{\min}|\leq 2\Delta
\end{equation*}
\end{lemma}
\begin{proof}
Let us assume that the event $G$ holds. Then for every element $(i,j)$, we have  $|A_{ij}-\bar A_{ij}|\leq \Delta$. Then we have $\left||A_{ij}-A_{i'j'}|-|\bar A_{ij}-\bar A_{i'j'}|\right|\leq 2\Delta$ for any $i,j,i',j'$. By repeatedly applying the Lemma \ref{lem:deviation}, we get $|\Delta_{\min}-\tilde \Delta_{\min}|\leq 2\Delta$.
\end{proof}
The following lemma upper bounds the number of time steps required to satisfy the condition $1\leq \frac{\tilde \Delta_{\min}+2\Delta}{\tilde \Delta_{\min}-2\Delta}\leq \frac{3}{2}$.
\begin{lemma}\label{lem:alg3:tmin}
Let $t$ be the time step when the condition $1\leq \frac{\tilde \Delta_{\min}+2\Delta}{\tilde \Delta_{\min}-2\Delta}\leq \frac{3}{2}$ holds true for the first time. If the event $G$ holds, then $t\leq \frac{800\log (\frac{8nT}{ \delta})}{ \Delta_{\min}^2}$.
\end{lemma}
\begin{proof}
Consider the time step $t= \frac{800\log (\frac{8nT}{ \delta})}{ \Delta_{\min}^2}$. Let us assume that the event $G$ holds. Then for every element $(i,j)$, we have  $|A_{ij}-\bar A_{ij}|\leq \Delta=\sqrt{\frac{2\log(\frac{8nT}{\delta})}{t}}=\frac{\Delta_{\min}}{20}$. Now observe that $\tilde \Delta_{\min}+2\Delta\leq \Delta_{\min}+4\Delta = \frac{6\Delta_{\min}}{5}$. Similarly, we have $\tilde \Delta_{\min}-2\Delta\geq \Delta_{\min}-4\Delta \geq \frac{4\Delta_{\min}}{5}$. Hence, we have $1\leq \frac{\tilde \Delta_{\min}+2\Delta}{\tilde \Delta_{\min}-2\Delta}\leq \frac{3}{2}$. 

\end{proof}
The following lemma bounds the ratio $\frac{\tilde \Delta_{\min}}{\Delta_{\min}}$.
\begin{lemma}\label{lem:alg3:dmin:ratio}
Let $t$ be the time step when the condition $1\leq \frac{\tilde \Delta_{\min}+2\Delta}{\tilde \Delta_{\min}-2\Delta}\leq \frac{3}{2}$ holds true for the first time. If the event $G$ holds, then $\frac{5}{6}\leq \frac{\tilde \Delta_{\min}}{\Delta_{\min}} \leq \frac{5}{4}$ at the time step $t$.
\end{lemma}
\begin{proof}
Let us assume that the event $G$ holds. Then for every element $(i,j)$, we have  $|A_{ij}-\bar A_{ij}|\leq \Delta$. As $\frac{\tilde \Delta_{\min}+2\Delta}{\tilde \Delta_{\min}-2\Delta}\leq \frac{3}{2}$, we have $\Delta\leq \frac{\tilde \Delta_{\min}}{10}$.
Now observe that $\frac{\tilde \Delta_{\min}}{\Delta_{\min}}\leq \frac{\tilde\Delta_{\min}}{\tilde\Delta_{\min}-2\Delta}\leq \frac{\tilde\Delta_{\min}}{4\tilde\Delta_{\min}/5}= \frac{5}{4}$. Next observe that $\frac{\tilde \Delta_{\min}}{\Delta_{\min}}\geq \frac{\tilde \Delta_{\min}}{\tilde \Delta_{\min}+2\Delta}\geq\frac{\tilde \Delta_{\min}}{6\tilde \Delta_{\min}/5}= \frac{5}{6}$.
\end{proof}
\noindent
Now let us define the notion of strong dominance.
\begin{definition}[Strongly dominate]
We say that a row $i$ of a matrix $A$ strongly dominates a row $j$ of $A$ if $A_{i1}>A_{j1}$ and $A_{i2}>A_{j2}$
\end{definition}

\noindent
The following lemma and the subsequent corollary relates the empirical matrix $\bar A$ to the input matrix $A$.
\begin{lemma}\label{lem:alg3:gap}
Let $t$ be the time step when the condition $1\leq \frac{\tilde \Delta_{\min}+2\Delta}{\tilde \Delta_{\min}-2\Delta}\leq \frac{3}{2}$ holds true for the first time. If the event $G$ holds, then at any time step $t_0$ such that  $t\leq t_0\leq T$, we have the following:
\begin{itemize}
    \item If $A_{ij_1}>A_{ij_2}$, then $\bar A_{ij_1}>\bar A_{ij_2}$ 
    \item If $A_{i_1j}>A_{i_1j}$, then $\bar A_{i_1j}>\bar A_{i_1j}$
    \item If $\bar A_{ij_1}>\bar A_{ij_2}$, then $ A_{ij_1}> A_{ij_2}$ 
    \item If $\bar A_{i_1j}>\bar A_{i_1j}$, then $ A_{i_1j}> A_{i_1j}$
\end{itemize}
\end{lemma}
\begin{proof}
As $\frac{\tilde \Delta_{\min}+2\Delta}{\tilde \Delta_{\min}-2\Delta}\leq \frac{3}{2}$, we have $\Delta\leq \frac{\tilde \Delta_{\min}}{10}$. Due to Lemma \ref{lem:alg3:dmin:ratio}, we have $\Delta\leq \frac{\Delta_{\min}}{8}$. As event $G$ holds, for any element $(i,j)$, we have $|A_{ij}-\bar A_{ij}|\leq \sqrt{\frac{2\log(\frac{8nT}{\delta})}{t_0}}\leq \Delta $.

If $A_{ij_1}>A_{ij_2}$, we have the following:
\begin{align*}
    \bar A_{ij_1} & \geq A_{ij_1}-\Delta \\
    & \geq A_{ij_2}+\Delta_{\min}-\Delta \tag{as $A_{ij_1}- A_{ij_2}\geq \Delta_{\min}$}\\
    & > A_{ij_2} + \Delta \tag{as $\Delta\leq \frac{ \Delta_{\min}}{8}$}\\
    & \geq \bar A_{ij_2} \tag{as event $G$ holds}\\
\end{align*}

If $A_{i_1j}>A_{i_2j}$, we have the following:
\begin{align*}
    \bar A_{i_1j} & \geq A_{i_1j}-\Delta \\
    & \geq A_{i_2j}+\Delta_{\min}-\Delta \tag{as $A_{i_1j}- A_{i_2j}\geq \Delta_{\min}$}\\
    & > A_{i_2j} + \Delta \tag{as $\Delta\leq \frac{ \Delta_{\min}}{8}$}\\
    & \geq \bar A_{i_2j} \tag{as event $G$ holds}\\
\end{align*}

If $\bar A_{ij_1}>\bar A_{ij_2}$, we have the following:
\begin{align*}
    A_{ij_1} & \geq \bar A_{ij_1}-\Delta \\
    & \geq \bar A_{ij_2}+\tilde\Delta_{\min}-\Delta \tag{as $\bar A_{ij_1}-\bar A_{ij_2}\geq \tilde\Delta_{\min}$}\\
    & >\bar A_{ij_2} + \Delta \tag{as $\Delta\leq \frac{\tilde \Delta_{\min}}{10}$}\\
    & \geq A_{ij_2} \tag{as event $G$ holds}\\
\end{align*}

If $\bar A_{i_1j}>\bar A_{i_1j}$, we have the following:
\begin{align*}
    A_{i_1j} & \geq \bar A_{i_1j}-\Delta \\
    & \geq \bar A_{i_2j}+\tilde\Delta_{\min}-\Delta \tag{as $\bar A_{i_1j}-\bar A_{i_2j}\geq \tilde\Delta_{\min}$}\\
    & >\bar A_{i_2j} + \Delta \tag{as $\Delta\leq \frac{\tilde \Delta_{\min}}{10}$}\\
    & \geq A_{i_2j} \tag{as event $G$ holds}\\
\end{align*}
\end{proof}
\begin{corollary}\label{cor:alg3:saddle}
Let $t$ be the time step when the condition $1\leq \frac{\tilde \Delta_{\min}+2\Delta}{\tilde \Delta_{\min}-2\Delta}\leq \frac{3}{2}$ holds true for the first time. If the event $G$ holds, then at any time step $t_0$ such that  $t\leq t_0\leq T$, we have the following:
\begin{itemize}
    \item $(i,j)$ is PSNE of $A$ if and only if $(i,j)$ is a PSNE of $\bar A$.    
    \item $A$ does not have a PSNE if and only if $\bar A$ does not have a PSNE.
    \item The row $i$ of $A$ strongly dominates the row $j$ of $A$ if and only if the row $i$ of $A$ strongly dominates the row $j$ of $\bar A$.
\end{itemize}
\end{corollary}
We now present the main lemma that establishes the correctness of the algorithm \ref{alg-ucb-1} when the input matrix $A$ does not have a PSNE.
\begin{lemma}\label{lem:alg3:main}
Let $(x^*,y^*)$ be the Nash equilibrium of $A$. If the condition in the line \ref{alg3:con5} of the algorithm \ref{alg-ucb-3} holds true and event $G$ holds, then $\{i_1,i_2\}=\supp(x^*)$ and $\{1,2\}=\supp(y^*)$.
\end{lemma}
\begin{proof}
As condition in the line \ref{alg3:con5} of the algorithm \ref{alg-ucb-3} holds true and event $G$ holds, $\bar A$ does not have a PSNE. Due to Corollary \ref{cor:alg3:saddle}, $A$ does not have a PSNE. Hence $\{1,2\}=\supp(y^*)$. Moreover, $|\supp(x^*)|=2$ as $A$ has a unique Nash equilibrium.

As $\bar A$ has no row that strongly dominates any other row, $A$ does not have a row that strongly dominates any other row due to Corollary \ref{cor:alg3:saddle}. As $\frac{\tilde \Delta_{\min}+2\Delta}{\tilde \Delta_{\min}-2\Delta}\leq \frac{3}{2}$, we have $\Delta\leq \frac{\tilde \Delta_{\min}}{10}$. Due to Lemma \ref{lem:alg3:dmin:ratio}, we have $\Delta'\leq \Delta\leq \frac{\Delta_{\min}}{8}$. Therefore, for any $i,j,k\in[n]$, $\frac{6\Delta'}{|A_{i1}-A_{i2}|+|A_{j1}-A_{j2}|+|A_{k1}-A_{k2}|}\leq \frac{6\Delta'}{3\Delta_{\min}}\leq \frac{1}{4}$. Hence, we can apply the lemmas in the Section \ref{tools:n*2}.  

Now let us assume that $\{i_1,i_2\}\neq \supp(x^*)$. Let $B=[A_{i_11},A_{i_12};A_{i_21},A_{i_22}]$ and $(x_B,y_B)$ is the Nash equilibrium of $B$. Then due to Lemma \ref{lem:n*2:suboptimal}, $\exists i\in \supp(x^*)\setminus\{i_1,i_2\}$ such that $V_{B}^*-\langle y_B, (A_{i1},A_{i2})\rangle<0$. Now we have the following:
\begin{align*}
    \tilde \Delta_g &\leq \frac{(|\bar A_{i_11}-\bar A_{i_12}|+|\bar A_{i_21}-\bar A_{i_22}|)(V^*_{\bar A}-\langle y',(\bar A_{i1},\bar A_{i2})\rangle)}{|\bar A_{i_11}-\bar A_{i_12}|+|\bar A_{i_21}-\bar A_{i_22}|+|\bar A_{i1}-\bar A_{i2}|}\\
    & \leq \frac{(| A_{i_11}- A_{i_12}|+| A_{i_21}- A_{i_22}|)(V^*_{B}-\langle y_B,( A_{i1}, A_{i2})\rangle)}{| A_{i_11}- A_{i_12}|+| A_{i_21}- A_{i_22}|+| A_{i1}- A_{i2}|} + 4\Delta' \tag{due to Lemma \ref{lem:n*2:deviation}}\\
    & < 4\Delta' \tag{as $V_{B}^*-\langle y_B, (A_{i1},A_{i2})\rangle<0$}
\end{align*}
This contradicts the fact that the condition in the line \ref{alg3:con5} of the algorithm \ref{alg-ucb-3} holds true. Hence $\{i_1,i_2\}=\supp(x^*)$.
\end{proof}
The following lemma upper bounds the number of time steps required to return the support of the Nash equilibrium when the input matrix $A$ does not have a PSNE.
\begin{lemma}\label{lem:alg3:dg}
Let the condition in the line \ref{alg3:con2} of the algorithm \ref{alg-ucb-3} hold true. Let $t'$ be the time step when both the conditions $|\supp(x')|=2$ and $\tilde\Delta_{g}\geq 4\Delta'$ hold true simultaneously for the first time. If the event $G$ holds, then $t'\leq \frac{722\log (\frac{8nT}{ \delta})}{ \Delta_{g}^2}$.
\end{lemma}
\begin{proof}
Consider the time step $t'= \frac{722\log (\frac{8nT}{ \delta})}{ \Delta_{g}^2}$. Let us assume that the event $G$ holds. Then $\Delta'=\sqrt{2\log(\frac{8nT}{\delta}) / (t')}=\frac{\Delta_g}{19}$. Let $(x^*,y^*)$ be the Nash equilibrium of $A$. Let $\supp(x^*)=(i^*,j^*)$. Let $A_*=[A_{i^*1},A_{i^*2};A_{j^*1},A_{j^*2}]$ and $\bar A_*=[\bar A_{i^*1},\bar A_{i^*2}; \bar A_{j^*1},\bar A_{j^*2}]$. Let $(x'',y'')$ be the Nash equilibrium of $\bar A_*$. Then for any $i\notin \supp(x^*)$, we have the following due to Lemma \ref{lem:n*2:deviation}.
\begin{align*}
&\frac{(|\bar A_{i^*1}-\bar A_{i^*2}|+|\bar A_{j^*1}-\bar A_{j^*2}|)(V^*_{\bar A_*}-\langle y'',(\bar A_{i1},\bar A_{i2})\rangle)}{|\bar A_{i^*1}-\bar A_{i^*2}|+|\bar A_{j^*1}-\bar A_{j^*2}|+|\bar A_{i1}-\bar A_{i2}|}\\ 
&\geq \frac{(| A_{i^*1}- A_{i^*2}|+| A_{j^*1}- A_{j^*2}|)(V^*_{A_*}-\langle y^*,( A_{i1}, A_{i2})\rangle)}{| A_{i^*1}- A_{i^*2}|+| A_{j^*1}- A_{j^*2}|+| A_{i1}- A_{i2}|}-15\Delta'\\
&\geq \Delta_g-15\Delta'\\
&=19\Delta'-15\Delta'\\
&=4\Delta'>0     
\end{align*}

Hence, due to Lemma \ref{lem:n*2:unique}, $(x',y')$ is the unique Nash equilibrium of $\bar A$ and $\supp(x')=\{i^*,j^*\}$. This implies that $\tilde \Delta_{g}\geq  \Delta_{g}-15\Delta'=19\Delta'-15\Delta'=4\Delta'$. 

\end{proof}

\subsection{Technical Lemmas for upper bound}\label{tools:n*2}
In this section, we present few technical lemmas that are used to establish the upper bound on the sample complexity of finding the support of the Nash equilibrium in $n\times 2$ matrix games.

We first define matrices $A_1$, $A_2$, $A_3$ and $A_4$ as follows:
\[
A_1 = \begin{bmatrix} 
   a & b \\
   c & d \\
   e & f \\
    \end{bmatrix}
A_2 = \begin{bmatrix} 
   a+\Delta_{11} & b+\Delta_{12} \\
   c+\Delta_{21} & d+\Delta_{22} \\
   e+\Delta_{31} & f+\Delta_{32} \\
    \end{bmatrix}    
A_3 = \begin{bmatrix} 
   a & b \\
   c & d \\
    \end{bmatrix}
A_4 = \begin{bmatrix} 
   a+\Delta_{11} & b+\Delta_{12} \\
   c+\Delta_{21} & d+\Delta_{22} \\
    \end{bmatrix}    
\]
Let us assume that $|\Delta_{ij}|\leq \Delta$ for all $i,j$. Let us assume that $a>b,a>c,d>b,d>c$. Let us also assume that $e\neq f, e\neq a,e\neq c,f\neq b,f\neq d$. Let us also assume that $a+\Delta_{11}>b+\Delta_{12},a+\Delta_{11}>c+\Delta_{21},d+\Delta_{22}>b+\Delta_{12},d+\Delta_{22}>c+\Delta_{21}$. Let us also assume that no row of $A_1$ strongly dominates any other row. Now we present the following lemma that would be useful to compute the parameter $\Delta_g$.
\begin{lemma}\label{lem:n*2:delg}
Let $(x^*,y^*)$ be the Nash equilibrium of $A_3$. Then we have the following:
\begin{equation*}
    V^*_{A_3}-\langle y^*, (e,f)\rangle=\frac{(ad-bc)-(af-be)+(cf-de)}{a-b-c+d}
\end{equation*}
\end{lemma}
\begin{proof}
First observe that $V^*_{A_3}=\frac{ad-bc}{a-b-c+d}$ and $y^*=(\frac{d-b}{a-b-c+d},\frac{a-c}{a-b-c+d})$. Next we have the following:
\begin{align*}
     V^*_{A_3}-\langle y^*, (e,f)\rangle&=\frac{ad-bc}{a-b-c+d}-e\cdot \frac{d-b}{a-b-c+d}-f\cdot \frac{a-c}{a-b-c+d}\\
    &=\frac{(ad-bc)-(af-be)+(cf-de)}{a-b-c+d}
\end{align*}
\end{proof}
Next we present the following two propositions that would be useful to prove other technical lemmas in this section.
\begin{proposition}\label{prop:n*2:<3}
$\frac{|d-f|+|c-e|+|b-f|+|a-e|+|b-d|+|c-a|}{|a-b|+|c-d|+|e-f|}\leq 3$
\end{proposition}
\begin{proof}
W.l.o.g let us assume that $e<f$. As no row of $A_1$ strongly dominates any other row, we have $a>e$ and $b<f$. Now observe that $|a-c|+|b-d|=a-b-c+d=|a-b|+|c-d|$ and $|a-e|+|b-f|=a-b-e+f=|a-b|+|e-f|$. If $c>e$ and $d<f$, then $a>c>e$ and $b<d<f$. Hence, we have $|c-e|+|d-f|<|a-e|+|b-f|=|a-b|+|e-f|$. Similarly, if $e>c$ and $f<d$, then $a>e>c$ and $b<f<d$. Hence, we have $|c-e|+|d-f|<|a-c|+|b-d|=|a-b|+|c-d|$. Therefore, we have $\frac{|d-f|+|c-e|+|b-f|+|a-e|+|b-d|+|c-a|}{|a-b|+|c-d|+|e-f|}\leq 3$.
\end{proof}
\begin{proposition}\label{prop:n*2:<2}
$\frac{|(ad-bc)-(af-be)+(cf-de)|}{(|a-b|+|c-d|+|e-f|)^2}\leq 2$
\end{proposition}
\begin{proof}
W.l.o.g let us assume that $e<f$. As no row of $A_1$ strongly dominates any other row, we have $a>e$ and $b<f$. We have $(ad-bc)-(af-be)+(cf-de)=(d-b)(a-b-e-f)-(f-b)(a-b-c+d)$. Now observe that $|a-c|+|b-d|=|a-b|+|c-d|$ and $|a-e|+|b-f|=|a-b|+|e-f|$. This implies that $|d-b|\leq |a-b|+|c-d|+|e-f|$, $|f-b|\leq |a-b|+|c-d|+|e-f|$, $|a-b-c+d|\leq |a-b|+|c-d|+|e-f|$  and $|a-b-e+f|\leq |a-b|+|c-d|+|e-f|$. Hence, we have $\frac{|(ad-bc)-(af-be)+(cf-de)|}{(|a-b|+|c-d|+|e-f|)^2}\leq \frac{|d-b||a-b-e-f|+|f-b||a-b-c+d|}{(|a-b|+|c-d|+|e-f|)^2}\leq 2$
\end{proof}
\noindent
We present the following lemma that serves as a concentration inequality for the empirical estimate of $\Delta_g$.
\begin{lemma}\label{lem:n*2:deviation}
Let $(x^*,y^*)$ be the Nash equilibrium of $A_3$ and $(x',y')$ be the Nash equilibrium of $A_4$. Let $\Delta_{A_1}= \frac{(a-b-c+d)(V^*_{A_3}-\langle y^*, (e,f)\rangle)}{|a-b|+|c-d|+|e-f|}$ and $\Delta_{A_2}=  \frac{(a'-b'-c'+d')(V^*_{A_4}-\langle y', (e',f')\rangle)}{|a'-b'|+|c'-d'|+|e'-f'|}$ where $a'=a+\Delta_{11},b'=b+\Delta_{12},c'=c+\Delta_{21},d'=d+\Delta_{22},e'=e+\Delta_{31}$ and $f'=f+\Delta_{32}$. Then we have the following:
\begin{equation*}
    \Delta_{A_2}\geq \Delta_{A_1}-15\Delta
\end{equation*}
Moreover, if $\frac{6\Delta}{|a-b|+|c-d|+|e-f|}\leq \frac{1}{4}$ and $\Delta_{A_1}\leq 0$, then we have the following:
\begin{equation*}
    \Delta_{A_2}\leq \Delta_{A_1}+4\Delta
\end{equation*}
\end{lemma}
\begin{proof}
Due to Lemma \ref{lem:n*2:delg}, we have $\Delta_{A_1}= \frac{(ad-bc)-(af-be)+(cf-de)}{|a-b|+|c-d|+|e-f|}$ and $\Delta_{A_2}=  \frac{(a'd'-b'c')-(a'f'-b'e')+(c'f'-d'e')}{|a'-b'|+|c'-d'|+|e'-f'|}$ where $a'=a+\Delta_{11},b'=b+\Delta_{12},c'=c+\Delta_{21},d'=d+\Delta_{22},e'=e+\Delta_{31}$ and $f'=f+\Delta_{32}$. Let $N=(ad-bc)-(af-be)+(cf-de)$ and $M=|a-b|+|c-d|+|e-f|$. Let $N'=(a'd'-b'c')-(a'f'-b'e')+(c'f'-d'e')$ and $M'=|a'-b'|+|c'-d'|+|e'-f'|$. Observe that $N'=N+(d-f)\Delta_{11}+(e-c)\Delta_{12}+(f-b)\Delta_{21}+(a-e)\Delta_{22}+(b-d)\Delta_{31}+(c-a)\Delta_{32}$. Now we have the following:
\begin{align*}
    \frac{N'}{M'}&\geq \frac{N-(|d-f|+|e-c|+|f-b|+|a-e|+|b-d|+|c-a|)\Delta}{M(1+\frac{6\Delta}{M})}\\
    &\geq \left(\frac{N}{M}-3\Delta\right)\left(1+\frac{6\Delta}{M}\right)^{-1} \tag{due to Proposition \ref{prop:n*2:<3}}\\
    &= \left(\frac{N}{M}-3\Delta\right)\left(1-\frac{\frac{6\Delta}{M}}{1+\frac{6\Delta}{M}}\right) \tag{as $(1+x)^{-1}=1-\frac{x}{1+x}$}\\
    & = \frac{N}{M}-3\Delta-\frac{N}{M}\cdot \frac{\frac{6\Delta}{M}}{1+\frac{6\Delta}{M}}+3\Delta\cdot\frac{\frac{6\Delta}{M}}{1+\frac{6\Delta}{M}}\\
    & \geq \frac{N}{M}-3\Delta-\frac{|N|}{M^2}\cdot 6\Delta\tag{as $\frac{6\Delta}{M}\geq 0$} \\
    & \geq \frac{N}{M}-3\Delta-2\cdot 6\Delta  \tag{due to Proposition \ref{prop:n*2:<2}}\\
    & = \frac{N}{M}-15\Delta
\end{align*}
If $\frac{N}{M}\leq 0$, then we have the following:
\begin{align*}
    \frac{N'}{M'}&\leq \frac{N+(|d-f|+|e-c|+|f-b|+|a-e|+|b-d|+|c-a|)\Delta}{M(1-\frac{6\Delta}{M})}\\
    &\leq \left(\frac{N}{M}+3\Delta\right)\left(1-\frac{6\Delta}{M}\right)^{-1} \tag{due to Proposition \ref{prop:n*2:<3}}\\
    &= \left(\frac{N}{M}+3\Delta\right)\left(1+\frac{\frac{6\Delta}{M}}{1-\frac{6\Delta}{M}}\right) \tag{as $(1-x)^{-1}=1+\frac{x}{1-x}$}\\
    & = \frac{N}{M}+3\Delta+\frac{N}{M}\cdot \frac{\frac{6\Delta}{M}}{1-\frac{6\Delta}{M}}+3\Delta\cdot\frac{\frac{6\Delta}{M}}{1-\frac{6\Delta}{M}}\\
    & \leq \frac{N}{M}+3\Delta+3\Delta\cdot\frac{1}{3} \tag{as $\frac{6\Delta}{M}\leq \frac{1}{4}$ and $\frac{N}{M}\leq 0$}\\
    & = \frac{N}{M}+4\Delta
\end{align*}
\end{proof}
Next we define matrix $B$ as follows:
\[
B = \begin{bmatrix} 
   a & b \\
   c & d \\
   e & f \\
   g & h
    \end{bmatrix}
\]
Let us assume that $a>b,d>c,e>f,h>g$. Let us also assume that $a>c$ and $d>c$. Now we present the following lemma where we establish important properties of the optimal rows (rows that are in the support of the Nash equilibrium). 
\begin{lemma}\label{lem:n*2:suboptimal}
Let $(x_1,y_1)$ be the Nash equilibrium of $B_1=[a,b;g,h]$, $(x_2,y_2)$ be the Nash equilibrium of $B_2=[e,f;c,d]$ and $(x_3,y_3)$ be the Nash equilibrium of $B_3=[e,f;g,h]$. Let us assume that $B_1,B_2$ and $B_3$ have unique Nash Equilibria which are not PSNE. If $B$ has a unique Nash equilibrium $(x^*,y^*)$ such that $|\supp(x^*)|=|\supp(y^*)|=\{1,2\}$, then we have the following:
\begin{itemize}
    \item $V_{B_1}^*-\langle y_1,(c,d) \rangle <0$
    \item $V_{B_2}^*-\langle y_2,(a,b) \rangle <0$
    \item $V_{B_3}^*-\max\{\langle y_3,(a,b) \rangle,\langle y_3,(c,d) \rangle \}<0$
\end{itemize}
\end{lemma}
\begin{proof}
Due to Lemma \ref{lem:n*2:delg}, we have $V_{B_1}^*-\langle y_1,(c,d)\rangle=\frac{(ah-bg)-(ad-bc)+(gd-hc)}{a-b-g+h}$. Due to Lemma \ref{lem:n*2:delg} and the fact that $\supp(x^*)=\{1,2\}$, we have $V_B^*-\langle y^*,(g,h)\rangle = \frac{(ad-bc)-(ah-bg)+(ch-dg)}{a-b-c+d} > 0$. Hence, we have $V_{B_1}^*-\langle y_1,(c,d)\rangle<0$.

Due to Lemma \ref{lem:n*2:delg}, we have $V_{B_2}^*-\langle y_2,(a,b)\rangle=\frac{(ed-fc)-(eb-fa)+(cb-da)}{e-f-c+d}$. Due to Lemma \ref{lem:n*2:delg} and the fact that $\supp(x^*)=\{1,2\}$, we have $V_B^*-\langle y^*,(e,f)\rangle = \frac{(ad-bc)-(af-be)+(cf-de)}{a-b-c+d} > 0$. Hence, we have $V_{B_2}^*-\langle y_2,(a,b)\rangle<0$.

 Let $y^*=(y_1^*,y_2^*)$ and $y_3=(y_{3,1},y_{3,2})$. As $\supp(x^*)=\{1,2\}$, we have $\langle y^*,(e,f) \rangle<V_B^*$ and $\langle y^*,(g,h)\rangle<V_B^*$. If $y_{3,1}\leq y_1^*$, we have $V_{B_3}^*=\langle y_{3,1},(e,f)\rangle\leq \langle y^*,(e,f) \rangle<V_B^*\leq\langle y_{3,1},(c,d)\rangle$. Similarly, if $y_{3,1}> y_1^*$, we have $V_{B_3}^*=\langle y_{3,1},(g,h)\rangle< \langle y^*,(g,h) \rangle<V_B^*<\langle y_{3,1},(a,b)\rangle$. Hence, $V_{B_3}^*-\max\{\langle y_3,(a,b) \rangle,\langle y_3,(c,d) \rangle \}<0$.
\end{proof}
Let $A$ be a $n\times 2$ matrix with no PSNE. Let us assume that $\forall i\in[n]$, $A_{i1}\neq A_{i2}$. Now we present the following lemma that relates the nash equilibrium of a submatrix of $A$ to the nash equilibrium of $A$. 
\begin{lemma}\label{lem:n*2:unique}
Consider two distinct row indices $i_1,i_2$ such that $A_{i_11}>A_{i_12},A_{i_21}<A_{i_22}$. Let $(x^*,y^*)$ be the Nash equilibrium of $C=[A_{i_11},A_{i_12};A_{i_21},A_{i_22}]$. If $A_{i_11}>A_{i_21},\; A_{i_12}<A_{i_22} $ and $V_C^*-\langle y^*,(A_{j1},A_{j2})\rangle>0$ for all $j\in [n]\setminus \{i_1,i_2\}$, then $A$ has a unique Nash equilibrium $(x',y')$ such that $\supp(x')=\{i_1,i_2\}$.
\end{lemma}
\begin{proof}
Let $(x,y)$ be a Nash equilibrium of $A$. Let $x=(x_1,x_2,\ldots,x_n)$, $y=(y_1,y_2)$ and $y^*=(y^*_1,y^*_2)$. 

If $y_1<y_1^*$, then $\langle y, (A_{i_21},A_{i_22})\rangle>V_C^*$. Consider $j\in \supp(x)$ such that $A_{j1}>A_{j2}$. In this case, we have $\langle y,(A_{j1},A_{j2})\rangle < \langle y^*,(A_{j1},A_{j2})\rangle \leq V_C^*< \langle y, (A_{i_21},A_{i_22})\rangle$. This contradicts the fact that $j\in \supp(x)$. Hence $y_1\geq y^*_1$.

If $y_1>y_1^*$, then $\langle y, (A_{i_11},A_{i_12})\rangle>V_C^*$. Consider $j\in \supp(x)$ such that $A_{j1}<A_{j2}$. In this case, we have $\langle y,(A_{j1},A_{j2})\rangle < \langle y^*,(A_{j1},A_{j2})\rangle \leq V_C^*< \langle y, (A_{i_11},A_{i_12})\rangle$. This contradicts the fact that $j\in \supp(x)$. Hence $y_1\leq y^*_1$. 

As $y_1\geq y^*_1$ and $y_1\leq y^*_1$, we have $y_1=y_1^*$. This implies that $V_A^*=\max_{j\in[n]}\langle y^*,(A_{j1},A_{j2})\rangle= V_C^*$. 

As $V_A^*-\langle y^*,(A_{j1},A_{j2})\rangle>0$ for all $j\in [n]\setminus \{i_1,i_2\}$ and $V_A^*=\langle y^*,(A_{j1},A_{j2})\rangle$ for all $j\in\{i_1,i_2\}$, we have $\supp(x)=\{i_1,i_2\}$. Now observe that $C$ has a unique Nash Equilibirum which is not a PSNE and $((x_{i_1},x_{i_2}),y)$ is also a Nash Equilibrium of $C$. This implies that $(x,y)$ is the unique Nash equilibrium of $A$ such that $\supp(x)=\{i_1,i_2\}$ and $((x_{i_1},x_{i_2}),y)=(x^*,y^*)$.

\end{proof}
\section{Proof of Lower Bound with respect to $\Delta_g$}\label{appendix:thm6}
Before finishing the proof of the theorem \ref{thm:lower4}, we begin with the proof of Lemma~\ref{low4:lem1}
\begin{proof}
Let $c'=c-\Delta$, $d'=d-\Delta$, $e'=e+\Delta$ and $f'=f+\Delta$. Let $V^*=\frac{ad'-bc'}{D_1}$. It can be shown that $V_{A_0}^*=V^*+\frac{(a-b)\Delta}{D_1}$, $V_{A_\Delta}^*=V^*$ and $V_{A_{2\Delta}}^*\geq V^*+\frac{(a-b)\Delta}{D_2}$. For any $\alpha,\beta\in[0,1]$ such that $\alpha+\beta\leq 1$ and $\gamma\in[\frac{b-d'}{D_1},\frac{a-c'}{D_1}]$, let $x'=(1-\alpha-\beta,\alpha,\beta)$ and $y'=(\frac{d'-b}{D_1}+\gamma,\frac{a-c'}{D_1}-\gamma)$. Note that this parameterization ensures the range of $x'$ is equal to $\simplex_3$ and the range of $y'$ is equal to $\simplex_2$. Let $k=(1-\alpha-\beta)(a-b)\gamma+\alpha(c'-d')\gamma+\beta(e'-f')\gamma$. It can be shown that $\langle x',A_{\Delta}y'\rangle=V^*+k$. It can also be shown that $\langle x',A_0 y'\rangle=V^*+k+(\alpha-\beta)\Delta$ and $\langle x',A_{2\Delta} y'\rangle=V^*+k-(\alpha-\beta)\Delta$. We refer the reader to the Appendix \ref{apendix:low4:lem1} for the detailed calculations.

If $|k|>\lambda/4$, then $|V_{A_{\Delta}}^*-\langle x', A_{\Delta}y' \rangle|=|k|> \lambda/4$. If $|k| \leq \lambda/4$ and $(\alpha-\beta)\Delta\geq 0$, then $V_{A_{2\Delta}}^*-\langle x',A_{2\Delta} y'\rangle\geq\frac{(a-b)\Delta}{D_2}-k+(\alpha-\beta)\Delta\geq \frac{3\lambda}{4}$. Similarly, If $|k| \leq \lambda/4$ and $(\alpha-\beta)\Delta< 0$, then $V_{A_{0}}^*-\langle x',A_{0} y'\rangle=\frac{(a-b)\Delta}{D_1}-k-(\alpha-\beta)\Delta\geq \frac{3\lambda}{4}$. Hence, we proved that there exists a matrix $B\in\{A_{0},A_{\Delta},A_{2\Delta}\}$ such that the following holds:
\begin{equation*}
    |V_B^*-\langle x', By' \rangle|> \frac{\lambda}{4}>\varepsilon
\end{equation*}
\end{proof}

\subsection{Proof of Theorem \ref{thm:lower4}}
Let $\nu_{i,j}^A = \mathcal{N}(A_{ij},1)$ be the distribution of an observation when playing pair $(i,j)$ with matrix $A$. 
Let $\P_A$ denote the probability law of the internal randomness of the algorithm and random observations.
If an algorithm is $(\varepsilon,\delta)$-PAC-good and outputs a solution $(\widehat{x},\widehat{y})$ then $\min_A \P_A(|V_A^*-\langle \widehat{x}, A \widehat{y} \rangle|\leq \varepsilon) \geq 1-\delta$.
We will show that if an algorithm is $(\varepsilon,\delta)$-PAC-good then it can also accomplish a particular hypothesis.
We will conclude by noting that any procedure that can accomplish the hypothesis test must take the claimed sample complexity. 

For any pair of mixed strategies $(\widehat{x},\widehat{y})$ output by the procedure at the stopping time $\tau$, define 
\begin{align*}
    \phi = \{A_{0},A_{\Delta},A_{2\Delta}\} \setminus \arg\max_{B \in \{A_{0},A_{\Delta},A_{2\Delta}\}} |V_{B}^*-\langle \widehat{x}, B \widehat{y} \rangle|,
\end{align*}
breaking ties arbitrarily in the maximum so that $\phi \in  \{A_{2\Delta},A_{0}\}\cup \{A_{2\Delta},A_{\Delta}\}\cup \{A_{0},A_{\Delta}\}$.
Note that
\begin{align}\label{eqn:correctness_delta:4}
    \P_{A_0}( A_0 \in \phi ) \geq \P_{A_0}( A_0 \in \phi,  |V_{A_0}^*-\langle \widehat{x},A_0 \widehat{y} \rangle| \leq \varepsilon) = \P_{A_0}(   |V_{A_0}^*-\langle \widehat{x}, A_0 \widehat{y} \rangle| \leq \varepsilon) \geq 1-\delta
\end{align}
where the equality follows from the Lemma~\ref{low4:lem1}: at least one of the three matrices must have a loss of more than $\varepsilon $, but $A_0$ has a loss of at most $\varepsilon$, thus $A_0 \in \phi$.
Now because
\begin{align*}
    2 \max\{ \P_{A_0}( \phi = \{A_0,A_{2\Delta}\} ) , \P_{A_0}( \phi = \{A_0,A_{\Delta}\} ) \} &\geq \P_{A_0}( \phi = \{A_0,A_{2\Delta}\} ) + \P_{A_0}( \phi = \{A_0,A_{\Delta}\} ) \\
    &= \P_{A_0}( A_0 \in \phi ) \geq 1-\delta
\end{align*}
we have that $\P_{A_0}( \phi = \{A_0,A_{\Delta}\} ) \geq \frac{1-\delta}{2}$ or $\P_{A_0}( \phi = \{A_0,A_{2\Delta} \}) \geq \frac{1-\delta}{2}$.
Let's assume the former (the latter case is handled identically).
By the same argument as \eqref{eqn:correctness_delta:4} we have that $\P_{A_{2\Delta}}( \phi = \{A_0,A_{\Delta}\} ) \leq \delta$.

For a stopping time $\tau$, let $N_{i,j}(\tau)$ denote the number of times $(i,j)$ is sampled. Recalling that $\nu_{i,j}^A = \mathcal{N}(A_{ij},1)$, we have by Lemma 1 of \citet{kaufmann2016complexity} that
\begin{align*}
\sum_{i=2}^3\sum_{j=1}^2\mathbb{E}_{A_0}[ N_{i,j}(\tau) ] KL( \nu_{i,j}^{A_0}, \nu_{i,j}^{A_{2\Delta}} ) \geq d( \P_{A_0}( \phi = \{A_0,A_{\Delta}\} ), \P_{A_{2\Delta}}( \phi = \{A_0,A_{\Delta}\} ) )
\end{align*}
where $KL( \nu_{2,1}^{A_0}, \nu_{1,1}^{A_{2\Delta}} ) = KL( \nu_{2,2}^{A_0}, \nu_{2,2}^{A_{2\Delta}} )= KL( \nu_{3,1}^{A_0}, \nu_{3,1}^{A_{2\Delta}} )= KL( \nu_{3,2}^{A_0}, \nu_{3,2}^{A_{2\Delta}} ) = 2\Delta^2$ and $d(p,q) = p \log(\frac{p}{q}) + (1-p) \log(\frac{1-p}{1-q})$.
Since 
\begin{align*}
    d( \P_{A_0}( \phi = \{A_0,A_{\Delta}\} ), \P_{A_{2\Delta}}( \phi = \{A_0,A_{\Delta}\} ) ) &\geq d( \tfrac{1-\delta}{2},\delta) \\
    &= \tfrac{1-\delta}{2} \log( \tfrac{1-\delta}{2 \delta} ) + \tfrac{1+\delta}{2} \log( \tfrac{1+\delta}{2(1-\delta)} ) \\
    &= \tfrac{1}{2} \log( \tfrac{1+\delta}{4 \delta} ) -  \tfrac{\delta}{2} \log( \tfrac{(1-\delta)^2}{ \delta (1+\delta)} ) \\
    &\geq \tfrac{1}{2} \log( \tfrac{1+\delta}{4 \delta} ) -  1/8 > \tfrac{1}{2} \log(1/30\delta)
\end{align*} 
and $\tau = \sum_{i=1}^3\sum_{j=1}^2N_{i,j}(\tau)$ we conclude that
\begin{align*}
    \mathbb{E}_{A_0}[ \tau ] \geq \frac{ \log(1/30 \delta) }{4\Delta^2} > \frac{ \log(1/30 \delta) }{4\Delta_{g}^2}
\end{align*}
as claimed. We get the last inequality as $0<\Delta<\Delta_g$ (see Appendix \ref{apendix:low4:lem1} for more details).

\subsection{Calculations for Lemma \ref{low4:lem1}}\label{apendix:low4:lem1}
Recall that $\Delta:=\frac{(d-b)D_2-(f-b)D_1}{D_1+D_2}$. Let $c'=c-\Delta$, $d'=d-\Delta$, $e'=e+\Delta$ and $f'=f+\Delta$. Observe that $D_1=a-b-c+d=a-b-c'+d'$ and $D_2=a-b-e+f=a-b-e'+f'$. Let $V^*=\frac{ad'-bc'}{D_1}$. 

First, observe that $\square=\frac{(d-b)D_2-(f-b)D_1}{D_1+D_2}$ satisfies the equality $\frac{d-b-\square}{D_1}=\frac{f-b+\square}{D_2}$. Next, observe that 
$0<\frac{d-b-\Delta}{D_1}<1$ as $\frac{d-b-\Delta}{D_1}=\frac{(d-b)+(f-b)}{D_1+D_2}$.



Hence, we have $(\frac{d'-b}{D_1},\frac{a-c'}{D_1})=(\frac{f'-b}{D_2},\frac{a-e'}{D_2})$ and $(\frac{d'-b}{D_1},\frac{a-c'}{D_1})\in \simplex_2$. Therefore, we have $V_{A_\Delta}^*=\frac{d'-b}{D_1}\cdot a+ \frac{a-c'}{D_1}\cdot b =\frac{ad'-bc'}{D_1}=V^*$. Next, observe that $V_{A_0}^*=\frac{d'-b+\Delta}{D_1}\cdot a+\frac{a-c'-\Delta}{D_1}\cdot b=\frac{ad'-bc'}{D_1}+\frac{(a-b)\Delta}{D_1}=V^*+\frac{(a-b)\Delta}{D_1}$. 

If $a-e'>\Delta$, then $V_{A_{2\Delta}}^*=\frac{f'-b+\Delta}{D_2}\cdot a+\frac{a-e'-\Delta}{D_2}\cdot b=\frac{ad'-bc'}{D_1}+\frac{(a-b)\Delta}{D_2}=V^*+\frac{(a-b)\Delta}{D_2}$. We get $\frac{f'-b}{D_2}\cdot a+\frac{a-e'}{D_2}\cdot b=\frac{ad'-bc'}{D_1}$ as $(\frac{d'-b}{D_1},\frac{a-c'}{D_1})=(\frac{f'-b}{D_2},\frac{a-e'}{D_2})$. Note that the second row is sub-optimal for $A_{2\Delta}$ as $\frac{f'-b+\Delta}{D_2}\cdot (c'-\Delta)+\frac{a-e'-\Delta}{D_2}\cdot (d'-\Delta)<V^*$.

If $a-e'\leq \Delta$, then the element in the third row and the first column of $A_{2\Delta}$ is a PSNE. Hence, $V_{A_{2\Delta}}^*=e'+\Delta$. Let $g=e'+\Delta-a$. Observe that $a-b<D_2$. Now we have $V^*+\frac{(a-b)\Delta}{D_2}= \frac{f'-b+\Delta}{D_2}\cdot a+\frac{a-e'-\Delta}{D_2}\cdot b=(1+\frac{g}{D_2})a-\frac{gb}{D_2}=a+\frac{(a-b)g}{D_2}\leq a+g=e'+\Delta$. Hence $V_{A_{2\Delta}}^*\geq V^*+\frac{(a-b)\Delta}{D_2}$.

Now we define a class of matrices $B_\square$ as follows:
\[
B_\square = \begin{bmatrix} 
   a & b \\
   c'+\square & d'+\square \\
   e'-\square & f'-\square \\
    \end{bmatrix}\]
Now observe that $B_{\Delta}=A_{0}$, $B_0=A_{\Delta}$ and $B_{-\Delta}=A_{2\Delta}$.    

Recall that $x'=(1-\alpha-\beta,\alpha,\beta)$ and $y'=(\frac{d'-b}{D_1}+\gamma,\frac{a-c'}{D_1}+\gamma)$. Let $V_1=\langle y',(a,b) \rangle$, $V_2=\langle y',(c'+\square , d'+\square) \rangle$ and $V_3=\langle y',(e'-\square , f'-\square) \rangle$. First, we have $V_1= (\frac{d'-b}{D_1}+\gamma)\cdot a+ (\frac{a-c'}{D_1}-\gamma)\cdot b= \frac{ad'-bc'}{D_1}+(a-b)\gamma= V^*+(a-b)\gamma$. Next, we have $V_2=(\frac{d'-b}{D_1}+\gamma)\cdot (c'+\square)+ (\frac{a-c'}{D_1}-\gamma)\cdot (d'+\square)=\frac{ad'-bc'}{D_1}+(c'-d')\gamma+(\frac{d'-b}{D_1}+\gamma+\frac{a-c'}{D_1}-\gamma)\square=V^*+\square+(c'-d')\gamma$. Similarly, we have $V_3=(\frac{d'-b}{D_1}+\gamma)\cdot (e'-\square)+ (\frac{a-c'}{D_1}-\gamma)\cdot (f'-\square)=\frac{ad'-bc'}{D_1}+(e'-f')\gamma-(\frac{d'-b}{D_1}+\gamma+\frac{a-c'}{D_1}-\gamma)\square=V^*-\square+(e'-f')\gamma$.

Let $k=(1-\alpha-\beta)(a-b)\gamma+\alpha(c'-d')\gamma+\beta(e'-f')\gamma$. Now observe that $\langle x',B_\square y'\rangle=\langle x' ,(V_1,V_2,V_3)\rangle=V^*+k+(\alpha-\beta)\square$. 

Now we present the following proposition.
\begin{proposition}
$\Delta<\Delta_g$
\end{proposition}
\begin{proof}
Observe that $\Delta=\frac{(d-b)D_2-(f-b)D_1}{D_1+D_2}=\frac{(ad-bc)-(af-be)+(cf-de)}{D_1+D_2}$. Due to Lemma \ref{lem:n*2:delg}, we have $\Delta_g=\frac{(ad-bc)-(af-be)+(cf-de)}{|a-b|+|c-d|+|e-f|}>0$. Hence, $0<\Delta<\Delta_g$. Note that $\Delta_g>0$ as the third row of $A$ is sub-optimal.
\end{proof}